\documentclass[aps,pra,twocolumn,nofootinbib, superscriptaddress,10pt]{revtex4-2}
\usepackage{graphicx} 
\usepackage{subfigure}
\usepackage{amsmath,amsfonts,amssymb,amsthm,bm}
\usepackage{bbm}
\usepackage[colorlinks=true,citecolor=blue,linkcolor=blue]{hyperref}
\usepackage{url}
\usepackage{lmodern}
\usepackage{mathrsfs}
\usepackage{chngcntr}
\usepackage{appendix}
\usepackage[dvipsnames]{xcolor}

\usepackage{csquotes}
\MakeOuterQuote{"}
\usepackage{mathtools,mleftright}
\mleftright
\usepackage{multirow}
\usepackage{makecell}
\newcommand{\lsp}{\hspace{0.1em}}
\newcommand{\equad}{\,\hphantom{=}\,}

\newcommand{\romanNum}[1]{\uppercase\expandafter{\romannumeral#1}}
\def\identity{\leavevmode\hbox{\small1\kern-3.8pt\normalsize1}}
\newtheorem{theorem}{Theorem}
\newtheorem{lemma}{Lemma}

\newtheorem{proposition}{Proposition}

\theoremstyle{plain}

\theoremstyle{remark}

\newcommand{\bbone}{\mathbbm{1}}
\newcommand{\Cl}{\mathrm{Cl}}
\newcommand{\MC}{\mathrm{MC}}

\newcommand{\CRM}{\mathrm{CRM}}
\newcommand{\lc}{\mathrm{LC}}
\newcommand{\THR}{\mathrm{THR}}

\newcommand{\haar}{\mathrm{Haar}}

\newcommand{\GHZ}{\mathrm{GHZ}}

\newcommand{\swap}{\mathrm{SWAP}}

\newcommand{\hrho}{\hat{\rho}}
\newcommand{\hO}{\hat{O}}
\newcommand{\hP}{\hat{P}}

\newcommand{\bepsilon}{\bar{\epsilon}}

\newcommand{\rme}{\operatorname{e}}
\newcommand{\rmi}{\mathrm{i}}

\newcommand{\rmC}{\mathrm{C}}

\newcommand{\rmH}{\mathrm{H}}

\newcommand{\rmM}{\mathrm{M}}

\newcommand{\rmT}{\mathrm{T}}
\newcommand{\rmU}{\mathrm{U}}

\newcommand{\SWAP}{\mathrm{SWAP}}

\newcommand{\bmv}{\bm{v}}

\newcommand{\bfp}{\mathbf{p}}

\newcommand{\bfh}{\mathbf{h}}

\newcommand{\bfs}{\mathbf{s}}
\newcommand{\bft}{\mathbf{t}}

\newcommand{\caD}{\mathcal{D}}
\newcommand{\caE}{\mathcal{E}}
\newcommand{\caG}{\mathcal{G}}
\newcommand{\caH}{\mathcal{H}}
\newcommand{\caL}{\mathcal{L}}

\newcommand{\caM}{\mathcal{M}}

\newcommand{\caO}{\mathcal{O}}
\newcommand{\caP}{\mathcal{P}}

\newcommand{\caS}{\mathcal{S}}

\newcommand{\caU}{\mathcal{U}}

\newcommand{\caW}{\mathcal{W}}
\newcommand{\bcaP}{\overline{\caP}}

\newcommand{\bbE}{\mathbb{E}}

\newcommand{\bbR}{\mathbb{R}}

\newcommand{\bbV}{\mathbb{V}}

\newcommand{\scrD}{\mathscr{D}}

\newcommand{\scrN}{\mathscr{N}}
\newcommand{\scrP}{\mathscr{P}}
\newcommand{\tXi}{\tilde{\Xi}}

\newcommand{\rank}{\operatorname{rank}}

\newcommand{\tr}{\operatorname{Tr}}
\newcommand{\Tensor}[2]{#1^{\otimes #2}}
\newcommand{\dagtensor}[2]{#1^{\dagger \otimes #2}}

\newcommand{\lref}[1]{Lemma~\ref{#1}}
\newcommand{\lsref}[1]{Lemmas~\ref{#1}}
\newcommand{\Lref}[1]{Lemma~\ref{#1}}
\newcommand{\Lsref}[2]{Lemmas~\ref{#1} and \ref{#2}}
\newcommand{\thref}[1]{Theorem~\ref{#1}}
\newcommand{\thsref}[2]{Theorems~\ref{#1} and~\ref{#2}}
\newcommand{\thssref}[3]{Theorems~\ref{#1}, \ref{#2}, and \ref{#3}}
\newcommand{\Thref}[1]{Theorem~\ref{#1}}
\newcommand{\Thsref}[2]{Theorems~\ref{#1} and \ref{#2}}
\newcommand{\pref}[1]{Proposition~\ref{#1}}
\newcommand{\psref}[1]{Propositions~\ref{#1}}
\newcommand{\Pref}[1]{Proposition~\ref{#1}}
\newcommand{\Psref}[1]{Propositions~\ref{#1}}

\def\eqref#1{\textup{(\ref{#1})}}
\newcommand{\eref}[1]{Eq.~\textup{(\ref{#1})}}
\newcommand{\eqsref}[2]{Eqs.~(\ref{#1}) and (\ref{#2})}
\newcommand{\eqssref}[3]{Eqs.~(\ref{#1}), (\ref{#2}), and (\ref{#3})}

\newcommand{\Eref}[1]{Equation~\textup{(\ref{#1})}}
\newcommand{\Eqsref}[2]{Equations~(\ref{#1}) and (\ref{#2})}

\newcommand{\sref}[1]{Sec.~\ref{#1}}
\newcommand{\ssref}[1]{Secs.~\ref{#1}}
\newcommand{\fref}[1]{Fig.~\ref{#1}}
\newcommand{\Fref}[1]{Figure~\ref{#1}}
\newcommand{\fsref}[2]{Figs.~\ref{#1} and \ref{#2}}
\newcommand{\Fsref}[1]{Figures~\ref{#1}}

\def\<{\langle}  
\def\>{\rangle}
\newcommand{\rcite}[1]{Ref.~\cite{#1}}
\newcommand{\rscite}[1]{Refs.~\cite{#1}}
\begin{document}

\title{High-Precision Fidelity Estimation with Common Randomized Measurements}
\author{Zhongyi Yang}
\affiliation{International Center for Quantum Materials, School of Physics, Peking University, Beijing 100871, China}
\author{Datong Chen}
\affiliation{State Key Laboratory of Surface Physics, Department of Physics, and Center for Field Theory and Particle Physics, Fudan University, Shanghai 200433, China}
\affiliation{Institute for Nanoelectronic Devices and Quantum Computing, Fudan University, Shanghai 200433, China}
\affiliation{Shanghai Research Center for Quantum Sciences, Shanghai 201315, China}

\author{Zihao Li}
\affiliation{State Key Laboratory of Surface Physics, Department of Physics, and Center for Field Theory and Particle Physics, Fudan University, Shanghai 200433, China}
\affiliation{Institute for Nanoelectronic Devices and Quantum Computing, Fudan University, Shanghai 200433, China}
\affiliation{Shanghai Research Center for Quantum Sciences, Shanghai 201315, China}
\affiliation{QICI Quantum Information and Computation Initiative, School of Computing and Data Science, The University of Hong Kong, Pokfulam Road, Hong Kong, China}

\author{Huangjun Zhu}
\email{zhuhuangjun@fudan.edu.cn}
\affiliation{State Key Laboratory of Surface Physics, Department of Physics, and Center for Field Theory and Particle Physics, Fudan University, Shanghai 200433, China}
\affiliation{Institute for Nanoelectronic Devices and Quantum Computing, Fudan University, Shanghai 200433, China}
\affiliation{Shanghai Research Center for Quantum Sciences, Shanghai 201315, China}

\date{\today}

\begin{abstract}
Efficient fidelity estimation of multiqubit quantum states is crucial to many  applications in quantum information processing. However, to estimate the infidelity $\epsilon$ with multiplicative precision, conventional estimation protocols  require (order) $1/\epsilon^2$ different circuits in addition to  $1/\epsilon^2$ samples, which is quite resource-intensive for high-precision fidelity estimation. Here we introduce an efficient estimation protocol by virtue of common randomized measurements (CRM)  integrated with shadow estimation based on the Clifford group, which only requires $1/\epsilon $ circuits. Moreover, in many scenarios of practical interest, in the presence of depolarizing or Pauli noise for example, our protocol only requires a constant number of circuits,  irrespective of the infidelity $\epsilon$ and the qubit number. For large and intermediate quantum systems, quite often one circuit is already sufficient. In the course of study,  we clarify the performance of CRM shadow estimation based on the Clifford group and 4-designs and highlight its advantages over standard  and thrifty shadow estimation. 
\end{abstract}

\maketitle

\emph{Introduction}---The ability to efficiently learn the properties of unknown quantum states is pivotal  for advancing quantum technologies  
\cite{Eisert2020Cert,Kliesch2021theo,Andreas2023ran,gebhart2023learning}. Fidelity estimation, as an archetypal learning task, plays a key role in quantum information processing, with great efforts dedicated to developing efficient estimation protocols based on simple measurements \cite{hayashi2006locc,Guhne2007Tool,Flammia2011Direct,Silva2011prac,pallister2018opt,Zhu2019opt,zhu2019general,Li2019eff,li2020opt,Cerezo2020var,Zhang2021direct,Yu2021stat,wang2022quantum,sun2025efficient,fang2025Optimal,zhang2020exp,jiang2020towards,Qin2024exp}. Notably, direct fidelity estimation (DFE) using Pauli measurements is effective 
for  important classes of quantum states like stabilizer states  \cite{Flammia2011Direct,Silva2011prac}.
A major breakthrough came with shadow estimation based on Clifford measurements~\cite{huang2020pred},  which can achieve a sample complexity independent of the qubit number~$n$. Similar results can also be generalized to qudit systems with odd prime local dimensions \cite{Mao2025qudit}. Subsequent refinements, including thrifty shadow estimation \cite{struchalin2021exp,Huggins2022Unb,helsen2023thrifty,zhou2023perf,chen2024nonstab} and common randomized measurements (CRM) \cite{garratt2023prob,vermersch2024enhanced,Vitale2024rob}, have further advanced the field.

Despite the progress mentioned above, to estimate the infidelity $\epsilon$ with multiplicative precision, existing protocols, including DFE \cite{Flammia2011Direct,Silva2011prac} and shadow estimation \cite{huang2020pred,helsen2023thrifty,zhou2023perf,chen2024nonstab},  usually require  $\caO\left(1/\epsilon^2\right)$ distinct measurement settings and samples. This overhead poses a significant challenge for high-precision fidelity estimation (HPFE), as adjusting experimental settings is resource-intensive for platforms like superconducting qubits and trapped ions. Overcoming this bottleneck is critical for practical applications in quantum information processing.

In this paper, we introduce an efficient
fidelity estimation protocol that integrates CRM with Clifford-based shadow
estimation. By deriving rigorous upper bounds on the variance in terms of the deviation between the system state and the prior state, we establish a robust performance guarantee for CRM shadow estimation. Our protocol achieves a significant resource reduction, requiring only $\caO(1/\epsilon)$ circuits to estimate the infidelity with multiplicative precision when the target state is chosen as the prior. Remarkably, in practical scenarios involving depolarizing,  Pauli, or certain coherent noise, our protocol requires only a constant number of circuits, irrespective of  the infidelity $\epsilon$ and the qubit number. Additionally, the nonstabilizerness of the target state can further enhance the efficiency.  For large and intermediate quantum systems, quite often we can even work with only one circuit. Compared with standard and thrifty shadow estimation, our protocol can reduce the circuit sample costs  thousands and even millions of times for HPFE with $\epsilon\leq 0.001$, making it particularly appealing for practical applications.

As benchmarks, we elucidate the performance of the alternative protocol
based on 4-designs and offer additional insights on Pauli measurements studied before \cite{vermersch2024enhanced}. Like in standard shadow estimation, 
Clifford and 4-design measurements may exhibit exponential advantages over  Pauli measurements with respect to the qubit number.
What is surprising is that Clifford measurements are not only simpler but also more efficient than  4-design measurements in HPFE for various practical situations. Our work represents a significant leap toward scalable and resource-efficient quantum state characterization.

\emph{CRM Shadow estimation}---Consider an $n$-qubit quantum system whose Hilbert space $\caH$ has dimension $d=2^n$; denote by $\bbone$ the identity operator on~$\caH$.   The basic idea of CRM shadow estimation introduced in \rscite{garratt2023prob,vermersch2024enhanced,Vitale2024rob} is illustrated in \fref{fig:procedure}. To estimate the expectation value of an observable (Hermitian operator) $O$  with respect to an unknown system state $\rho$ on $\caH$, we first choose a prior state $\sigma$ and a suitable unitary ensemble $\caU$ on $\caH$.  In each round, a random unitary $U$ is sampled from $\caU$  and applied to the  state $\rho$ followed by a measurement in the computational basis. After repeating this procedure (for the same $U$) $R$ times and obtaining measurement outcomes $\bfs_1, \bfs_2,\ldots, \bfs_R\in \{0,1\}^n$, a CRM  estimator for  $\rho$ is constructed as  $\hrho_{\sigma}=\hrho-\hat{\sigma}+\sigma$, where 
\begin{equation}\label{eq:CRMestimator}
\begin{aligned}
	\hrho&=\frac{1}{R}\sum_{k=1}^{R}\caM^{-1}\left(U^\dag|\bfs_k \> \< \bfs_k|U\right),\\
	\hat{\sigma}&=\sum_{\bfs}P_{\sigma}(\bfs \mid U)\caM^{-1}\left(U^\dag|\bfs \> \< \bfs|U\right).	
\end{aligned}
\end{equation}
Here $\caM^{-1}(\cdot)$ denotes the inverse of the measurement channel determined by $\caU$ and is known as the reconstruction map \cite{huang2020pred}, and $P_{\sigma}(\bfs \mid U) = \< \bfs | U \sigma U^\dag | \bfs \>$. Accordingly, $\tr(O\hrho_{\sigma})$ is an estimator for $\tr(O\rho)$, whose variance  is denoted by $\bbV_R(\hO)$. Note that 
$\hrho$ is the estimator for $\rho$ in thrifty (THR) shadow estimation \cite{helsen2023thrifty,zhou2023perf,vermersch2024enhanced}, and may be regarded as a special CRM estimator with $\sigma=0$. Henceforth the above CRM and THR estimation protocols are referred to as $\CRM(\caU,\sigma,R)$ and $\THR(\caU,R)$, respectively.

\begin{figure}[t]
    \centering
    \includegraphics[width=0.4\textwidth]{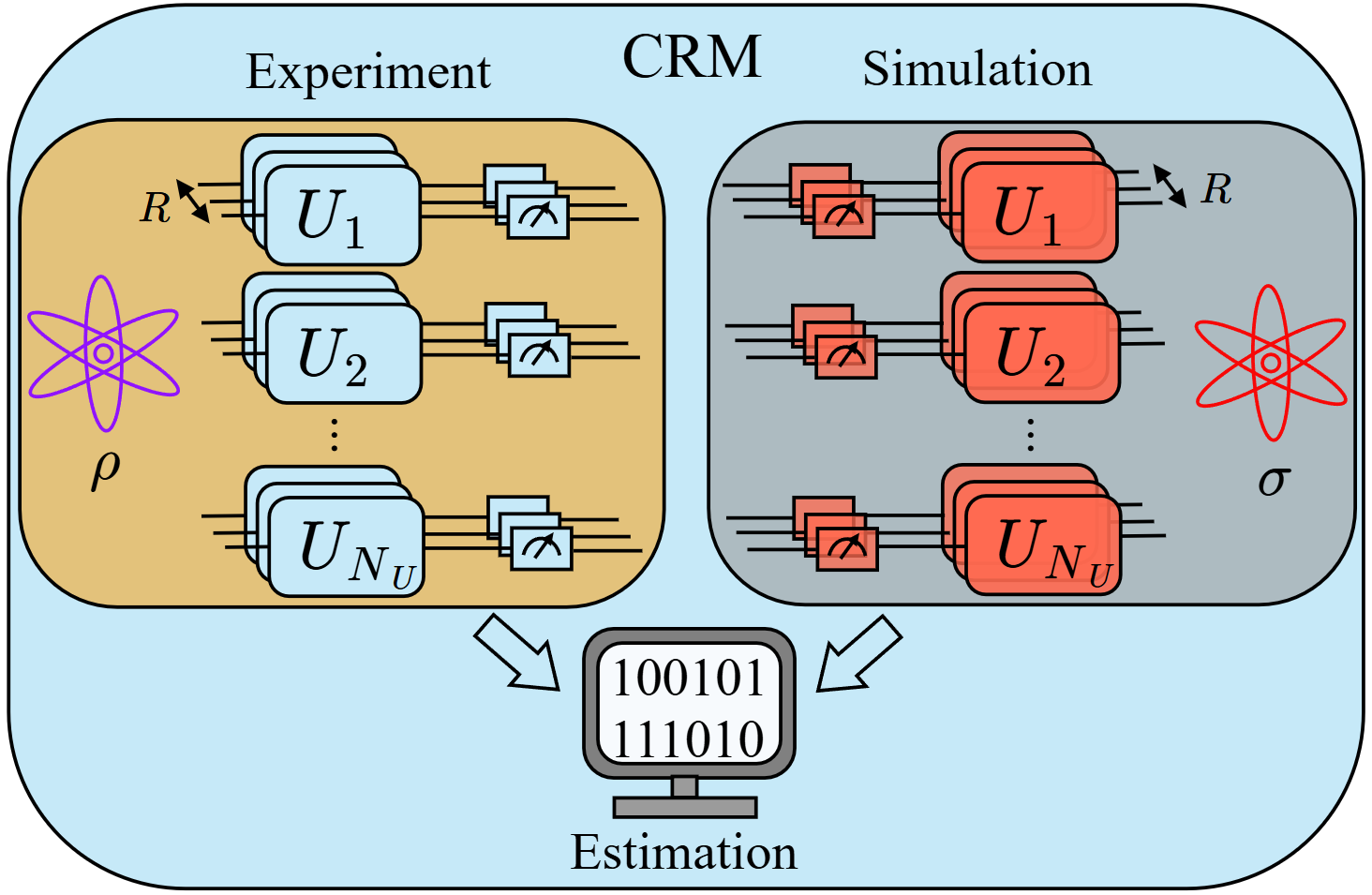}
    \caption{Schematic of CRM shadow estimation. Here $\rho$ is the unknown system state, $\sigma$ is the prior state stored in a classical computer,
   $R$ is the number of circuit reusing, and  $N_U$ is the number of circuits sampled. The procedure involving quantum measurements, shown in the orange box, is the same as in thrifty shadow estimation, while the prior state $\sigma$ is employed to simulate the outcomes of randomized measurements and to construct the CRM estimator, which can usually reduce the variance.}
    \label{fig:procedure}
\end{figure}

After repeating the above procedure $N_U$ times, $N_U$ estimators for $\tr(O\rho)$ can be constructed, corresponding to the  random unitaries $U_1, U_2, \ldots, U_{N_U}$ sampled in $N_U$ rounds.  Then  a more precise estimator can be constructed from the empirical mean. In conjunction with the median of means estimation we can further reduce the probability of large deviation. By virtue of a  similar reasoning employed in \rcite{huang2020pred}, it is straightforward to derive the following result on the sample cost. 
\begin{lemma}\label{lem:upperbound_samp}
Suppose $O$ is an observable on $\caH$. Then the  circuit sample cost required to estimate the expectation value of $O$ within error $\varepsilon$ and with probability at least $1-\delta$ is upper bounded by
	\begin{equation}
		N_U \le  \left\lceil\frac{68\bbV_R(\hO)}{\varepsilon^2} \ln\left(\frac{2}{\delta}\right)\right\rceil.\label{eq:upperbound_samp_gen}
	\end{equation}
\end{lemma}
Here the variance $\bbV_R(\hO)$ only depends on the traceless part of $O$. \Lref{lem:upperbound_samp} applies to both CRM and THR shadow estimation. 
From now on, we assume that $O$ is a traceless observable (Hermitian operator) on $\caH$, denote by $\Delta=\rho-\sigma$ the deviation of the system state $\rho$ from the prior state $\sigma$, and denote by $\epsilon=1-\lVert\sqrt{\rho}\sqrt{\sigma}\lsp\rVert_1^2$ the infidelity between $\rho$ and $\sigma$, where $\lVert\cdot\rVert_p$ ($p\geq 1$) denotes the Schatten $p$-norm. Note that $\epsilon=1-\tr(\rho \sigma)$ when $\sigma$ is a pure state.

\emph{High-precision fidelity estimation}---Efficient estimation of the fidelity between the system state $\rho$ and a given target pure state $|\phi\>$ is crucial to many applications \cite{Flammia2011Direct,Zhang2021direct,lee2025efficient}. Although shadow estimation proves to be useful for this task, the circuit sample cost is usually  inversely proportional to the mean square error, which is quite resource-intensive for HPFE. Here we will show that CRM can effectively circumvent this problem, as long as we choose the target pure state as the prior, that is, $\sigma=|\phi\>\<\phi|$. Our solution is based on two  simple but crucial observations: 1. the estimation error that can be tolerated in practical applications is usually proportional to the infidelity $\epsilon$; 2. the estimation error in CRM  with  the prior state $\sigma$ tends to decrease monotonically with  $\epsilon$. Now, the goal of HPFE is to estimate the fidelity $F=\<\phi|\rho|\phi\>$ or the expectation of the observable $O=\sigma-\bbone/d$, 
within error  $\varepsilon =r\epsilon$ and significance level $\delta$, where $r$ and $\delta$ are precision parameters determined by specific applications.

\emph{Variances in CRM shadow estimation}---To understand the performance of CRM shadow estimation and its application in HPFE, we need to evaluate the variance $\bbV_R(\hO)$. To this end, we first define the 4th cross moment operator introduced in \rcite{chen2024nonstab}:
\begin{equation}\label{eq:CrossMoment}
\Omega(\caU)=\sum_{\bfs,\bft}\bbE_{U\sim\caU}\dagtensor{U}{4}\left[\Tensor{\left(|\bfs\>\<\bfs|\right)}{2}\otimes\Tensor{\left(|\bft\>\<\bft|\right)}{2}\right]\Tensor{U}{4},
\end{equation}
where the summation runs over $\bfs, \bft\in \{0,1\}^n$. Then we can lay the groundwork for understanding CRM shadow estimation and its advantages over standard and THR shadow estimation. The following lemma is proved in  \sref{SM:VRCRMproof} in  Supplemental Material (SM), which includes \rscite{collins2006inte,roth2018recover,nielsen2020QCQI,Pfeuty1970,sachdev2007quantum,Oliviero2022Magic,Tarabunga2023Many,Tarabunga2024Non,Ding2025Eval}.
\begin{lemma}\label{lem:VRCRM}
The variance $\bbV_R(\hO)$	in  $\CRM(\caU,\sigma,R)$ reads
    \begin{equation}\label{eq:generalvar0}
    \bbV_R(\hO)=\bbV_*(O,\Delta)+\frac{\bbV(O,\rho)-\bbV_*(O,\rho)}{R},
\end{equation}
where $\bbV(O,\rho)$ denotes the variance in standard shadow estimation and
\begin{equation}\label{eq:V*}
    \bbV_*(O,\tau):=\tr\left\{\Omega(\caU)\Tensor{\left[\caM^{-1}(O)\otimes\tau\right]}{2}\right\}-[\tr(O\tau)]^2
\end{equation}
for $\tau=\Delta$ and $\tau=\rho$. 
\end{lemma}
In the large-$R$ limit, we have $\bbV_R(\hO)=\bbV_*(O,\Delta)$ and $	N_U=\caO(\bbV_*(O,\Delta))$ by \lref{lem:upperbound_samp}, so $\bbV_*(O,\Delta)$ determines  the variance  and circuit sample cost.
By replacing $\sigma$  and $\Delta$ with 0 and $\rho$ in \eref{eq:generalvar0} we can reproduce the variance in THR shadow estimation \cite{helsen2023thrifty,chen2024nonstab}, which is   effective if $\bbV_*(O,\rho)\ll\bbV(O,\rho)$, but not so effective if $\bbV_*(O,\rho)\approx\bbV(O,\rho)$. In the latter case, CRM is particularly appealing as we will see.

In the rest of this paper, we are mainly interested in estimation protocols based on 3-designs. The following lemma is a simple corollary of \lsref{lem:upperbound_samp}, \ref{lem:VRCRM} above and \pref{pro:3designVar} in  Appendix A.
\begin{lemma}\label{lem:upperbound_samp2}
Suppose $\caU$ is a unitary 3-design.  Then the variance $\bbV_R(\hO)$ in  $\CRM(\caU,\sigma,R)$ satisfies
\begin{gather}
     \bbV_R(\hO)\leq \bbV_*(O,\Delta)+ \frac{2\lVert O\rVert_2^2}{R}.
\end{gather}
In HPFE with $O=\sigma-\bbone/d$, the circuit sample cost required to achieve precision $r\epsilon$ satisfies 
	\begin{gather}
		N_U \leq \left \lceil\frac{ 68\left[V_*(O,\Delta)+\frac{2}{R}\right]}{r^2\epsilon^2} \ln\left(\frac{2}{\delta}\right)\right\rceil.
	\end{gather}
\end{lemma}

\begin{figure*}
    \centering
    \includegraphics[width =0.8\linewidth]{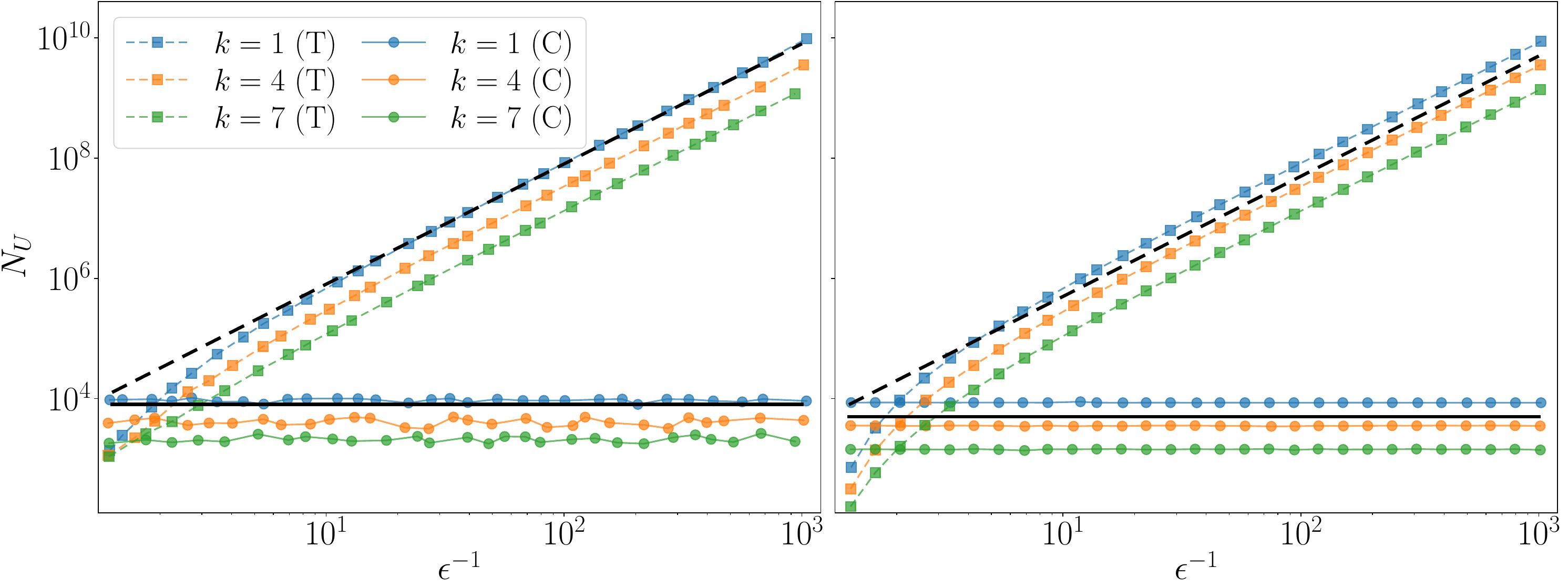}
    \caption{Circuit sample costs  $N_U$ required for HPFE in THR (T) and CRM (C) shadow estimation based on  Clifford  measurements. Here, $ r = 0.25 $, $ \delta = 0.01$, and $R = \lceil 10/\epsilon^2 \rceil$. The target and prior state has the form $\sigma=|S_{7,k}\> \<S_{7,k}|$ with $k=1,4,7$.  Different system  states $\rho$ are generated by applying  random  local rotations described in \sref{sec:DescriptRandomRotation} (left plot) and random Pauli channels described in \sref{sec:DescriptRandomPauli} (right plot) to the target state $\sigma$. The black solid and dashed lines represent $N_U =8000 $ and $N_U =8000 /\epsilon^2$, respectively, in the left plot, while they represent $N_U =5000$ and $N_U =5000/\epsilon^2$ in the right plot. }\label{fig:hpfe_together}
\end{figure*}

\emph{CRM based on 4-design measurements}---From now on we turn to CRM shadow estimation based on more concrete unitary ensembles. First, we clarify the performance of  4-design measurements. 
\begin{theorem}\label{thm:4designVstar}
    Suppose $\caU$ is a unitary 4-design. Then the variance  $\bbV_*(O,\Delta)$ in  $\CRM(\caU,\sigma,R)$ satisfies
 \begin{align}
 \frac{ \lVert\Delta\rVert_2^2\lVert O\rVert_2^2}{d+2}&\le\bbV_*(O,\Delta)\le\frac{4 \lVert\Delta\rVert_2^2\rVert O\lVert_2^2}{d}\leq \frac{8\epsilon\lVert O \rVert_2^2}{d}
\label{eq:4duppline1}.
\end{align}

\end{theorem}
\Thref{thm:4designVstar} shows that CRM is very effective for enhancing the performance of shadow estimation based on 4-designs. Notably,  the asymptotic variance $\bbV_*(O,\Delta)$ is mainly determined by $\lVert\Delta\rVert_2^2$, which decreases monotonically as the system state $\rho$ and prior state $\sigma$ gets closer. 
For HPFE, the prior state $\sigma$ coincides with the target pure state and
we have $O=\sigma-\bbone/d$ and $\lVert O\rVert_2^2\leq 1$, which means  $\bbV_R(\hO)\le4 \lVert\Delta\rVert_2^2/d+2/R \le8\epsilon /d+2/R$. In this case, CRM can even improve the scaling behavior of the sample cost with respect to the infidelity. The following theorem is a simple corollary of \lref{lem:upperbound_samp2} and \thref{thm:4designVstar}.

\begin{theorem}\label{thm:4designHPFE}
 Suppose $\caU$ is a unitary 4-design. Then the circuit sample cost required to achieve  HPFE  in  $\CRM(\caU,\sigma,R)$ satisfies 	
\begin{equation}
	N_U  \le\left\lceil \frac{136\bigl (\frac{2 \lVert\Delta\rVert_2^2 }{d}+\frac{1}{R}\bigr)}{r^2 \epsilon^2}  \ln\left(\frac{2}{\delta}\right)\right\rceil.
\end{equation}
\end{theorem}

Since  $\lVert \Delta\rVert_2^2\le2\epsilon$ thanks to \pref{pro:Delta12NormIFsim} in  Appendix B,  CRM based on 4-design measurements can achieve the scaling behavior $N_U=\caO(1/\epsilon)$  in the large-$R$ limit, which is much better than the counterpart  $\caO(1/\epsilon^2)$ for THR. Moreover, 
CRM can even achieve a constant scaling behavior whenever $\lVert\Delta\rVert_2^2=\caO(\epsilon^2)$, which holds for many scenarios of practical interest. Notably,  if $\sigma$ is a Haar-random pure state and $\rho=\scrP(\sigma)$ with $\scrP$ being a Pauli channel, then we have $\lVert\Delta\rVert_2^2=\caO(\epsilon^2)$ on average thanks to \pref{pro:pchannel_haarSim}
in Appendix B. For a sufficiently large system with $d \gg 1/(r\epsilon)^2$, the circuit sample cost may even drop to $N_U = 1$.

\emph{CRM based on Clifford measurements}---In real experiments, it is quite difficult to realize  4-designs. 
As a simpler alternative, the Clifford group~$\Cl_n$ forms a 3-design \cite{kueng2015qubit,webb2016clifford,Zhu20173Design} and can be regarded as the best approximation to a 4-design among discrete groups \cite{zhu2016clifford,bannai2020unitary}. Can we achieve a similar performance in CRM shadow estimation based on Clifford measurements instead of 4-design measurements?

To answer the above question, we first introduce some additional concepts. Let $\bcaP_n=\Tensor{\{I,X,Y,Z\}}{n}$ be the set of $n$-qubit Pauli operators. The characteristic function $\Xi_O$ of the observable $O$  is  defined  as $	\Xi_O(P) := \tr(OP)$ for $P \in \bcaP_n$,
which can be regarded as a vector of $d^2$ entries. The stabilizer 2-R\'{e}nyi entropy (2-SRE) of a pure state $\sigma=|\phi\>\<\phi|$ reads  $M_2(\phi)=M_2(\sigma)=\log_2\left(d/\rVert \Xi_\sigma\rVert_{4}^{4}\right)$
\cite{leone2022stab,Leone2024monotone},
where $\lVert\Xi_\sigma\rVert_{p}$ for $p\geq 1$ denotes the $\ell_p$-norm of $\Xi_\sigma$. The cross characteristic function $\Xi_{O_1,O_2}$ and twisted cross characteristic function $\tXi_{O_1,O_2}$ of two observables $O_1,O_2$ on $\caH$ are defined as \cite{chen2024nonstab}
\begin{equation}\label{eq:CrossChar}
\begin{aligned}
	\Xi_{O_1,O_2}(P) &:= \tr(O_1 P)\tr(O_2 P), \quad  P \in \bcaP_n ,\\
	\tXi_{O_1,O_2}(P) &:= \tr(O_1 P O_2 P), \quad  P \in \bcaP_n.
\end{aligned}
\end{equation}
Now, we can introduce an upper bound for $\bbV_*(O,\Delta)$:
 \begin{align}\label{eq:Vcirc}
\bbV_\circ(O,\Delta):=\frac{d+1}{d(d+2)}\left(\lVert\Xi_{\Delta,O}\rVert_2^2+\tXi_{\Delta,O}\cdot \Xi_{\Delta,O}\right),
 \end{align}
 where $\tXi_{\Delta,O}\cdot \Xi_{\Delta,O}$ denotes the inner product between two cross characteristic functions viewed as vectors.
 \Thsref{thm:CliffordVstar}{thm:CliffordHPFEpn} below are proved in SM \sref{SM:CliffordVar}. 
 \begin{theorem}\label{thm:CliffordVstar}
 	 The variance $\bbV_*(O,\Delta)$ associated with $\CRM(\Cl_n,\sigma,R)$ reads
 	\begin{align}
 		\bbV_*(O,\Delta)&=\bbV_\circ(O,\Delta)-\frac{\left[\tr\left(\Delta O\right)\right]^2}{d+2} \le\frac{2}{d}\lVert\Xi_{\Delta,O}\rVert_2^2\nonumber\\
 		&\le 2\lVert\Delta\rVert_1^2\lVert O\lVert_2^2\le 8\epsilon \lVert O \rVert_2^2. 
 		\label{eq:Clifford_variance}
 	\end{align}
 	If $\sigma$ is a pure state and $O=\sigma-\bbone/d$, then 
 	\begin{align}
 		\bbV_*(O,\Delta)& \le 2^{2-M_2(\sigma)/2}\min\left\{\sqrt{2}\lsp \epsilon, \lVert\Delta\rVert_2^2\right\}.
 		\label{eq:CliffordVarFE1}
 	\end{align}		
 \end{theorem}

 \Thref{thm:CliffordVstar} shows that the variance $\bbV_*(O,\Delta)$ tends to decrease monotonically when the system state $\rho$ and the prior state $\sigma$ gets closer as in CRM shadow estimation based on 4-designs. This fact is particularly important to HPFE. As a simple corollary of \lref{lem:upperbound_samp2} and \thref{thm:CliffordVstar}, we can deduce the following theorem.

\begin{theorem}\label{thm:CliffordHPFE}
 Suppose $\sigma$ is a pure state on $\caH$. Then the  circuit sample cost $N_U$ required for HPFE in $\CRM(\Cl_n,\sigma,R)$ satisfies
	\begin{align}
	&N_U\le  \left\lceil\frac{136\left[2^{1-M_2(\sigma)/2}\lVert\Delta\rVert_2^2+\frac{1}{R}\right]}{r^2\epsilon^2}\ln\left(\frac{2}{\delta}\right)\right\rceil.\label{eq:hpfeUppGeneral}
	\end{align}    
\end{theorem} 

Thanks to  \thsref{thm:CliffordVstar}{thm:CliffordHPFE} and \psref{pro:Delta12NormIFsim}-\ref{pro:CharPropSim} in Appendix B, CRM shadow estimation based on Clifford measurements is very effective for fidelity estimation even when the target state is a stabilizer state, for which THR is useless. In the large-$R$ limit,  the circuit sample cost $N_U$  tends to decrease exponentially with  2-SRE.  In addition,   $N_U$ increases at most linearly with $1/\epsilon$ instead of $1/\epsilon^2$, unlike THR. Moreover, constant scaling can be achieved whenever 
$\lVert\Delta\rVert_2^2=\caO(\epsilon^2)$ or $\lVert\Xi_{\Delta,O}\rVert_2^2/d=\caO(\epsilon^2)$, which holds in many scenarios of practical interest. Notably, this is the case when $\rho$ is related to the target state via a depolarizing or Pauli channel by \pref{pro:CharPropSim}.
\begin{theorem}\label{thm:CliffordHPFEpn}
    Suppose $\sigma$ is a pure state on $\caH$ and  $\rho=\scrP(\sigma)$, where $\scrP$ is a Pauli channel. Then the  circuit sample cost $N_U$ required for HPFE in $\CRM(\Cl_n,\sigma,R)$ satisfies
    \begin{equation}
        N_U\le\left\lceil136\ln\left(\frac{2}{\delta}\right)\left(\frac{2}{r^2}+\frac{1}{Rr^2\epsilon^2}\right)\right\rceil.\label{eq:hpfe_pauligen}
    \end{equation}
    If $\scrP$ is a depolarizing channel, then 
    	\begin{equation}
		N_U\le\left\lceil136\ln\left(\frac{2}{\delta}\right)\left[\frac{2^{-M_2(\sigma)}}{r^2}+\frac{1}{Rr^2\epsilon^2}\right]\right\rceil.\label{eq:hpfe_depo}
	\end{equation}
\end{theorem}

\emph{Numerical illustration}---To demonstrate the advantages  of CRM over THR in HPFE,  we consider an example in which the target state is $|S_{n,k}\>:=|0\>^{\otimes(n-k)}\otimes |T\>^{\otimes k}$, where $|T\>:=\left(|0\>+\rme^{\pi \rmi/4}|1\>\right)/\sqrt{2}$ is a magic state and $M_2\left(|S_{n,k}\>\right)=k\log_2(4/3)$ \cite{chen2024nonstab}.
Different system  states are generated by applying random local rotations (coherent noise) or  random Pauli channels (incoherent noise) to $\sigma$. \Fref{fig:hpfe_together} shows the circuit sample costs required in THR and CRM as functions of the infidelity~$\epsilon$. In both cases, THR requires $\caO(1/\epsilon^2)$ circuits, while CRM only requires a constant number of circuits, which is dramatically better. For HPFE with $\epsilon\le 0.001$, CRM can reduce the circuit sample cost millions of times.  As shown in SM \sref{SM:NumericsAdd}, for certain special coherent noise, CRM requires $\caO(1/\epsilon)$ circuits, which is suboptimal but still much better than THR.  In addition,  the 2-SRE of the target state can significantly enhance the performance of  Clifford-based CRM in HPFE, which is similar to THR as shown in \rcite{chen2024nonstab}. See SM \sref{SM:NumericsAdd} for more numerical results.

\begin{figure}[t]
    \centering
    \includegraphics[width=0.4\textwidth]{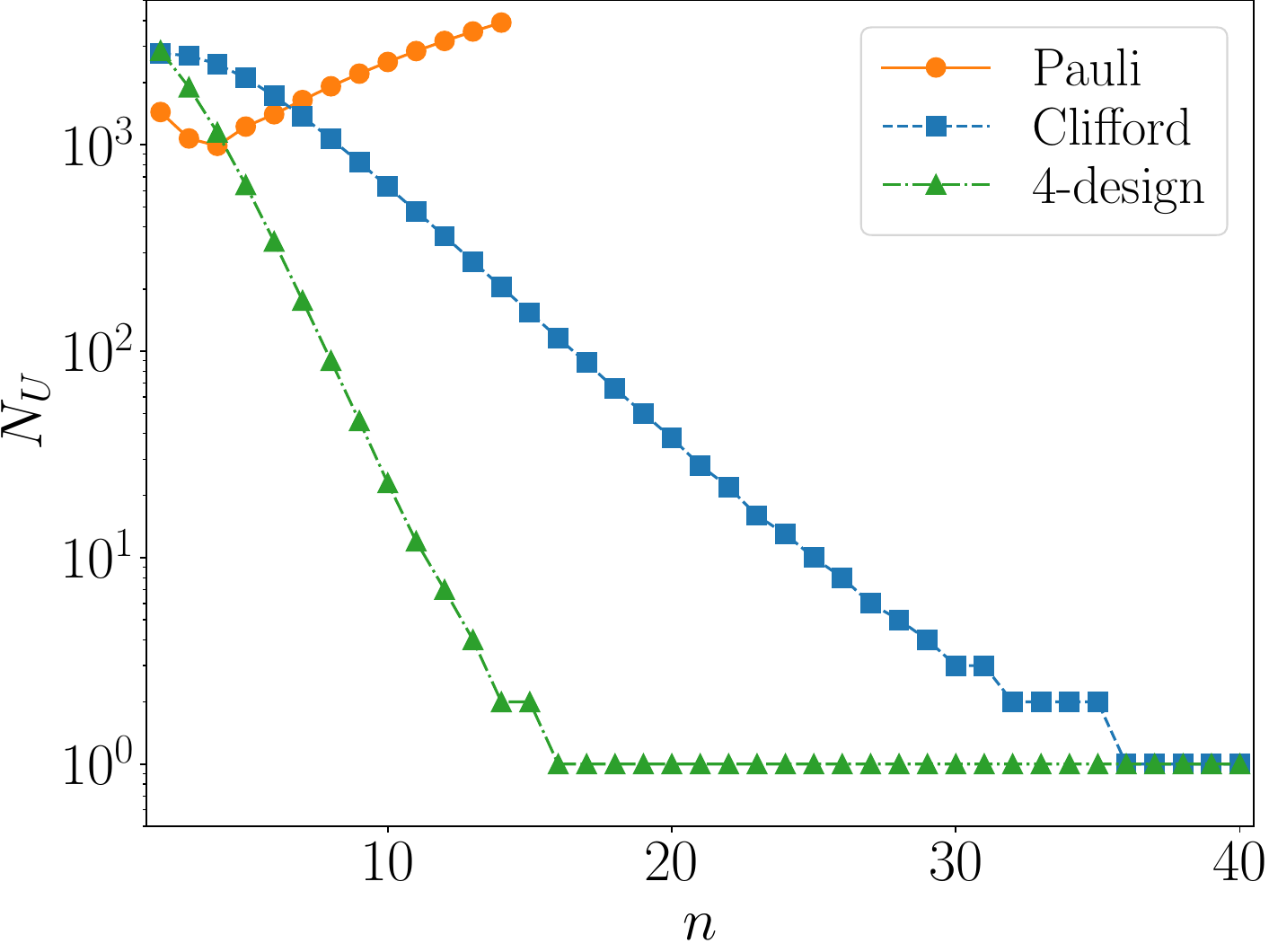}
    \caption{Circuit sample costs $ N_U $ required for HPFE in CRM shadow estimation based on three measurement ensembles, plotted as  functions of the qubit number $n$.
Here $r=0.25,\delta=0.01$, $\epsilon=0.001$, and $R=2\times10^{10}$. The target state $\sigma=|\MC_n\>\<\MC_n|$ is the $n$-qubit magic cluster state with $n=2,3,\ldots,40$. Different system states $\rho$ are generated by applying depolarizing noise to the target state: $\rho =(1-p)\sigma +p\bbone/d$ with $p=\epsilon/(1-d^{-1})$. }
    \label{fig:three_ensembles_cluster}
\end{figure}

Next, we compare the performance of 4-design, Clifford, and Pauli measurements in HPFE based on CRM shadow estimation. Here we choose the $n$-qubit magic cluster state  $|\MC_n\>:=\prod_{i=1}^{n-1}\text{CZ}_{i,i+1}|S_{n,n}\>$ as the target state, and the system state is affected by depolarizing noise. The circuit sample costs $N_U$ tied to the three measurement ensembles are shown in \fref{fig:three_ensembles_cluster}. For 4-design and Clifford  measurements, $N_U$ decrease exponentially with the qubit number $n$ until $N_U=1$, which is expected according to \thsref{thm:4designHPFE}{thm:CliffordHPFEpn} given that $M_2(\MC_n)=M_2(S_{n,n})$ increases linearly with $n$. For Pauli measurements, by contrast, $N_U$ increases exponentially with $n$ for $n\geq 4$. We can expect similar results when $|\MC_n\>$ is replaced by  generic multiqubit pure states, which have nearly maximal entanglement and nonstabilizerness.

In the above example, 4-design measurements are the most efficient. In certain situations, surprisingly, Clifford measurements can achieve a  better scaling behavior with respect to the infidelity, as shown in Appendix C.  According  to {\thsref{thm:4designVstar}{thm:CliffordHPFE}}, intuitively, this peculiar phenomenon would occur when $\lVert\Delta\rVert_2^2=\caO(\epsilon)$, but   $\lVert\Xi_{\Delta,O}\rVert_2^2/d=\caO(\epsilon^2)$. This fact is a blessing for practical quantum information processing since it is much easier to realize Clifford measurements than 4-design measurements.

\emph{Summary}---We introduced a general and highly efficient fidelity estimation protocol by integrating CRM with Clifford-based shadow estimation. To estimate the infidelity $\epsilon$ with multiplicative precision, our protocol requires only $\caO(1/\epsilon)$ circuits instead of $\caO(1/\epsilon^2)$ circuits required in conventional approaches, including standard and THR shadow estimation. Moreover, in practical scenarios affected by depolarizing, Pauli, or certain coherent  noise, our protocol demands only a constant number of circuits, irrespective of the infidelity $\epsilon$ and the qubit number~$n$. For typical  target pure states of large and intermediate quantum systems, only one circuit is required.  For HPFE with $\epsilon\leq 0.001$, our protocol can reduce the circuit sample cost by factors of thousands to millions, which is instrumental to practical applications. In the course of study, we 
derived general analytical formulas for the variances in CRM shadow estimation based on  4-design, Clifford, and Pauli measurements, which are of interest beyond fidelity estimation. Our work marks a substantial step forward in developing  scalable learning protocols for high-precision applications, with broad implications for  quantum information processing. In the future, it would be desirable to generalize our results to qudits and measurement ensembles based on shallow circuits.

\emph{Acknowledgments}---This work is supported by the National Natural Science Foundation
of China (Grant No.~92365202 and No.~12475011), Shanghai Science and Technology Innovation Action Plan (Grant No.~24LZ1400200), Shanghai Municipal Science and Technology Major Project (Grant No.~2019SHZDZX01), and the National Key Research and Development Program of China (Grant No.~2022YFA1404204).

\let\oldaddcontentsline\addcontentsline
\renewcommand{\addcontentsline}[3]{}
\bibliography{ref}
\let\addcontentsline\oldaddcontentsline
\clearpage

\newpage
\let\oldaddcontentsline\addcontentsline
\renewcommand{\addcontentsline}[3]{}
\appendix

Across this Appendix, $\caH$ denotes  a $d$-dimensional Hilbert space on which all operators act, and  $\bbone$ denotes the identity operator as in the main text.   In addition, $O$ denotes a traceless observable, $\rho,\sigma$ denote two quantum states, $\Delta=\rho-\sigma$, and $\epsilon=1-\lVert\sqrt{\rho}\sqrt{\sigma}\lsp\rVert_1^2$ denotes the infidelity between $\rho$ and $\sigma$, which can be simplified as $\epsilon=1-\tr(\rho\sigma)$
when $\sigma$ is a pure state. In CRM shadow estimation, $\rho$ and $\sigma$ denote the system state and prior state, respectively; in HPFE, $\sigma$ also denotes the target pure state.

\section{Variance in shadow estimation based on 3-designs}
Here we offer a universal upper bound for the variance in shadow estimation based on 3-designs, which is nearly tight. The following proposition is proved in SM \sref{pro:3designVarProof}. 
\begin{proposition}\label{pro:3designVar}
	Suppose $\caU$ is a unitary 3-design on $\caH$. Then
	\begin{align}
	\bbV(O,\rho)&=\frac{d+1}{d+2}\left[\tr\left(O^2\right)+2\tr\left(\rho O^2\right)\right]-\left[\tr(\rho O)\right]^2\nonumber\\
	&\le 2\rVert O\rVert_2^2.\label{eq:3designVar}
	\end{align}
\end{proposition}
The equality in \eref{eq:3designVar} was derived in \rcite{huang2020pred}, which also proved a weaker upper bound, that is,  $\bbV(O,\rho)\leq 3\rVert O\rVert_2^2$.
When $O=|\phi\>\<\phi|-\bbone/d$ with $|\phi\>$ being a pure state in $\caH$, we have \cite{chen2024nonstab}
\begin{equation}
	\bbV(O,\rho)=-F^2+\frac{d(2F+1)}{d+2}<2-\epsilon^2\le2,\label{eq:v_vstar2}
\end{equation}
where $F = \<\phi|\rho|\phi\>$ is the fidelity between $\rho $ and $|\phi\>$ and $\epsilon=1-F$ is the infidelity. This result shows that the inequality in \eref{eq:3designVar} is nearly tight.

\section{Relations between $\lVert\Delta\rVert_2^2$, $\lVert\Delta\rVert_1^2$, $\Vert\Xi_{\Delta,O}\rVert_2^2$, and the infidelity $\epsilon$}\label{app:NormsInfid}
\begin{figure*}
	\centering
	\includegraphics[width=0.8\linewidth]{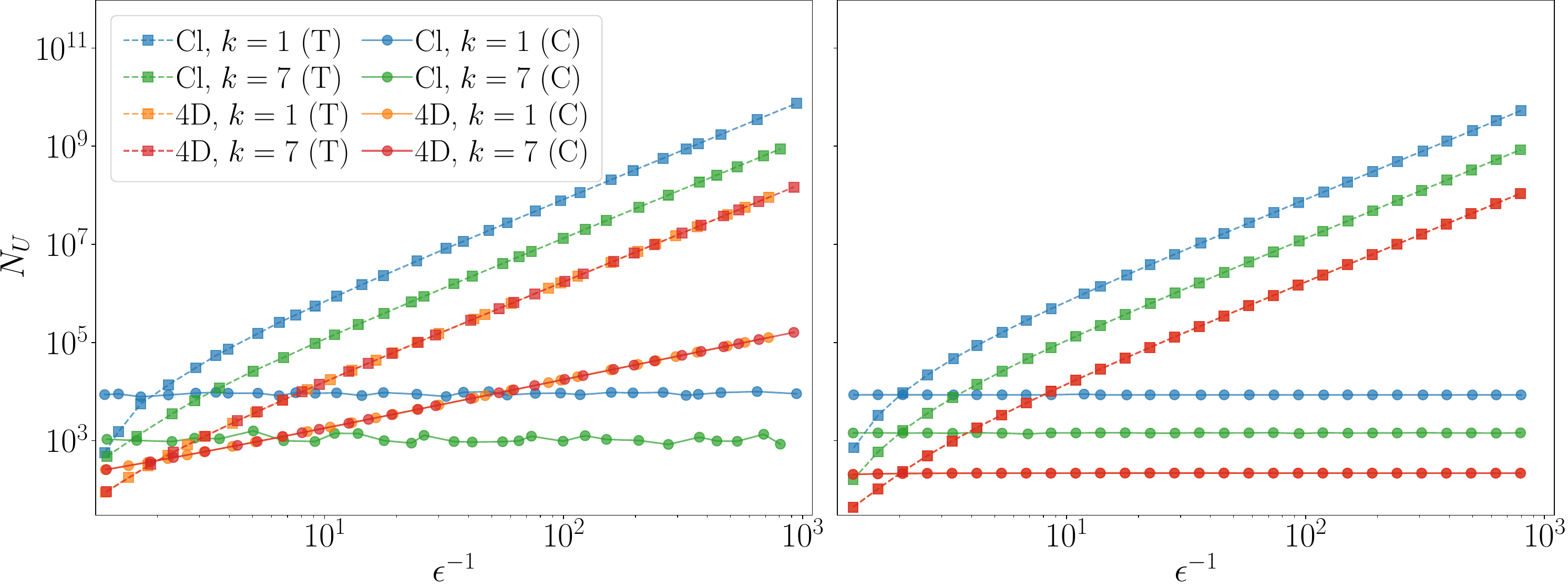}
	\caption{Circuit sample costs  $N_U$ required for HPFE in THR (T) and CRM (C) shadow estimation based on Clifford (Cl)  and  4-design (4D) measurements. Here, $ r = 0.25 $, $ \delta = 0.01$, and $R = \lceil d/\epsilon^2 \rceil$. The target and prior state has the form $\sigma=|S_{7,k}\> \<S_{7,k}|$ with $k=1,7$, and different system  states $\rho$ are generated by applying  random  local rotations described in SM \sref{sec:DescriptRandomRotation} (left plot) and random Pauli channels described in SM \sref{sec:DescriptRandomPauli} (right plot) to the target state $\sigma$. The results on  4-design measurements are almost independent of $k$, especially in the right plot. 
	}\label{fig:hpfe_4dcom}
\end{figure*}

To better understand the performance of HPFE based on CRM shadow estimation, here we clarify the relations between $\lVert\Delta\rVert_2^2$, $\lVert\Delta\rVert_1^2$, $\Vert\Xi_{\Delta,O}\rVert_2^2$, and $\epsilon$, assuming that $\sigma$ is a pure state; see SM \ssref{SM:NormsInfid} and \ref{SM:CharProp} for additional results.

\begin{proposition}\label{pro:Delta12NormIFsim} 
	The following inequalities hold:
	\begin{gather}
		2\epsilon^2\le 	2\lVert\Delta\rVert_2^2\le \lVert\Delta\rVert_1^2\le \min\left\{4\epsilon,4\lVert\Delta\rVert_2^2\right\}; \label{eq:Delta12NormIF}
	\end{gather}
the inequality 	$\lVert\Delta\rVert_2^2\leq 2\epsilon$ is saturated if and only if $\rho$ is a pure state. 
\end{proposition}

\begin{proposition}\label{pro:pchannel_haarSim}
	Suppose $\sigma$ is a Haar-random pure state on $\caH$ and $\rho=\scrP(\sigma)$
	with $\scrP$ being a Pauli channel. Then 
	\begin{gather}
		\overline{\epsilon^2}\le\overline{\lVert\Delta\rVert_2^2}\le\frac{2(d+1)}{d}\bepsilon^2,\;\;
		4\lsp\overline{\epsilon^2}\le\overline{\lVert\Delta\rVert_1^2}\le 8\bepsilon^2.\label{eq:MeanDelta1NormIFpauli}
	\end{gather}
\end{proposition}

\begin{proposition}\label{pro:CharPropSim}
	Suppose $O=\sigma-\bbone/d$; then 
	\begin{align}
		\lVert \Xi_{\Delta,O}\rVert_2^2&= \lVert \Xi_{\Delta,\sigma}\rVert_2^2 \le 2^{[3-M_2(\sigma)]/2}d\epsilon. \label{eq:CharPropFEsim}
	\end{align}		
	If $\rho=\scrP(\sigma)$ with $\scrP$ being a Pauli channel, then 
	\begin{align}
		\lVert \Xi_{\Delta,O}\rVert_2^2=\lVert\Xi_{\Delta,\sigma}\lVert_2^2\leq 2d\epsilon^2. \label{eq:CharPropPauliFEsim}
	\end{align}
\end{proposition}

\Psref{pro:Delta12NormIFsim} and  \ref{pro:pchannel_haarSim} follow from \psref{pro:Delta12NormIF} and \ref{pro:pchannel_haar} in SM \sref{SM:NormsInfid}, respectively. \Pref{pro:CharPropSim} follows from \lsref{lem:CharProp} and \ref{lem:InfidCrossCharUB} in SM \sref{SM:CharProp}.

\section{Comparison of Clifford measurements and 4-design measurements}

In this appendix, we show that Clifford measurements are not only simpler but also more efficient than 4-design measurements in HPFE in the presence of certain coherent noise of practical interest.

As a complement to \fref{fig:hpfe_together}, which shows the circuit sample cost of HPFE based on  Clifford measurements, here we consider the counterpart based on  4-design measurements. The target and prior state has the form $\sigma=|S_{7,k}\> \<S_{7,k}|$ with $k=1,7$.  Both random local rotations and random Pauli noise are considered as before. The sample costs required  in  THR and CRM shadow estimation can be determined by virtue of  \lsref{lem:upperbound_samp} and \ref{lem:VRCRM} together with the variance formulas for $\bbV(O,\rho)$, $\bbV_*(O,\rho)$, and $\bbV_*(O,\Delta)$ presented in \eref{eq:3designVar} and \lref{lem:V*4design} in SM \sref{SM:4_design}, and the results are shown in \fref{fig:hpfe_4dcom}.

For THR shadow estimation based on 4-design measurements, the sample cost scales with $1/\epsilon^2$ for both noise types.
For CRM shadow estimation based on 4-design measurements, by contrast, the sample cost scales with $1/\epsilon$  in the first case (random local rotations), but is almost independent of $\epsilon$ in the second case (random Pauli noise). Although CRM can significantly improve the scaling behavior, the improvement in the first case is not so dramatic compared with the counterpart for Clifford measurements, which can achieve a constant scaling behavior. Therefore, Clifford measurements are not only simpler but also more efficient  for HPFE in the presence of such coherent noise. In addition,  2-SRE (or the parameter $k$) has little influence on the sample costs for both THR and CRM shadow estimation based on 4-design measurements, unlike the counterparts based on Clifford measurements.

Next, we offer another example to further demonstrate the advantages of Clifford measurements over 4-design measurements.  
Here  the target and prior state $\sigma$ is a pure state, and the system state has the form $\rho=(\sigma+ Z_1\sigma Z_1)/2$, which means 
\begin{align}
	\Delta=\frac{1}{2} (Z_1\sigma Z_1-\sigma),\;\; \lVert \Delta\rVert _2^2=\epsilon=1-\tr(\rho\sigma). 
\end{align}
According to \thref{thm:4designVstar}, the variance $\bbV_*(O,\Delta)$ in CRM shadow estimation based on 4-design measurements satisfies
\begin{align}
	\frac{(d-1)\epsilon}{d(d+2)}&\le\bbV_*(O,\Delta)\le\frac{4 \epsilon}{d}.
\end{align}
By contrast, thanks to 
\thref{thm:CliffordVstar} and \pref{pro:CharPropSim}, the variance $\bbV_*(O,\Delta)$  tied to Clifford  measurements satisfies $\bbV_*(O,\Delta)\leq 4\epsilon^2$, which means a constant scaling behavior of the  circuit sample cost $N_U$ with respect to the target precision. So Clifford  measurements become more and more advantageous as the target precision rises.

In the numerical simulation, each target state is a seven-qubit pure state of the form $\sigma=|\phi\>\<\phi|$ with
\begin{align}\label{eq:SampleCliffordAdv4design}
	|\phi\>=\alpha |\phi_0\>\otimes |0\>+\sqrt{1-\alpha^2} \lsp |\phi_1\>\otimes |1\>,
\end{align}
where $|\phi_0\>$ and $|\phi_1\>$ are six-qubit Haar-random pure states and $\alpha$ is a random real number such that $\ln \alpha$ is uniformly distributed in the interval $[-5,0]$. The sample costs required for HPFE are shown in \fref{fig:Cliffordvs4designHPFE}. Notably,  4-design measurements require $\caO(1/\epsilon)$ circuits, while Clifford measurements only require a constant number of circuits.
\begin{figure}
	\centering
	\includegraphics[width=0.4\textwidth]{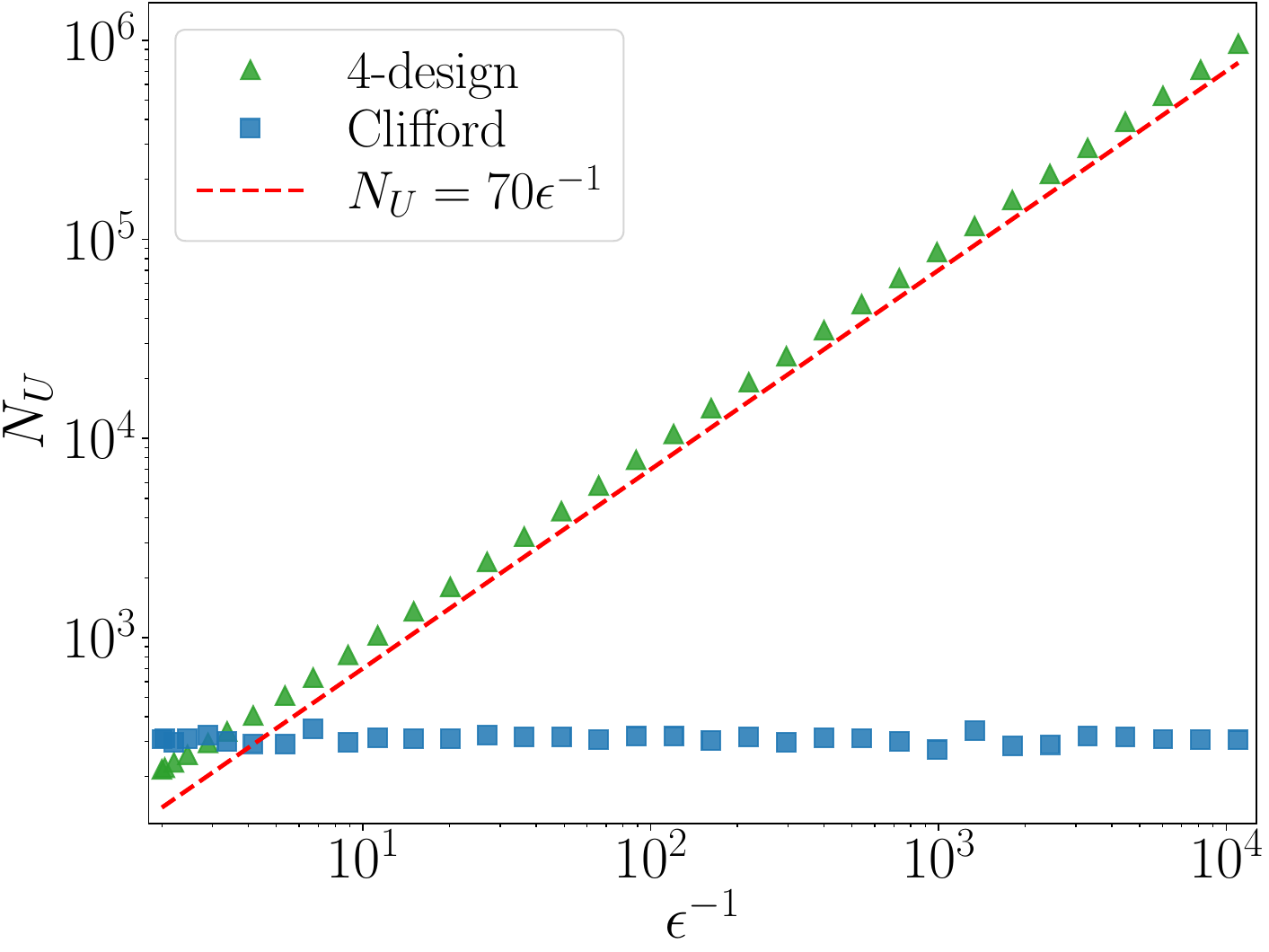}
	\caption{Circuit sample costs  $N_U$ required for HPFE in CRM shadow estimation based on  4-design and Clifford measurements. Here,  $r = 0.25 $, $\delta = 0.01 $, and $ R = \lceil d/\epsilon^2 \rceil$. The target and prior states $\sigma$ are random seven-qubit pure states that have the form in \eref{eq:SampleCliffordAdv4design}, and the system states have the form $\rho=\frac{1}{2}\sigma+\frac{1}{2} Z_1\sigma Z_1$, where $Z_1$ denotes the Pauli $Z$ operator acting on the first qubit.}\label{fig:Cliffordvs4designHPFE}
\end{figure}

\let\addcontentsline\oldaddcontentsline

\clearpage

\newpage

\setcounter{equation}{0}
\setcounter{figure}{0}
\setcounter{table}{0}
\setcounter{theorem}{0}
\setcounter{lemma}{0}
\setcounter{section}{0}
\counterwithout{equation}{section}

\setcounter{page}{1}

\renewcommand{\theequation}{S\arabic{equation}}
\renewcommand{\thefigure}{S\arabic{figure}}
\renewcommand{\thetable}{S\arabic{table}}
\renewcommand{\thetheorem}{S\arabic{theorem}}
\renewcommand{\thelemma}{S\arabic{lemma}}
\renewcommand{\theproposition}{S\arabic{proposition}}
\renewcommand{\thecorollary}{S\arabic{corollary}}
\renewcommand{\thesection}{S\arabic{section}}

\makeatletter
\renewcommand{\theHtheorem}{S\arabic{theorem}}
\renewcommand{\theHlemma}{S\arabic{lemma}}
\renewcommand{\theHproposition}{S\arabic{proposition}}
\renewcommand{\theHcorollary}{S\arabic{corollary}}
\renewcommand{\theHequation}{S\arabic{equation}}
\renewcommand{\theHfigure}{S\arabic{figure}}
\renewcommand{\theHtable}{S\arabic{table}}
\renewcommand{\theHsection}{S\arabic{section}}
\makeatother

\onecolumngrid	
\begin{center}
	\textbf{\large High-Precision Fidelity Estimation with Common Randomized Measurements:
Supplemental Material}
\end{center}

\tableofcontents

\bigskip

In this Supplemental Material (SM), we prove the results presented in Appendix A and the main text, including  \pref{pro:3designVar}, \lref{lem:VRCRM}, and  \thssref{thm:4designVstar}{thm:CliffordVstar}{thm:CliffordHPFEpn}. In addition, we clarify the average performance of CRM shadow estimation based on 2-designs and provide additional results on CRM shadow estimation based on 4-design, Clifford, and Pauli measurements. To this end, we 
explicate the relations between $\rVert\Delta\rVert_2^2$, $\lVert\Delta\rVert_1^2$, and the infidelity $\epsilon$ and elucidate the  properties of cross characteristic functions, which are crucial to  understanding HPFE based on CRM shadow estimation. Finally, we provide extensive numerical results on the performance of CRM shadow estimation based on Clifford, Pauli, and 4-designs measurements, and explicate the dependence of the circuit sample costs on the nonstabilizerness and the number $R$ of circuit reusing.

Throughout this SM, $\caH$ denotes  a $d$-dimensional Hilbert space with $d=2^n$, and  $\bbone$ denotes the identity operator on $\caH$ as in the main text.  We also denote by $\caL(\caH)$, $\caL^\rmH(\caH)$, and $\caD(\caH)$ the sets of linear operators, Hermitian operators, and  density operators on $\caH$,  respectively, and denote by $\rmU(\caH)$ the group of unitary operators on $\caH$. In addition, $O, Q\in \caL(\caH)$ and $O$ is traceless; 
$\rho,\sigma\in \caD(\caH)$, $\Delta=\rho-\sigma$, and $\epsilon=1-\lVert\sqrt{\rho}\sqrt{\sigma}\lsp\rVert_1^2$ denotes the infidelity between $\rho$ and $\sigma$, which can be simplified as $\epsilon=1-\tr(\rho\sigma)$
when $\sigma$ is a pure state. In CRM shadow estimation, $\rho$ and $\sigma$ denote the system state and prior state, respectively; in HPFE, $\sigma$ also denotes the target pure state. 
 Furthermore, we denote by $\Tensor{\bbone}{2}$ and $\SWAP$ the identity operator and swap operator on $\Tensor{\caH}{2}$.

\section{\label{pro:3designVarProof}Proof of \pref{pro:3designVar}}

The equality in \eref{eq:3designVar} of \pref{pro:3designVar} was derived in \rcite{huang2020pred}, and 
the inequality follows from \lref{lem:var_exp_ineq} below.

\begin{lemma}\label{lem:var_exp_ineq}
	Suppose $Q\in\caL^\rmH(\caH)$ and $\rho\in \caD(\caH)$. Then 
	\begin{equation}\label{eq:var_exp_ineq}
		\tr\left(Q^2\right)+[\tr(\rho Q)]^2\ge2\tr\left(\rho Q^2\right).
	\end{equation}
\end{lemma}

\begin{proof}[Proof of \lref{lem:var_exp_ineq}]
Note that both sides in \eref{eq:var_exp_ineq} are invariant under  simultaneous unitary transformations of $Q$ and $\rho$. So  we can assume that $Q$ is diagonal in the computational basis without loss of generality. Then $\tr(\rho Q)$  and $\tr\left(\rho Q^2\right)$ only depend on the diagonal part of  $\rho$, so we can assume that $\rho$ is also diagonal without loss of generality. Now $Q$ and $\rho$ can be expressed as follows:
\begin{align}
	Q= \sum_i q_i |i\> \< i|,\quad \rho= \sum_i \gamma_i |i\> \< i|,
\end{align}
where $q_i$ ($\gamma_i$) is the eigenvalue of $Q$ ($\rho$) corresponding to the eigenvector $|i\>$. 

Next, we introduce
\begin{gather}
\varrho=\sum_{i,j}\sqrt{\gamma_i\gamma_j}\lsp |i\>\<j|,\quad  \quad|v\>=\sum_i q_i|i\>\otimes|i\>,
 \end{gather}
where $\varrho\in \caD(\caH)$ is a pure state  and $|v\>\in \Tensor{\caH}{2}$  is not necessarily normalized. Then  
\begin{equation}
\begin{aligned}
	\tr\left(Q^2\right)&=\sum_iq_i^2=\<v|v\>=    \<v|\bbone\otimes\bbone|v\>,\\
	\tr\left(\rho Q^2\right)&= \sum_i \gamma_i q_i^2=
	\<v|\varrho\otimes\bbone|v\>=\<v|\bbone\otimes\varrho|v\>,\\
	[\tr(\rho Q)]^2 &=\left(\sum_i \gamma_i q_i\right)^2 = \<v|\varrho\otimes\varrho|v\>.
\end{aligned}
\end{equation} 
    As a simple corollary, we have 
    \begin{align}
        \tr\left(Q^2\right)+\left[\tr(\rho Q)\right]^2-2\tr\left(\rho Q^2\right)&=\<v|(\bbone-\varrho)\otimes (\bbone-\varrho)|v\>\ge0,
    \end{align}
    where the inequality holds because $\bbone-\varrho$ is positive semidefinite given that $\varrho$ is a density operator. This result implies \eref{eq:var_exp_ineq} and completes the proof of \lref{lem:var_exp_ineq}. 
\end{proof}

\section{\label{SM:VRCRMproof}Proof of \lref{lem:VRCRM}}

\begin{proof}[Proof of \lref{lem:VRCRM}]
To achieve our goal, we first derive a general expression for the variance $\bbV_R(\hO)$ in CRM shadow estimation following \rcite{vermersch2024enhanced}. Note that the estimator $\hrho$ in  \eref{eq:CRMestimator} can also be expressed as follows: 
\begin{equation}
	\hrho=\sum_{\bfs}\hP_\rho(\bfs\mid U)\caM^{-1}\left(U^\dag|\bfs\>\<\bfs|U\right),
\end{equation}
where $\hP_\rho(\bfs\mid U)=R_{\bfs}/R$ denotes the experimentally estimated probability of 
outcome $\bfs$ after performing computational-basis measurements  on the state $U\rho U^{\dag}$. Here, $(R_\bfs)_{\bfs\in \{0,1\}^n}$ obeys the multinomial distribution $\text{Multinomial}\left(R;(p_\bfs)_{\bfs\in \{0,1\}^n}\right)$, where 
$p_\bfs=P_\rho(\bfs\mid U)=\<\bfs|U\rho U^\dag|\bfs\>$ denotes  the  theoretical probability of outcome $\bfs$ when performing a computational-basis measurement on the state $U\rho U^\dag$. Accordingly, the CRM estimator for $\tr(O \rho)$ reads
\begin{align}
	\hO &= \tr(O \hrho)=\sum_{\bfs}\left[\hP_\rho(\bfs\mid U)-P_\sigma(\bfs\mid U)\right]\tr\left[O\caM^{-1}\left(U^{\dag}|\bfs\>\<\bfs|U\right)\right]+\tr(O\sigma)\notag\\
	&=\sum_\bfs \left[\hP_\rho(\bfs\mid U)-P_\sigma(\bfs\mid U)\right]\left[\caM^{-1}(O)\right](U,\bfs)+\tr(O\sigma),
\end{align}
where $\left[\caM^{-1}(O)\right](U,\bfs):=\tr\left[O\caM^{-1}\left (U^{\dag}|\bfs\>\<\bfs|U\right)\right]=\tr\left[\caM^{-1}(O)U^{\dag}|\bfs\>\<\bfs|U\right]$. The variance of this estimator reads
\begin{align}
	\bbV_R(\hO)&=\bbV\left\{\sum_\bfs \left[\hP_\rho(\bfs\mid U)-P_\sigma(\bfs\mid U)\right]\left[\caM^{-1}(O)\right](U,\bfs)\right\}\notag\\
	&=\bbE_U\left(\sum_{\bfs,\bft}\bbE_\rmM\left\{ \left[\hP_\rho(\bfs\mid U)-P_\sigma(\bfs\mid U)\right]\left[\hP_\rho(\bft\mid U)-P_\sigma(\bft\mid U)\right]\right\}\left[\caM^{-1}(O)\right](U,\bfs)\left[\caM^{-1}(O)\right](U,\bft)\right)\notag\\
	&\quad-\left\{\tr\left[(\rho-\sigma)O\right]\right\}^2,\label{eq:VarianceCRMDer}
\end{align}
where $\bbE_\rmM$ represents the average over the multinomial distribution. 

According to well-known results on the multinomial distribution, we have
\begin{align}
	\bbE_\rmM(R_\bfs R_{\bft})&=\begin{cases}
		R(R-1)P_\rho(\bfs\mid U)P_\rho(\bft\mid U),&\bfs\neq\bft,\\
		R^2\left[P_\rho(\bfs\mid U)\right]^2+RP_\rho(\bfs\mid U)\left[1-P_\rho(\bfs\mid U)\right], &\bfs=\bft, \end{cases}\\
	\bbE_\rmM\left[\hP_\rho(\bfs\mid U)\hP_\rho(\bft\mid U)\right]&=P_\rho(\bfs\mid U)P_\rho(\bft\mid U)+\frac{\delta_{\bfs,\bft}P_\rho(\bfs\mid U)-P_\rho(\bfs\mid U)P_\rho(\bft\mid U)}{R}.\label{eq:MultJoint2ndMoment}
\end{align}
By virtue of \eref{eq:MultJoint2ndMoment}, the variance $\bbV_R(\hO)$ in  \eref{eq:VarianceCRMDer} can be expressed as follows \cite{vermersch2024enhanced}:
\begin{equation}
	\bbV_R(\hO)=\bbV_U\left[f_{\rho,\sigma}(U)\right]+\frac{\bbE_U\left[g_\rho (U)\right]}{R},
\end{equation}
where
\begin{equation}
\begin{aligned}
	f_{\rho,\sigma}(U)&:=\sum_{\bfs}\left[P_{\rho}(\bfs\mid U)-P_{\sigma}(\bfs\mid U)\right]\left[\caM^{-1}(O)\right](U,\bfs)=\sum_{\bfs}P_{\Delta}(\bfs\mid U)\left[\caM^{-1}(O)\right](U,\bfs),  \\
	g_{\rho}(U)&:=\sum_{\bfs}P_{\rho}(\bfs\mid U)\left\{\left[\caM^{-1}(O)\right](U,\bfs)\right\}^2-\left\{\sum_{\bfs}P_{\rho}(\bfs\mid U)\left[\caM^{-1}(O)\right](U,\bfs)\right\}^2,
\end{aligned}
\end{equation}
with $\Delta=\rho-\sigma$ and $P_{\Delta}(\bfs\mid U)=\<\bfs|U\Delta U^\dag|\bfs\>$.
To complete the proof of \lref{lem:VRCRM}, it remains to prove the following equalities:
\begin{align}
	\bbV_U[f_{\rho,\sigma}(U)]&=\bbV_*(O,\Delta),\quad 
	\bbE_U\left[g_\rho (U)\right]=\bbV(O,\rho)-\bbV_*(O,\rho).\label{eq:CRMvarfg}
\end{align}

By the definition of $f_{\rho,\sigma}(U)$ we can deduce that	
\begin{equation}
\begin{aligned}
	\bbE_U[f_{\rho ,\sigma}(U)]&=\bbE_U\left\{\sum_{\bfs}\left[ P_{\rho}(\bfs\mid U)-P_{\sigma}(\bfs\mid U) \right]\left[\caM^{-1}(O)\right](U,\bfs)\right\} =\tr[(\rho-\sigma) O]=\tr(\Delta O),\\
	\bbE_U\left\{[f_{\rho,\sigma}(U)]^2\right\}&=\sum_{\bfs,\bft}\bbE_U\left\{[\caM^{-1}(O)](U,\bfs)P_{\Delta}(\bfs\mid U)[\caM^{-1}(O)](U,\bft) P_{\Delta}(\bft\mid U)\right\}\\
	&=\sum_{\bfs,\bft}\bbE_{U}\left[\<\bfs|U \caM^{-1}(O) U^\dag|\bfs\> \<\bfs|U \Delta U^\dag|\bfs\> \<\bft|U \caM^{-1}(O) U^\dag|\bft\>\<\bft|U \Delta U^\dag|\bft\>\right] \\
	&=\tr\left\{\Omega(\caU)\Tensor{\left[\caM^{-1}(O)\otimes\Delta\right]}{2}\right\},
\end{aligned}
\end{equation}
where the last equality follows from the definition of $\Omega(\caU)$ in \eref{eq:CrossMoment}. 
	Therefore,
	\begin{equation}
		\bbV_U[f_{\rho,\sigma}(U)]=\bbE_U\left\{[f_{\rho,\sigma}(U)]^2\right\}-\left\{\bbE_U[f_{\rho ,\sigma}(U)]\right\}^2=\tr\left\{\Omega(\caU)\Tensor{\left[\caM^{-1}(O)\otimes\Delta\right]}{2}\right\}-\left[\tr(\Delta O)\right]^2=\bbV_*(O,\Delta),\label{eq:ProofTwirl2}
	\end{equation}
	which confirms the first equality in \eref{eq:CRMvarfg}. Here the last equality follows directly from the definition of $\bbV_*(O,\rho)$ in \eref{eq:V*}. Similarly, the second equality in \eref{eq:CRMvarfg} can be proved as follows:
	\begin{align}
		\bbE_U[g_\rho(U)]&=\bbE_U\left(\sum_{\bfs}P_{\rho}(\bfs\mid U)\left\{\left[\caM^{-1}(O)\right](U,\bfs)\right\}^2\right)-\left\{\bbE_U\sum_{\bfs}P_{\rho}(\bfs\mid U)\left[\caM^{-1}(O)\right](U,\bfs)\right\}^2\notag\\
		&\equad-\bbE_U\left(\left\{\sum_{\bfs}P_{\rho}(\bfs\mid U)\left[\caM^{-1}(O)\right](U,\bfs)\right\}^2\right)+\left\{\bbE_U\sum_{\bfs}P_{\rho}(\bfs\mid U)\left[\caM^{-1}(O)\right](U,\bfs)\right\}^2\notag\\
		&=\bbV(O,\rho)-\tr\left\{\Omega(\caU)\Tensor{\left[\caM^{-1}(O)\otimes\rho\right]}{2}\right\}+[\tr(\rho O)]^2=\bbV(O,\rho)-\bbV_*(O,\rho),
	\end{align}
	where we have employed the definitions of $\Omega(\caU)$ and $\bbV_*(O,\rho)$ in \eqsref{eq:CrossMoment}{eq:V*}, respectively. This observation completes the proof of \lref{lem:VRCRM}. 	
\end{proof}

\section{Average performance of CRM shadow estimation based on 2-designs}
In this section, we clarify the average  performance of CRM shadow estimation based on 2-designs. For  $Q\in \caL^\rmH(\caH)$, define the following ensemble of observables:
\begin{equation}
\caE(Q):=\left\{UQU^{\dag}\,|\, U\in \rmU(\caH)\right\}\label{eq:ObEnsemble},
\end{equation}
which inherits the Haar measure on the unitary group $\rmU(\caH)$. We then define $\bbV(\caE(Q),\rho)$, $\bbV_*(\caE(Q),\rho)$, and $\bbV_*(\caE(Q),\Delta)$ as the average variances over the  ensemble $\caE(Q)$, where $\Delta=\rho-\sigma$, and $\rho$ and $\sigma$ denote the system state and prior state, respectively, as in the main text.

\begin{proposition}\label{pro:average}
	Suppose $\caU$ is a unitary 2-design on $\caH$ and $O$ is a traceless observable on $\caH$. Then the average variances $\bbV(\caE(O),\rho)$, $\bbV_*(\caE(O),\rho)$, and $\bbV_*(\caE(O),\Delta)$ in $\CRM(\caU,\sigma,R)$ read
	\begin{align}
	\bbV(\caE(O),\rho) &= \left[1+\frac{d-\tr(\rho^2)}{d^2-1}\right]\rVert O\rVert_2^2, \label{eq:VOrhoAvg}\\
	\bbV_*(\caE(O),\rho) &= \frac{d\tr(\rho^2)-1}{d^2-1}\lVert O \rVert_2^2<\frac{\lVert O \rVert_2^2}{d},  \label{eq:V*OrhoAvg}\\
	\bbV_*(\caE(O),\Delta)&=\frac{d \lVert\Delta\rVert_2^2\lVert O \rVert_2^2}{d^2-1}. \label{eq:V*EODelta}
	\end{align}
\end{proposition}
The results in \eqsref{eq:VOrhoAvg}{eq:V*OrhoAvg} were derived in \rcite{chen2024nonstab}. On average, $\bbV_*(\caE(O),\rho)$ is already exponentially smaller than $\bbV(\caE(O),\rho)$, which means THR shadow estimation is efficient on average. Nevertheless, incorporating an accurate prior state can further reduce the variance  and enhance the overall performance significantly. Before proving \pref{pro:average}, we introduce an auxiliary lemma proved in \rcite{chen2024nonstab}.
\begin{lemma}\label{lem:UnitaryAvg}
	Suppose $\caU$ is a unitary 2-design on  $\caH$ and $A,B\in\caL(\caH)$. Then 
	\begin{gather}
	\bbE_{U\sim\caU}U A U^\dag=\frac{\tr(A)}{d}\bbone,  \label{eq:TwirlAvg} \\
	\bbE_{U\sim\caU} \left[\left(U A U^\dag\right)\otimes \left(U A U^\dag\right)^\dag\right] = \frac{d|\tr(A)|^2 -\lVert A \rVert_2^2}{d(d^2-1)}\Tensor{\bbone}{2} + \frac{d\lVert A\rVert_2^2-|\tr(A)|^2}{d(d^2-1)}\SWAP,\label{eq:TensorAvg}\\
	\bbE_{U\sim\caU} \left|\tr\left(UAU^\dag B\right)\right|^2=\frac{d|\tr(A)|^2|\tr(B)|^2+d\lVert A\rVert_2^2\lVert B\rVert_2^2-|\tr(A)|^2\lVert B\rVert_2^2-\lVert A\rVert_2^2|\tr(B)|^2}{d(d^2-1)}. \label{eq:TraceAvg}
	\end{gather}
\end{lemma}

\begin{proof}[Proof of \pref{pro:average}]
	\Eqsref{eq:VOrhoAvg}{eq:V*OrhoAvg} were derived in \rcite{chen2024nonstab}, so it remains to prove \eref{eq:V*EODelta}. 
	
	By assumption $\caU$ is a unitary 2-design  and $O$ is traceless, so the reconstruction map $\caM^{-1}(\cdot)$  has a simple form:  $\caM^{-1}(O)=(d+1)O$ \cite{huang2020pred}. In conjunction with the definition of $\bbV_*(O,\tau)$ in  \eref{eq:V*} we can deduce that	
	\begin{align}
	&\bbV_*(\caE(O),\Delta)= \bbE_{W\sim \haar}\left\{(d+1)^2\tr\left[\Omega(\caU)\left(WOW^{\dag}\otimes\Delta\otimes WOW^{\dag}\otimes\Delta\right)\right]-\left[\tr\left(WOW^{\dag}\Delta\right)\right]^2\right\}\notag\\
	&=(d+1)^2\bbE_{W\sim\haar}\bbE_{U\sim \caU}\sum_{\bfs,\bft}\tr\left[(|\bfs\>\<\bfs|\otimes|\bft\>\<\bft|)\left(UWOW^{\dag}U^{\dag}\right)^{\otimes 2}\right]\tr\left[(|\bfs\>\<\bfs|\otimes |\bft\>\<\bft|)\left(U\Delta U^\dag\right)^{\otimes 2}\right]\notag\\
	&\equad-\frac{\lVert O \rVert_2^2 \lVert\Delta\rVert_2^2}{d^2-1},
	\end{align}
	where the second equality follows from the definition of $\Omega(\caU)$ in \eref{eq:CrossMoment}. Since $O$ and $\Delta$ are traceless, according to \lref{lem:UnitaryAvg}, we have
	\begin{equation}
	\begin{aligned}
	\bbE_{W\sim \haar}\left(UWOW^{\dag}U^\dag\right)^{\otimes 2}&=\frac{\lVert O \rVert_2^2}{d(d^2-1)}U^{\otimes 2}(d\lsp \swap-\bbone)U^{\dag\otimes 2}=\frac{\lVert O \rVert_2^2}{d(d^2-1)}(d\lsp \swap-\bbone),\\
	\bbE_{U\sim \caU}\left(U\Delta U^\dag\right)^{\otimes 2}&=\frac{\lVert \Delta \rVert_2^2}{d(d^2-1)}(d\lsp \swap-\bbone).	
\end{aligned}
	\end{equation}
	Based on these results, the above expression of $\bbV_*(\caE(O),\Delta)$ can be simplified as follows:
	\begin{align}
	&\bbV_*(\caE(O),\Delta)=(d+1)^2\frac{\lVert O \rVert_2^2}{d(d^2-1)}\bbE_{U\sim\caU}\sum_{\bfs,\bft}\Bigl\{
	\left(d\delta_{\bfs,\bft}-1\right)\tr\left[(|\bfs\>\<\bfs|\otimes|\bft\>\<\bft|)\left(U\Delta U^\dag\right)^{\otimes 2}\right]\Bigr\}-\frac{\lVert O \rVert_2^2 \lVert\Delta\rVert_2^2}{d^2-1}\notag\\
	&		
	=\frac{\lVert O \rVert_2^2\lVert \Delta \rVert_2^2}{d^2(d-1)^2}	\sum_{\bfs,\bft}\left(d\delta_{\bfs,\bft}-1\right)^2-\frac{\lVert O \rVert_2^2 \lVert\Delta\rVert_2^2}{d^2-1}
	=\left(\frac{1}{d-1}-\frac{1}{d^2-1}\right)\lVert O \rVert_2^2 \lVert\Delta\rVert_2^2
	=\frac{d}{d^2-1}\lVert O \rVert_2^2 \lVert\Delta\rVert_2^2,
	\end{align}
	which confirms \eref{eq:V*EODelta} and completes the proof of \pref{pro:average}. 	
\end{proof}

\section{\label{SM:4_design}Variances in CRM shadow estimation based on 4-design measurements}

In this section, we derive  exact formulas for the 4th cross moment operator and the variances $\bbV_*(O,\Delta)$ and $\bbV_*(O,\rho)$ in  CRM shadow estimation  based on 4-design  measurements (see \lsref{lem:CrossMom4design} and \ref{lem:V*4design} below). On this basis we prove  \thref{thm:4designVstar} in the main text.

\subsection{The 4th cross moment operator}

Given an ensemble  $\caU$ of unitary transformations on $\caH$, according to the definition in \eref{eq:CrossMoment}, the 4th cross moment operator $\Omega(\caU)$ reads 
\begin{align}\label{eq:OmegaDecom}
	\Omega(\caU)&=\sum_{\bfs,\bft}\bbE_{U\sim\caU}\dagtensor{U}{4}\left[\Tensor{\left(|\bfs\>\<\bfs|\right)}{2}\otimes\Tensor{\left(|\bft\>\<\bft|\right)}{2}\right]\Tensor{U}{4}=\Omega_1(\caU)+\Omega_2(\caU),  
\end{align}
where 
\begin{equation}\label{eq:omega12}
	\begin{aligned}
		\Omega_1(\caU)&:=\sum_{\bfs}\bbE_{U\sim\caU}\dagtensor{U}{4}\left[\Tensor{\left(|\bfs\>\<\bfs|\right)}{4}\right]\Tensor{U}{4},\\
		\Omega_2(\caU)&:=\sum_{\bfs\neq\bft}\bbE_{U\sim\caU}\dagtensor{U}{4}\left[\Tensor{\left(|\bfs\>\<\bfs|\right)}{2}\otimes\Tensor{\left(|\bft\>\<\bft|\right)}{2}\right]\Tensor{U}{4}.
	\end{aligned}
\end{equation}

Now, suppose $\caU$ is a unitary 4-design. To simplify the expressions of $\Omega_1(\caU)$, $\Omega_2(\caU)$, and $\Omega(\caU)$, we need to introduce some auxiliary results. According to Schur-Weyl duality, the tensor power $\Tensor{\caH}{4}$ can be decomposed as follows:
\begin{align}
	\Tensor{\caH}{4}=\bigoplus_\lambda\caH_\lambda=\bigoplus_\lambda\caW_\lambda\otimes \caS_\lambda,
\end{align}
where $\lambda$ denotes a nonincreasing partition of 4 into no more than $d$ parts, while $\caW_\lambda$ and $\caS_\lambda$ denote the corresponding  Weyl module and  Specht module, which carry  irreducible representations of the unitary group $\rmU(\caH)$ and  permutation group $S_4$, respectively. Let $D_\lambda=\dim(\caW_\lambda)$ and $d_\lambda=\dim(\caS_\lambda)$. The projector onto $\caH_\lambda$ can be expressed as follows:
\begin{equation}
	P_\lambda=\frac{d_\lambda}{4!}\sum_{\pi\in S_{4}}\chi^{\lambda}(\pi)W_{\pi},\label{eq:projection_irrep}
\end{equation}
where $W_{\pi}$ denotes the permutation operator on $\Tensor{\caH}{4}$ associated with the permutation $\pi$, and $\chi^{\lambda}(\pi)$ denotes the character of $\pi$ corresponding to the representation $\lambda$.

\begin{lemma}\label{lem:prep}
	Suppose $\caU$ is a unitary 4-design on $\caH$ and  $Q\in\caL^\rmH(\Tensor{\caH}{4})$. Then
	\begin{equation}
		\bbE_{U\sim\caU}\tr\left(\Tensor{U}{4}Q\dagtensor{U}{4}\right)=\frac{1}{4!}\sum_{\pi\in S_{4}}\tr\left(QW_{\pi}\right)W_{\pi^{-1}}\sum_\lambda\frac{d_\lambda}{D_\lambda}P_\lambda.
	\end{equation}
\end{lemma}
We omit the proof of \lref{lem:prep} here; interested readers may refer to \rscite{collins2006inte,roth2018recover,zhou2023perf} for rigorous proofs and further discussions.

Next, it is convenient to introduce the special subgroup $G=\{e,(12),(34),(12)(34)\}$  of $S_4$, where $e$ denotes the identity permutation.  Denote by $P_{G}=\sum_{\pi\in G}W_{\pi}/4$  the projector onto the
subspace of $\caH^{\otimes 4}$ that carries the trivial representation of 
$G$. Note that $P_G$ commutes with each projector $P_\lambda$  in \eref{eq:projection_irrep}. 
\begin{lemma}\label{lem:CrossMom4design}
	Suppose $\caU$ is a unitary 4-design on $\caH$. Then  $\Omega_1(\caU)$, $\Omega_2(\caU)$, and the cross moment operator $\Omega(\caU)$ read
	\begin{align}
		\Omega_1(\caU)&=\frac{24P_{[4]}}{(d+1)(d+2)(d+3)},\quad 
		\Omega_2(\caU)=\frac{1}{6}\left(d^2-d\right)\sum_\lambda\frac{d_\lambda}{D_\lambda}P_{G}P_\lambda,  \label{eq:Omega12Usim}\\
		\Omega(\caU)&=\Omega_1(\caU)+\Omega_2(\caU)=
		\frac{24P_{[4]}}{(d+1)(d+2)(d+3)}+\frac{1}{6}\left(d^2-d\right)\sum_\lambda\frac{d_\lambda}{D_\lambda}P_{G}P_\lambda.  \label{eq:CrossMom4design}
	\end{align}
\end{lemma}

\begin{proof}[Proof of \lref{lem:CrossMom4design}]
	According to \eref{eq:omega12} we have
	\begin{equation}
		\Omega_1(\caU)=\sum_{\bfs}\bbE_{U\sim\caU}\dagtensor{U}{4}\left[\Tensor{\left(|\bfs\>\<\bfs|\right)}{4}\right]\Tensor{U}{4}=\frac{24P_{[4]}}{(d+1)(d+2)(d+3)},
	\end{equation}
	where the second equality follows from Schur-Weyl duality given that $\caU$ is a unitary 4-design by assumption. In addition, by virtue of  \lref{lem:prep} we can deduce that
	\begin{align}
		\Omega_2(\caU)&=\sum_{\bfs\neq\bft}\bbE_{U\sim\caU}\dagtensor{U}{4}\left[\Tensor{\left(|\bfs\>\<\bfs|\right)}{2}\otimes\Tensor{\left(|\bft\>\<\bft|\right)}{2}\right]\Tensor{U}{4}\nonumber\\
		&=\frac{1}{4!}\sum_{\pi\in S_{4}}\sum_{\bfs \neq \bft}\tr\left[\Tensor{(|\bfs\>\<\bfs|)}{2}\otimes\Tensor{(|\bft\>\<\bft|)}{2}W_{\pi^{-1}}\right]W_{\pi}\sum_\lambda\frac{d_\lambda}{D_\lambda}P_\lambda \nonumber\\
		& =\frac{1}{4!}\left(d^2-d\right)\sum_{\pi\in G}W_{\pi}\sum_\lambda\frac{d_\lambda}{D_\lambda}P_\lambda=\frac{1}{6}\left(d^2-d\right)\sum_\lambda\frac{d_\lambda}{D_\lambda}P_{G}P_\lambda,  \label{eq:omega2}
	\end{align}
	where the third equality holds because 
	$\tr\left[\Tensor{(|\bfs\>\<\bfs|)}{2}\otimes\Tensor{(|\bft\>\<\bft|)}{2}W_{\pi^{-1}}\right]$ is equal to 1 (0) when $\pi\in G$ ($\pi\notin G$), while the last equality follows from the definition of $P_G$. The above two equations together confirm \eref{eq:Omega12Usim}. Finally, \eref{eq:CrossMom4design} follows from \eqsref{eq:OmegaDecom}{eq:Omega12Usim}, which completes the proof of \lref{lem:CrossMom4design}. 
\end{proof}

\subsection{Exact formula for the variance $\bbV_*(O,\Delta)$}

\begin{lemma}\label{lem:V*4design}
	Suppose $\caU$ is a unitary 4-design on $\caH$, $O$ is a traceless observable on $\caH$, and $\rho\in \caD(\caH)$ is the system state. Then, the variances $\bbV_*(O,\Delta)$ and $\bbV_*(O,\rho)$ in $\CRM(\caU,\sigma,R)$ read
	\begin{align}
		\bbV_*(O,\Delta)&=\frac{d^2+3d+4}{d(d+2)(d+3)}[\tr(\Delta O)]^2+\frac{2(d+1)^2}{d(d+2)(d+3)}\tr\left(\Delta^2 O^2\right)\notag\\
		&\equad+\frac{d+1}{d(d+3)}\tr\left(\Delta^2\right)\tr\left(O^2\right)-\frac{2(d+1)}{d(d+2)(d+3)}\tr(\Delta O \Delta O).\label{eq:V*4design}\\
        \bbV_*(O,\rho)&=\frac{d^2+3d+4}{d(d+2)(d+3)}[\tr(\rho O)]^2+\frac{2(d+1)^2}{d(d+2)(d+3)}\tr\left(\rho^2 O^2\right)\notag\\
		&\equad+\frac{d+1}{d(d+3)}\tr\left(\rho^2\right)\tr\left(O^2\right)-\frac{2(d+1)}{d(d+2)(d+3)}\tr(\rho O \rho O)-\frac{4(d+1)}{d(d+2)(d+3)}\tr\left(\rho O^2\right).\label{eq:V*_rho_4design}
	\end{align}
\end{lemma}

\begin{proof}[Proof of \lref{lem:V*4design}]
	By assumption $\caU$ is a unitary 4-design  and $O$ is traceless, so the reconstruction map $\caM^{-1}(\cdot)$  has a simple form:  $\caM^{-1}(O)=(d+1)O$ \cite{huang2020pred}. In conjunction with \eref{eq:V*} we can deduce that
	\begin{align}
		\bbV_*(O,\Delta)&=\tr\left\{\Omega(\caU)\Tensor{\left[\caM^{-1}(O)\otimes\Delta\right]}{2}\right\}-\left[\tr(O\Delta)\right]^2=(d+1)^2\tr\left[\Omega(\caU)\Tensor{\left(O\otimes\Delta\right)}{2}\right]-\left[\tr(O\Delta)\right]^2\notag\\
		&=(d+1)^2\tr\left[\Omega_1(\caU)\Tensor{\left(O\otimes\Delta\right)}{2}\right]+(d+1)^2\tr\left[\Omega_2(\caU)\Tensor{\left(O\otimes\Delta\right)}{2}\right]-\left[\tr(O\Delta)\right]^2,\label{eq:4design_3terms}
	\end{align}
	where $\Omega(\caU)$ is defined in \eref{eq:CrossMoment} and the last equality holds because  $\Omega(\caU)=\Omega_1(\caU)+\Omega_2(\caU)$ according to \eref{eq:OmegaDecom}. 
	Thanks to \lref{lem:CrossMom4design},
	the first term in \eref{eq:4design_3terms} can be evaluated as follows:
	\begin{align}
		&(d+1)^2\tr\left[\Omega_1(\caU)\Tensor{\left(O\otimes\Delta\right)}{2}\right]
		=\frac{24(d+1)}{(d+2)(d+3)}\tr\left[P_{[4]}\Tensor{\left(O\otimes\Delta\right)}{2}\right]	
		=\frac{d+1}{(d+2)(d+3)}\sum_{\pi \in S_4}\tr\left[W_{\pi}\Tensor{\left(O\otimes\Delta\right)}{2}\right]\notag \\
		&=\frac{d+1}{(d+2)(d+3)}\left\{2\left[\tr(\Delta O) \right]^2+\tr\left(\Delta^2\right)\tr\left(O^2\right)+2\tr(\Delta O \Delta O)+4\tr\left(\Delta^2O^2\right)\right\}.\label{eq:firstterm_exact}
	\end{align}
	Similarly, the second term in \eref{eq:4design_3terms} can be evaluated as follows:
	\begin{align}
		&(d+1)^2\tr\left[\Omega_2(\caU)\Tensor{\left(O\otimes\Delta\right)}{2}\right]=\frac{1}{6}\left(d^2-d\right)(d+1)^2\sum_\lambda\frac{d_\lambda}{D_\lambda}\tr\left[\left(O \otimes \Delta \otimes O\otimes \Delta\right) P_{G}P_\lambda\right]\notag\\
		&=\frac{1}{6}\left(d^2-d\right)(d+1)^2\sum_\lambda\frac{d_\lambda^2}{D_\lambda}\frac{1}{4\times4!}\sum_{\zeta \in G, \pi\in S_{4}}\chi^{\lambda}(\pi)\tr\left[\left(O \otimes \Delta \otimes O\otimes \Delta\right) W_{\zeta \pi}\right],\label{eq:sum_twirling}
	\end{align}
	where the second equality follows from the definition $P_{G}=\sum_{\zeta\in G}W_{\zeta}/4$ and the formula for $P_\lambda$ in  \eref{eq:projection_irrep}. In addition, the dimensions $d_\lambda, D_\lambda$ and characters $\chi^{\lambda}(\pi)$ featured in the above equation can be found in \rcite{zhu2016clifford}. After careful calculations of individual terms in \eref{eq:sum_twirling} we can derive the following result:
	\begin{align}
		(d+1)^2\tr\left[\Omega_2(\caU)\Tensor{\left(O\otimes\Delta \right)}{2}\right]&=\frac{(d+1)(d^2+3d+4)}{d(d+2)(d+3)}\left[\tr(\Delta O)\right]^2-\frac{2(d+1)(d-1)}{d(d+2)(d+3)}\tr\left(\Delta^2 O^2\right)\notag\\
		&\equad+\frac{2(d+1)}{d(d+2)(d+3)}\tr\left(\Delta^2\right)\tr\left(O^2\right)-\frac{2(d+1)^2}{d(d+2)(d+3)}\tr(\Delta O \Delta O)\label{eq:twirl_2nd}.
	\end{align}
Combining \eqssref{eq:4design_3terms}{eq:firstterm_exact}{eq:twirl_2nd}, we get
	\begin{align}
		\bbV_*(O,\Delta)&=\frac{d^2+3d+4}{d(d+2)(d+3)}\left[\tr(\Delta O)\right]^2+\frac{2(d+1)^2}{d(d+2)(d+3)}\tr\left(\Delta^2 O^2\right)\notag\\
		&\equad+\frac{d+1}{d(d+3)}\tr\left(\Delta^2\right)\tr\left(O^2\right)-\frac{2(d+1)}{d(d+2)(d+3)}\tr(\Delta O \Delta O),
	\end{align}
	which confirms \eref{eq:V*4design}.

Finally, \eref{eq:V*_rho_4design} can be proved in a similar way, and the  additional term compared with \eref{eq:V*4design} is tied to the fact that $\tr(\rho) \neq 0$. This observation completes the proof of \lref{lem:V*4design}.
\end{proof}

\subsection{Proof of \thref{thm:4designVstar}}
Before proving \thref{thm:4designVstar}, we need to introduce another auxiliary lemma.
\begin{lemma}\label{lem:ABAB}
	Suppose $A,B\in\caL^\rmH(\caH)$ are traceless Hermitian operators on  $\caH$, and  $W_{\pi}$ is the permutation operator  on $\Tensor{\caH}{4}$ associated with  some permutation  $\pi$ in $S_4$. Then 
	\begin{gather}
		[\tr(AB)]^2\le\lVert A\rVert_2^2\lVert B\rVert_2^2,\quad 
		\left|\tr(ABAB)\right|\le\tr\left(A^2B^2\right)\le \lVert A\rVert_2^2\lVert B\rVert^2\le \lVert A\rVert_2^2\lVert B\rVert_2^2, \label{eq:ABAB} \\ 	
		\left|\tr\left[W_{\pi}\Tensor{(A\otimes B)}{2}\right]\right|\le\lVert A\rVert_2^2\lVert B\rVert_2^2.\label{eq:ABABRpi}
	\end{gather}
\end{lemma}

\begin{proof}[Proof of \lref{lem:ABAB}]
	The first two inequalities in \eref{eq:ABAB} follow from the operator Cauchy-Schwarz inequality, the third inequality follows from the operator H\"older  inequality, and the last inequality holds because $\lVert B\rVert\le \lVert B\rVert_2$. 
	
	Next, we turn to \eref{eq:ABABRpi}. Recall that any permutation can be expressed as a product of disjoint cycles. Since $A,B$ are traceless by assumption, $\tr\left[W_{\pi}\Tensor{(A\otimes B)}{2}\right]$ vanishes whenever $\pi$ contains a trivial cycle (or has a fixed point). If $\pi$ does not contain any trivial cycle, then  $\pi$ is composed of either one cycle of length 4 or  two cycles of length 2. In the former case, $\tr\left[W_{\pi}\Tensor{(A\otimes B)}{2}\right]$ is equal to either $\tr(ABAB)$ or $\tr\left(A^2B^2\right)$. In the latter case, $\tr\left[W_{\pi}\Tensor{(A\otimes B)}{2}\right]$ is equal to either $[\tr(AB)]^2$ or $\tr(A^2)\tr(B^2)=\lVert A\rVert
    _2^2\lVert B\rVert_2^2$. So \eref{eq:ABABRpi} is a simple corollary of \eref{eq:ABAB}. 
\end{proof}

\begin{proof}[Proof of \thref{thm:4designVstar}]
	By virtue of \Lsref{lem:V*4design}{lem:ABAB} we can deduce that 
	\begin{align}
		\bbV_*(O,\Delta)&=\frac{d^2+3d+4}{d(d+2)(d+3)}\left[\tr(\Delta O)\right]^2+\frac{2(d+1)^2}{d(d+2)(d+3)}\tr\left(\Delta^2 O^2\right)\notag\\
		&\equad+\frac{d+1}{d(d+3)}\tr\left(\Delta^2\right)\tr\left(O^2\right)-\frac{2(d+1)}{d(d+2)(d+3)}\tr(\Delta O \Delta O)\nonumber\\
		&\ge\frac{2(d+1)^2}{d(d+2)(d+3)}\tr\left(\Delta^2O^2\right)+\frac{d+1}{d(d+3)}\tr\left(\Delta^2\right)\tr\left(O^2\right)-\frac{2(d+1)}{d(d+2)(d+3)}\tr\left(\Delta ^2  O^2\right)\notag\\
		&\ge \frac{d+1}{d(d+3)}\tr\left(\Delta^2\right)\tr\left(O^2\right)\ge\frac{ \lVert\Delta\rVert_2^2\lVert O \rVert_2^2}{d+2},
	\end{align}
	which confirms the first inequality in \eref{eq:4duppline1}. On the other hand, we have
	\begin{align}
		\bbV_*(O,\Delta)&\le 
		\left[\frac{d^2+3d+4}{d(d+2)(d+3)}+\frac{2(d+1)^2}{d(d+2)(d+3)}+\frac{d+1}{d(d+3)}+\frac{2(d+1)}{d(d+2)(d+3)}\right]	\lVert O \rVert_2^2 \lVert\Delta\rVert_2^2\nonumber\\
		&=\frac{2(2d^2+6d+5)}{d(d+2)(d+3)}	\lVert O \rVert_2^2 \lVert\Delta\rVert_2^2 \le\frac{4}{d}\lVert O \rVert_2^2 \lVert\Delta\rVert_2^2\label{eq:V*4designUB},
	\end{align}	
	which confirms the second inequality in \eref{eq:4duppline1}. The last inequality in \eref{eq:4duppline1} follows from \pref{pro:Delta12NormIFsim}, which completes the proof of 
	\thref{thm:4designVstar}.	
	\end{proof}

\section{Noise channels} 
\subsection{Depolarizing and Pauli channels}
In this section, we review some noise channels that are relevant to the current study. A depolarizing channel $\scrN$ is defined via the following action:
\begin{align}\label{eq:Depolarizing}
	\scrN(\sigma)=(1-p)\sigma+p\frac{\bbone}{d}, \quad \sigma\in \caD(\caH),
\end{align}
where $0\le p\le 1$ quantifies the strength of  noise. Note that $\scrN(\sigma)$ always commutes with $\sigma$.

Recall that $\bcaP_n=\Tensor{\{I,X,Y,Z\}}{n}$ denotes the set of $n$-qubit Pauli operators. We can label the elements in $\bcaP_n$ using integers $0,1,\ldots, d^2-1$, which leads to the expression $\bcaP_n=\{P_i\}_{i=0}^{d^2-1}$. By convention $P_0=\bbone=\Tensor{I}{n}$, but the specific labeling of other elements is not essential to us. An $n$-qubit Pauli channel $\scrP$ has the form:
\begin{align}
	\scrP(\sigma) =\sum_{P\in\bcaP_n}p(P)P\sigma P 
	=\left(1-\sum_{i=1}^{d^2-1}p_i\right)\sigma+\sum_{i=1}^{d^2-1}p_i P_i \sigma P_i, \quad  \sigma\in \caD(\caH), \label{eq:PauliChannel}
\end{align}
where $p_i\geq 0$ for $i=0, 1,\ldots, d^2-1$ and $\sum_{i=0}^{d^2-1}p_i=1$. Note that $\scrP$ is completely characterized by the vector $\bfp=(p_1, p_2, \ldots, p_{d^2-1})$. Let 
\begin{align}
	\lVert\bfp\rVert_a=\left(\sum_{i=1}^{d^2-1}p_i^a\right)^{1/a}, \quad a\geq 1.
\end{align}
Then the above Pauli channel can also be expressed as follows:
\begin{align}
	\scrP(\sigma)=\left(1-\lVert\bfp\rVert_1\right)\sigma+\sum_{i=1}^{d^2-1}p_i P_i \sigma P_i.  \label{eq:PauliChannel2}
\end{align}
If there is only one nonzero element $p_k$ in $\bfp$, then the channel reduces to a single-error Pauli channel:
\begin{equation}
    \scrP(\sigma) = (1-p_k)\sigma+p_k P_k \sigma P_k.\label{eq:def_SingleError}
\end{equation}

In general, $\scrP(\sigma)$ does not necessarily commute with $\sigma$. Nevertheless, if  $\sigma$ is a stabilizer state, then $P_i \sigma P_i$ is also a stabilizer state for each $i=0,1,\ldots, d^2-1$; moreover, each $P_i \sigma P_i$ is either identical to $\sigma$ or orthogonal to $\sigma$, so  $\scrP(\sigma)$ commutes with $\sigma$.

\subsection{Sampling random Pauli channels}\label{sec:DescriptRandomPauli}
In this section, we explain the numerical simulation procedure for sampling a random Pauli channel as employed in the main text. We begin by generating a random real vector $\bfp = (p_1, p_2, \ldots, p_{d^2-1})$ of dimension $d^2 - 1$, where each entry $p_i$ is independently sampled from the uniform distribution over $[0,1]$. Next, we choose a parameter $\beta$ with $0\leq \beta \leq 1$ and  renormalize the vector  such that $\lVert\bfp\rVert_1 = \beta$. Then the resulting vector $\bfp$ can be used to define a Pauli channel according to \eref{eq:PauliChannel2}. On average we have $\bepsilon\sim \lVert\bfp\rVert_1=\beta$ by  \eref{eq:MeanIFpauli} in \sref{SM:NormsInfid}. By adjusting the parameter $\beta$ we can control the average infidelity $\bepsilon$ or generate a desired distribution.

\subsection{Sampling random local rotations}\label{sec:DescriptRandomRotation}

A rotation operator for a single qubit is a unitary operator of the form
\begin{equation}
	 \exp\left(-\rmi\frac{\theta}{2} \bmv \cdot \bm{\sigma} \right), 
\end{equation}
where $\bm{\sigma}$ denotes the vector composed of the three Pauli operators, $\bmv$ is a normalized real unit vector in dimension 3 that denotes the rotation axis, and $\theta\in [0,2\pi)$ denotes the rotation angle. A local rotation operator on $\caH$   is a unitary operator  that can be expressed as a tensor product of $n$ rotation operators for a single qubit:
\begin{equation}\label{eq:Utheta}
	U=\bigotimes_k U_k,  \quad U_k= \exp\left(-\rmi\frac{\theta_k}{2} \bmv_k \cdot \bm{\sigma} \right)=\cos\left(\frac{\theta_k}{2} \right)I-\rmi\sin\left(\frac{\theta_k}{2} \right)(\bmv_k \cdot \bm{\sigma}), 
\end{equation}
where $\bmv_k$ denotes the rotation axis of qubit $k$ and $\theta_k\in [0,2\pi)$ denotes the corresponding rotation angle. Here we are interested in local rotation operators because they can model local coherent noise.

The following proposition clarifies the average gate infidelity of the above rotation operator, which is instructive for sampling random unitary operators in numerical simulations.

\begin{proposition}\label{pro:SingeRotationAverage}
    Suppose $U$ is the rotation operator presented in \eref{eq:Utheta}, and $|\phi\>$ is a Haar-random pure state in $\caH$. Then the average gate infidelity $\bepsilon$ of $U$ reads 
    \begin{equation}
     \bepsilon=   1-\bbE_{|\phi\>\sim \haar}\tr\left[|\phi\>\<\phi| U|\phi\>\<\phi| U^{\dag}\right]=\frac{d\left(1-\xi\right)}{d+1},\label{eq:SingleRotationAverage}
    \end{equation}
    where $\xi=\prod_k\left[\cos(\theta_k/2)\right]^2$.
\end{proposition}

By virtue of \eref{eq:Utheta} and \pref{pro:SingeRotationAverage} we can introduce a method for sampling a random rotation operator with a given average gate infidelity $\bepsilon$. First, we randomly choose $n$ rotation axes $\bmv_1, \bmv_2,\ldots, \bmv_n$ uniformly and independently  on the Bloch sphere.  Then, we generate a random real vector $\bfh=(h_1,h_2,\ldots,h_n)\in \bbR^n$, where each entry is sampled from the uniform distribution over $(0,1]$ and renormalize the vector such that
\begin{equation}
	\prod_k h_k = \xi=1-\frac{(d+1)\bepsilon}{d}.
\end{equation}
Finally, the rotation angle for each qubit 
can be chosen as follows:
\begin{equation}
    \theta_k = \arccos(2h_k-1),\quad k=1,2,\ldots, n.
\end{equation}
By virtue of the axes $\bmv_k$ and angles $\theta_k$ for $k=1,2,\ldots, n$ we can construct a local rotation operator as in \eref{eq:Utheta}. Based on this method, it is also straightforward to sample a random rotation operator such that  $1/\bepsilon$ is uniformly distributed in a certain interval.

\begin{proof}[Proof of \pref{pro:SingeRotationAverage}]
    According to Schur–Weyl duality, we have 
	\begin{equation}
		\bbE_{|\phi\>\sim\haar}(|\phi\>\<\phi|\otimes|\phi\>\<\phi|)=\frac{\bbone+\SWAP}{d(d+1)}.
	\end{equation}
    Therefore,
    \begin{equation}
        \bbE_{|\phi\>\sim \haar}\tr\left[|\phi\>\<\phi| U|\phi\>\<\phi| U^{\dag}\right]=\bbE_{|\phi\>\sim \haar}\tr\left[\left(|\phi\>\<\phi| \otimes|\phi\>\<\phi|\right) \left(U\otimes U^{\dag}\right)\right]=\frac{d+|\tr(U)|^2}{d(d+1)}=\frac{1+d\xi}{d+1},
    \end{equation}
which implies \eref{eq:SingleRotationAverage}. Here the last equality holds because $\tr(U)=d\prod_k\cos(\theta_k/2)=d\sqrt{\xi}$.
\end{proof}

\section{Relations between $\lVert\Delta\rVert_2^2$, $\lVert\Delta\rVert_1^2$, and the infidelity $\epsilon$}\label{SM:NormsInfid}

\subsection{Basic relations}

According to the discussions in the main text, the deviation $\Delta=\rho-\sigma$ of  the system state $\rho$ from the prior  state $\sigma$ plays a crucial role in CRM shadow estimation. 
In this section, we focus on the case in which 
 $\sigma$ is a pure state and
clarify the relations between $\lVert\Delta\rVert_2^2$, $\lVert\Delta\rVert_1^2$, and the infidelity $\epsilon$, which are crucial to understanding the performance of HPFE. 
The following three propositions are proved in \sref{SM:NormsInfidProof}, in which \psref{pro:Delta12NormIF} and \ref{pro:pchannel_haar}  strengthen
\psref{pro:Delta12NormIFsim} and \ref{pro:pchannel_haarSim} in the main text.

\begin{proposition}\label{pro:Delta12NormIFgen}
	Suppose $\sigma,\rho\in \caD(\caH)$, $\Delta=\rho-\sigma$, and $\epsilon$ is the infidelity between $\rho$ and $\sigma$. Then 
	\begin{gather}
		2\lVert\Delta\rVert_2^2\le \lVert\Delta\rVert_1^2\le 4\epsilon, \label{eq:Delta12NormIFgen}
	\end{gather}
	where the first inequality is saturated if and only if (iff) $\rank(\Delta)\le 2$, and the second inequality is saturated iff both $\rho$ and $\sigma$ are pure states.
\end{proposition}

\begin{proposition}\label{pro:Delta12NormIF}
	Suppose $\sigma,\rho\in \caD(\caH)$, where $\sigma$ is a pure state, $\Delta=\rho-\sigma$, and $\epsilon$ is the infidelity between $\rho$ and $\sigma$. Then 
	\begin{gather}
		2\lVert\Delta\rVert_2^2\le \lVert\Delta\rVert_1^2\le \frac{4(d-1)}{d}\lVert\Delta\rVert_2^2, \label{eq:Delta12Norm}\\
		\frac{d}{d-1}\epsilon^2\le \lVert\Delta\rVert_2^2\le2\epsilon, \label{eq:Delta2NormIF}	\\	
		4\epsilon^2\le \lVert\Delta\rVert_1^2\le4\epsilon. \label{eq:Delta1NormIF}
	\end{gather}
	Here the first inequality in \eref{eq:Delta12Norm} is saturated iff $\rank(\Delta)\le 2$, and the second inequality is saturated iff all eigenvalues of $\Delta$, except for the smallest one, are equal. The first inequality in \eref{eq:Delta2NormIF} is saturated iff $\rho =(1-\epsilon)\sigma+\epsilon (\bbone-\sigma)/(d-1)$, and the second inequality is saturated iff $\rho$ is a pure state. The first inequality in \eref{eq:Delta1NormIF} is saturated iff $\rho$ and $\sigma$ commute, and the second inequality is
 saturated iff $\rho$ is a pure state or  is  orthogonal to $\sigma$.
\end{proposition}

Now, suppose $\scrN$ is a noise channel and $\rho=\scrN(\sigma)$. If  $\scrN$ is induced by  coherent noise, then $\rho$ is a pure state and $ \lVert\Delta\rVert_2^2 = 2\epsilon$, which saturates the upper bound in \eref{eq:Delta2NormIF}. If $\scrN$ is the depolarizing channel defined in \eref{eq:Depolarizing}, then $ \lVert\Delta\rVert_2^2=d\epsilon^2/(d-1)$, which saturates the lower bound in \eref{eq:Delta2NormIF}.

Next, suppose $\scrN=\scrP$ is the Pauli channel characterized by the vector $\bfp=(p_1, p_2, \ldots, p_{d^2-1})$ as in \eref{eq:PauliChannel2}. Then 
\begin{align}
	\Delta&=\scrP(\sigma)-\sigma=\sum_{i=1}^{d^2-1} p_i(P_i \sigma P_i-\sigma),\quad 
	\lVert\Delta\rVert_1\le \sum_{i=1}^{d^2-1} p_i\lVert P_i \sigma P_i-\sigma\rVert_1\le 2\lVert\bfp\rVert_1. 
\end{align}
If in addition $\sigma$ is a pure state, then the infidelity $\epsilon$ between $\rho$ and $\sigma$ reads 
\begin{align}
	\epsilon=1-\tr[\scrP(\sigma)\sigma]=\sum_{i=1}^{d^2-1}p_i[1- \tr(P_i \sigma P_i\sigma)]\le \lVert\bfp\rVert_1. 
\end{align}
If in addition $\sigma$ is a stabilizer state, then $\rho$ commutes with $\sigma$, which means 	$2\lVert\Delta\rVert_2^2\le 	\lVert\Delta\rVert_1^2= 4\epsilon^2$ .  To be concrete, we have $\lVert\Delta\rVert_1=2\epsilon =2\sum_i' p_i$,
where the summation runs over all $i$ such that $\tr(P_i \sigma P_i\sigma)=0$, that is, $P_i$ does not commute with all stabilizer operators of $\sigma$. If instead $\sigma$ is a Haar-random pure state, then $\lVert\Delta\rVert_2$, $\lVert\Delta\rVert_1$, and $\epsilon$ usually have the same order of magnitude as $\lVert\bfp\rVert_1$ according  to the following proposition.

\begin{proposition}\label{pro:pchannel_haar}
	Suppose $\sigma = |\phi\>\<\phi|$ is a Haar-random pure state on $\caH$, $\scrP$ is the Pauli channel characterized by the vector $\bfp=(p_1,p_2, \ldots, p_{d^2-1})$, $\rho=\scrP(\sigma)$,  $\Delta=\rho-\sigma$, and $\epsilon$ is the infidelity between $\rho$ and $\sigma$. Then 
	\begin{gather}
		\bepsilon=1-\bbE_{|\phi\>\sim \haar}\tr\left[\sigma\scrP(\sigma)\right]=\frac{d}{d+1}\lVert\bfp\rVert_1,\label{eq:MeanIFpauli}\\
		\bbE_{|\phi\>\sim \haar}\lVert\Delta\rVert_2^2=\frac{d}{d+1}\lVert\bfp\rVert_1^2+\frac{d}{d+1}\lVert\bfp\rVert_2^2,\label{eq:MeanDelta2NormPauli}\\
		\bepsilon^2\le 	\overline{\epsilon^2}\le\frac{d}{d-1} \overline{\epsilon^2}\le\bbE_{|\phi\>\sim \haar}\lVert\Delta\rVert_2^2\le\frac{2(d+1)}{d}\bepsilon^2,\label{eq:MeanDelta2NormIFpauli}\\
	4\bepsilon^2\le4 \overline{\epsilon^2}\le\bbE_{|\phi\>\sim \haar}\lVert\Delta\rVert_1^2\le\frac{8(d^2-1)}{d^2}\bepsilon^2\le 8\bepsilon^2.\label{eq:MeanDelta1NormIFpaulitt}
	\end{gather}
\end{proposition}

\begin{figure}
	\centering
	\includegraphics[width=0.9\columnwidth]{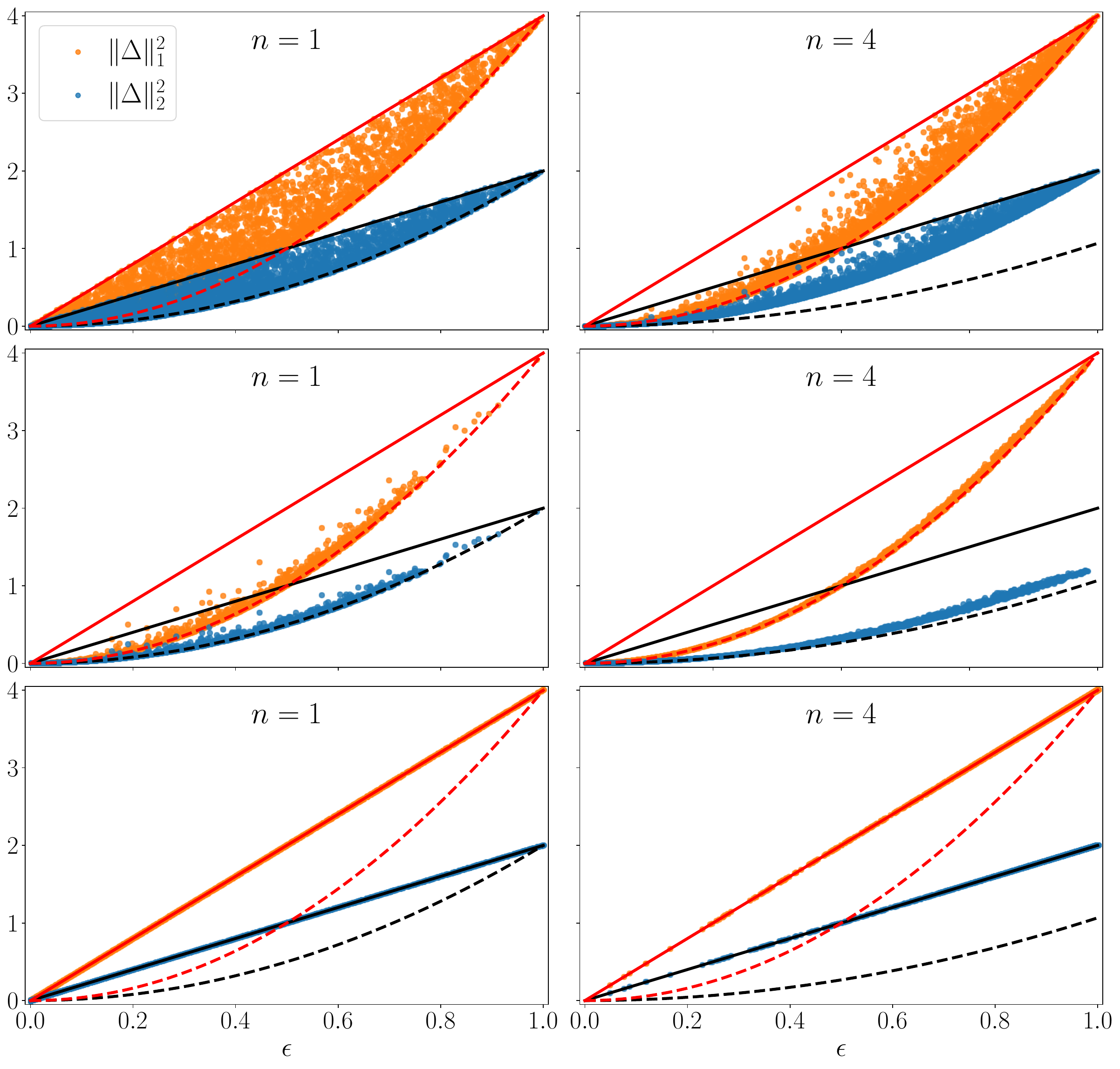}
	\caption{Scatter plots on the relations  between  $\lVert \Delta \rVert_1^2$, $\lVert \Delta \rVert_2^2$, and the infidelity $\epsilon$, where $\Delta=\rho-\sigma$.  Here each $\sigma$ is an $n$-qubit Haar-random pure state with $n=1,4$ and $\rho=\scrN(\sigma)$, where $\scrN$ denotes a random single-error Pauli channel for the first row, a random Pauli channel  (described in \sref{sec:DescriptRandomPauli}) for the second row, and a
	random coherent channel [random local rotations with rotation axes randomly oriented on the Bloch sphere and rotation angles uniformly distributed in the interval $[0,2\pi)$] for the third row. For each plot, 2000 Haar-random pure states together with 2000 random channels are generated. The red solid and dashed lines represent the upper and lower bounds for $\lVert \Delta \rVert_1^2$ presented  in \eref{eq:Delta1NormIF}, while the black solid and dashed lines represent the upper and lower bounds for $\lVert \Delta \rVert_2^2$ presented in \eref{eq:Delta2NormIF}.
}\label{fig:compare_three_cases}
\end{figure}

\Fref{fig:compare_three_cases} illustrates the relations between $\lVert\Delta\rVert_2^2$, $\lVert\Delta\rVert_1^2$, and the infidelity $\epsilon$ under three different noise models, namely, single-error Pauli noise, general Pauli noise, and coherent noise. In the first two cases, $\lVert \Delta \rVert_2^2$ and $\lVert \Delta \rVert_1^2$ are usually close to their lower bounds in \eqsref{eq:Delta2NormIF}{eq:Delta1NormIF} (see also \pref{pro:pchannel_haar}). Especially for the second case, the scaling behaviors of $\lVert \Delta \rVert_2^2$ and $\lVert \Delta \rVert_1^2$ are similar to the counterparts under depolarizing noise.
In the third case, by contrast, both  $\lVert \Delta \rVert_2^2$ and $\lVert \Delta \rVert_1^2$ coincide with their respective upper bounds in  \eqsref{eq:Delta2NormIF}{eq:Delta1NormIF}
because $\scrN(\sigma)$ remains a pure state under coherent noise.

\subsection{\label{SM:NormsInfidProof}Technical proofs}

\begin{proof}[Proof of \pref{pro:Delta12NormIFgen}]
	By definition $\Delta$ is traceless and can be expressed as $\Delta=\Delta_+-\Delta_-$, where 	$\Delta_+$ and $\Delta_-$ are positive semidefinite operators that have orthogonal supports and satisfy $\tr(\Delta_+)=\tr(\Delta_-)$. Therefore,
	\begin{gather}
		\lVert\Delta\rVert_1^2=		\left(\lVert\Delta_+\rVert_1+\lVert\Delta_-\rVert_1\right)^2= 2\lVert\Delta_+\rVert_1^2+2\lVert\Delta_-\rVert_1^2
		\geq 2\lVert\Delta_+\rVert_2^2+2\lVert\Delta_-\rVert_2^2	=	
		2\lVert\Delta\rVert_2^2,
	\end{gather}
	where the inequality is saturated iff  $\rank(\Delta_+)=\rank(\Delta_-)\le 1$. This result  confirms the first inequality in \eref{eq:Delta12NormIFgen} and shows that this inequality is saturated iff $\rank(\Delta)\le 2$. 

The second inequality in \eref{eq:Delta12NormIFgen} follows from the well-known relation between the trace distance and fidelity, and is known to be saturated when $\rho$ and $\sigma$ are pure states \cite{nielsen2020QCQI}.
\end{proof}

\begin{proof}[Proof of \pref{pro:Delta12NormIF}]By assumption $\rho$ is a density operator and $\sigma$ is a rank-1 projector, so  $\Delta$ is traceless and has at most one negative eigenvalue.  
	Let $\mu_1, \mu_2, \ldots,\mu_d$ be the eigenvalues of $\Delta$ arranged in nonincreasing order, that is, $\mu_1\geq \mu_2\geq \cdots\geq \mu_d$. Then $\mu_1,\mu_2,\ldots, \mu_{d-1}\geq0$, $\mu_d\le 0$, and 
	$\mu_1+\mu_2+\cdots+\mu_{d-1}=-\mu_d$. Therefore,
	\begin{equation}
	\begin{gathered}
	\lVert\Delta\rVert_1=\sum_{j=1}^d |\mu_j|=2|\mu_d|, \quad \lVert\Delta\rVert_1^2=4\mu_d^2, \\
	\lVert\Delta\rVert_2^2=\sum_{j=1}^d \mu_j^2\le  2\mu_d^2=\frac{1}{2}\lVert\Delta\rVert_1^2,\quad \lVert\Delta\rVert_2^2\geq \mu_d^2+\frac{\left(\sum_{j=1}^{d-1} \mu_j \right)^2}{d-1} =\frac{d}{d-1}\mu_d^2=\frac{d}{4(d-1)}\lVert\Delta\rVert_1^2,
\end{gathered}
	\end{equation}
	which imply \eref{eq:Delta12Norm}. 
	Here the first inequality is saturated iff $\mu_1=-\mu_d$ and $\mu_j=0$ for $j=2,3,\ldots, d-1$, while the second inequality is saturated iff $\mu_j=-\mu_d/(d-1)$ for  $j=1,2,\ldots, d-1$. These observations further confirm the saturation conditions of the inequalities in \eref{eq:Delta12Norm}. The first inequality in \eref{eq:Delta12Norm} and its saturation condition also follow from \pref{pro:Delta12NormIFgen}.
	
	The inequalities in \eref{eq:Delta2NormIF} and their saturation conditions are 
	simple corollaries of \lref{lem:PurityRange} below and the following relation:
	\begin{align}\label{eq:Delta2NormFProof}
		\lVert\Delta\rVert_2^2=\tr\left[\left(\rho-\sigma\right)^2\right]=1+\tr\left(\rho^2\right)-2\tr(\rho\sigma)=2\epsilon-1+\tr\left(\rho^2\right). 
	\end{align}

The first inequality in  \eref{eq:Delta1NormIF}  can be proved as follows: 
\begin{align}
        \lVert \Delta \rVert_1^2&=4\mu_d^2=4\left[\min_{|\psi\>\in\caH}\tr\left(\Delta|\psi\>\<\psi|\right)\right]^2\ge4\left[\tr\left(\Delta\sigma\right)\right]^2=4\epsilon^2. \label{eq:Delta1NormlowProof}
    \end{align}
The inequality in \eref{eq:Delta1NormlowProof} holds because $\min_{|\psi\>\in\caH}\left[\tr\left(\Delta|\psi\>\<\psi|\right)\right]\le\tr\left(\Delta\sigma\right)=-\epsilon\leq 0$; it is saturated iff $\sigma$ is an eigenstate of $\Delta$ with eigenvalue $\mu_d$, that is, $\rho$ and $\sigma$ commute. So the first inequality in  \eref{eq:Delta1NormIF} is saturated iff $\rho$ and  $\sigma$ commute.

Finally, we turn to the second inequality in  \eref{eq:Delta1NormIF}. Let $\rho=\sum_l q_l|\psi_l\>\<\psi_l|$ be a convex decomposition of  $\rho$ into pure states, where $q_l>0$ for all $l$,  and let $\epsilon_l =1-\<\psi_l|\sigma|\psi_l\>$. Then $\epsilon = 1-\tr(\rho\sigma)=\sum_l q_l\epsilon_l$ and
\begin{align}
	\lVert\Delta\rVert_1 &= \left\lVert\sum_l q_l|\psi_l\>\<\psi_l|-\sigma \right\rVert_1\le\sum_l q_l \lVert|\psi_l\>\<\psi_l|-\sigma\rVert_1=2\sum_l q_l\sqrt{\epsilon_l }.\label{eq:Delta12NormUppConvex}
\end{align}
In conjunction with the Cauchy-Schwarz inequality we can deduce that
\begin{align}
	\lVert\Delta\rVert_1^2 &\le4 \left(\sum_l q_l\sqrt{\epsilon_l }\right)^2\le4\left(\sum_l q_l\right)\left(\sum_l q_l\epsilon_l \right)=4\sum_l q_l\epsilon_l =4\epsilon,\label{eq:Delta12NormUppCau}
\end{align}
which confirms the second inequality in  \eref{eq:Delta1NormIF}. Alternatively, this inequality follows from \pref{pro:Delta12NormIFgen}. If $\rho$ is a pure state, then $\lVert\Delta\rVert_1^2=2\lVert\Delta\rVert_2^2=4\epsilon$, so  the second inequality in  \eref{eq:Delta1NormIF} is saturated. If $\rho$ is orthogonal to $\sigma$, then $\epsilon=1$ and $\lVert\Delta\rVert_1=2$, so the second inequality in  \eref{eq:Delta1NormIF} is also saturated.

Conversely, if the second inequality in  \eref{eq:Delta1NormIF} is saturated, then the inequality in \eref{eq:Delta12NormUppCau} is saturated, which means all $\epsilon_i$ are equal to each other, and a similar conclusion holds for all convex decompositions of  $\rho$ into pure states. It follows that all pure states in the support of $\rho$ have the same fidelity with $\sigma$. This condition can hold iff $\rho$ is a pure state or is orthogonal to $\sigma$. Therefore,  the second inequality in  \eref{eq:Delta1NormIF} is saturated iff $\rho$ is a pure state or is orthogonal to $\sigma$, which completes the proof of \pref{pro:Delta12NormIF}.
\end{proof}

Next, we prove an auxiliary lemma employed in the proof of \pref{pro:Delta12NormIF}.
\begin{lemma}\label{lem:PurityRange}
	Suppose $\sigma,\rho\in \caD(\caH)$, where $\sigma$ is a pure state, and  let $\epsilon=1-\tr(\rho \sigma)$ be the infidelity between $\rho$ and $\sigma$. Then 
	\begin{gather}
		(1-\epsilon)^2+ \frac{\epsilon^2}{d-1}\le \tr\left(\rho^2\right)\le1,  \label{eq:PurityRange}
	\end{gather}
	where the first inequality is saturated iff $\rho =(1-\epsilon)\sigma+\epsilon (\bbone-\sigma)/(d-1)$ and the second inequality is saturated iff $\rho$ is a pure state. If in addition $\rho$ commutes with $\sigma$, then 
	\begin{gather} 
		\tr\left(\rho^2\right)\le1-2\epsilon+2\epsilon^2,  \label{eq:PurityRange2}
	\end{gather}	
	and the upper bound is saturated iff $\rank[(\bbone-\sigma) \rho(\bbone-\sigma)] \le 1$. 
	\end{lemma}

\begin{proof}[Proof of \lref{lem:PurityRange}]
	The second  inequality in \eref{eq:PurityRange} and the saturation condition are trivial. 
	To prove the first inequality in \eref{eq:PurityRange}, let 
	\begin{align}
		\varrho=\sigma \rho  \sigma +(\bbone -\sigma)\rho (\bbone -\sigma)=(1-\epsilon)\sigma+B, 
	\end{align}
	where $B=(\bbone -\sigma)\rho (\bbone -\sigma)$ is a positive operator supported in a $(d-1)$-dimension subspace of $\caH$. In addition, $B$ 
	satisfies  
	\begin{align}\label{eq:PurityLBproof1}
		\tr(B)=\epsilon,\quad \frac{\epsilon^2}{d-1}\le \tr\left(B^2\right)\le \epsilon^2,
	\end{align} 
	where the first inequality is saturated iff $B$ is proportional to $\bbone -\sigma$, and the second inequality is saturated iff $B$ is proportional to a rank-1 projector, that is, $\rank(B)\le 1$. 
	Therefore,
	\begin{align}\label{eq:PurityLBproof2}
		\tr\left(\rho^2\right)\geq \tr(\varrho^2)=(1-\epsilon)^2+\tr\left(B^2\right)\geq (1-\epsilon)^2+ \frac{\epsilon^2}{d-1},
	\end{align}
	which confirms the first inequality in \eref{eq:PurityRange}. If $\rho =(1-\epsilon)\sigma+\epsilon (\bbone-\sigma)/(d-1)$, then it is straightforward to verify that the first inequality in \eref{eq:PurityRange} is saturated. Conversely, if this inequality is saturated, then the two inequalities in \eref{eq:PurityLBproof2} are saturated, which means $\rho=\varrho$ and  $B$ is proportional to $\bbone -\sigma$, so $\rho =(1-\epsilon)\sigma+\epsilon (\bbone-\sigma)/(d-1)$.

	Next, suppose $\rho$ commutes with $\sigma$. Then $\rho=\varrho$ and 
	\begin{align}\label{eq:PurityLBproof3}
		\tr\left(\rho^2\right)= \tr(\varrho^2)=(1-\epsilon)^2+\tr\left(B^2\right)\le (1-\epsilon)^2+ \epsilon^2=1-2\epsilon+2\epsilon^2,
	\end{align}
	which confirms \eref{eq:PurityRange2}. Here the inequality follows from \eref{eq:PurityLBproof1} and is saturated iff $\rank(B)\le 1$. This observation completes the proof of \lref{lem:PurityRange}.
\end{proof}

\begin{proof}[Proof of \pref{pro:pchannel_haar}]
	According to Schur–Weyl duality, we have 
	\begin{equation}
		\bbE_{|\phi\>\sim\haar}(\sigma\otimes\sigma)=\frac{\bbone+\SWAP}{d(d+1)}.
	\end{equation}
Therefore, for any two Pauli operators $P_i$ and $P_j$ in $\bcaP_n$, we have
	\begin{align}
		\bbE_{|\phi\>\sim\haar} \tr(P_i \sigma P_j \sigma)&=\bbE_{|\phi\>\sim\haar}\left[\tr(P_i\sigma)\tr(P_j\sigma)\right]=\tr\left[ (P_i\otimes P_j)\bbE_{|\phi\>\sim\haar}(\sigma\otimes\sigma)\right]=\frac{d\delta_{i0}\delta_{j0}+\delta_{ij}}{d+1}, \label{eq:HaarPauli}
	\end{align}	
	which implies that
	\begin{align}
		\bbE_{|\phi\>\sim\haar} \tr(P_i P_j \sigma P_j P_i \sigma)&
		=\frac{d\delta_{ij}+1}{d+1}.\label{eq:HaarPauli2}
	\end{align}

	Next, by definition we have
	\begin{gather}
		\tr\left[\sigma\scrP(\sigma)\right] = \tr\left[(1-\lVert\bfp\rVert_1)\sigma^2+ \sum_{i=1}^{d^2-1} p_i P_i \sigma P_i\sigma\right]=(1-\lVert\bfp\rVert_1)+\sum_{i=1}^{d^2-1} p_i \tr(P_i \sigma P_i \sigma).
	\end{gather}
	Together with \eref{eq:HaarPauli}, this equation implies that
	\begin{equation}
	\begin{gathered}
	\bbE_{|\phi\>\sim\haar}\tr\left[\sigma\scrP(\sigma)\right]= (1-\lVert\bfp\rVert_1)+ \frac{1}{d+1}\sum_{i=1}^{d^2-1} p_i=  1 -\frac{d}{d+1}\lVert\bfp\rVert_1,\\
	\bepsilon=1-\bbE_{|\phi\>\sim\haar}\tr\left[\sigma\scrP(\sigma)\right] =\frac{d}{d+1}\lVert\bfp\rVert_1,
\end{gathered}
	\end{equation}
	which confirm \eref{eq:MeanIFpauli}.

	In addition, $\Delta$ and $\lVert\Delta\rVert_2^2$ can be expressed as follows:
	\begin{equation}
	\begin{aligned}
	\Delta&=\scrP(\sigma)-\sigma=\sum_{i=1}^{d^2-1}p_i (P_i \sigma P_i-\sigma),\\
	\lVert\Delta\rVert_2^2&=\lVert\scrP(\sigma)-\sigma\rVert_2^2=\sum_{i,j=1}^{d^2-1}p_ip_j \left[\tr\left(P_iP_j\sigma P_j P_i\sigma\right)-\tr(P_i\sigma P_i\sigma)-\tr\left(P_j\sigma P_j\sigma\right)+1\right]\\
	&=\lVert\bfp\rVert_1^2+\sum_{i,j=1}^{d^2-1}p_ip_j\tr\left(P_iP_j\sigma P_j P_i\sigma\right)	
	-	2\lVert\bfp\rVert_1 \sum_{i=1}^{d^2-1}p_i\tr(P_i\sigma P_i\sigma).
\end{aligned}
	\end{equation}
	In conjunction with \eqsref{eq:HaarPauli}{eq:HaarPauli2} we can derive that                                            
	\begin{align}
		\bbE_{|\phi\>\sim\haar}  \lVert\Delta\rVert_2^2&=\lVert\bfp\rVert_1^2+\sum_{i,j=1}^{d^2-1}\frac{p_ip_j(d\delta_{ij}+1)}{d+1}-\frac{2\lVert\bfp\rVert_1^2}{d+1}=\frac{d}{d+1}\lVert\bfp\rVert_1^2+\frac{d}{d+1}\lVert\bfp\rVert_2^2,
	\end{align}
	which confirms \eref{eq:MeanDelta2NormPauli}.

	Finally, we turn to \eqsref{eq:MeanDelta2NormIFpauli}{eq:MeanDelta1NormIFpaulitt}. The first two inequalities in \eref{eq:MeanDelta2NormIFpauli} are trivial; the third inequality follows from \pref{pro:Delta12NormIF}; the last inequality  is a simple corollary of \eqsref{eq:MeanIFpauli}{eq:MeanDelta2NormPauli} together with the following inequalities:
	\begin{equation}
		\frac{\lVert\bfp\rVert_1^2}{d^2-1}\le\lVert\bfp\rVert_2^2\le\lVert\bfp\rVert_1^2,
	\end{equation}
which also imply that $\bbE_{|\phi\>\sim\haar}  \lVert\Delta\rVert_2^2\geq d\bepsilon^2/(d-1)$.

\Eref{eq:MeanDelta1NormIFpaulitt} follows from \pref{pro:Delta12NormIF} and \eref{eq:MeanDelta2NormIFpauli}, which  completes the proof of \pref{pro:pchannel_haar}. 
\end{proof}

\section{\label{SM:CharProp}Properties of cross characteristic functions}
In this section, we clarify the key properties of cross characteristic functions defined in \eref{eq:CrossChar}, which are crucial to understanding the variance in  CRM shadow estimation based on the Clifford group.

\subsection{Basic properties}

\begin{lemma}\label{lem:CharProp}
	Suppose $\sigma,\rho\in \caD(\caH)$,  $\Delta=\rho-\sigma$,  $\epsilon$ is the infidelity between $\rho$ and $\sigma$, and $Q\in\caL^{\rmH}(\caH)$. Then 
	\begin{align}\label{eq:CharProp1}
		\tXi_{\Delta,Q}\cdot \Xi_{\Delta,Q}\le \lVert\tXi_{\Delta,Q}\rVert_2^2=\lVert \Xi_{\Delta,Q}\rVert_2^2\le \min\left\{d\lVert\Delta\rVert_1^2\lVert Q\rVert _2^2, \sqrt{d}\lsp \lVert\Delta\rVert_1\lVert\Delta\rVert_2\lVert\Xi_Q\rVert_4^2\right\}\leq 4d\epsilon \lVert Q\rVert _2^2. 
	\end{align}	
	If $\sigma$ is a pure state, then 
	\begin{align}
		\lVert\Xi_{\Delta,Q}\rVert_2^2&\le d\lVert \Delta\rVert_1^2\lVert Q\rVert _2^2\le \min\left\{4d\epsilon,4(d-1)\lVert\Delta\rVert_2^2\right\}\lVert Q\rVert_2^2, \label{eq:CharPropPure1}\\
		\lVert \Xi_{\Delta,Q}\rVert_2^2& \le  \sqrt{d}\lsp \lVert\Delta\rVert_1\lVert\Delta\rVert_2\lVert\Xi_Q\rVert_4^2\le \min\left\{ 2\sqrt{2d}\lsp \epsilon, 2\sqrt{d-1} \lsp\lVert\Delta\rVert_2^2  \right\}\lVert \Xi_Q\rVert_4^2. \label{eq:CharPropPure2}
	\end{align}	
	If $\sigma$ is a pure state and $Q=\sigma-\bbone/d$, then 	
	\begin{align}
		\lVert\Xi_{\Delta,Q}\rVert_2^2&= \lVert\Xi_{\Delta,\sigma}\rVert_2^2\le 	d\lVert\Delta\rVert_2^2 \leq 2d\epsilon,  \label{eq:CharPropFE0} \\
		\lVert \Xi_{\Delta,Q}\rVert_2^2&= \lVert \Xi_{\Delta,\sigma}\rVert_2^2 \le
		2^{-M_2(\sigma)/2} d \lVert\Delta\rVert_1 \lVert\Delta\rVert_2 \le 2^{1-M_2(\sigma)/2}d\min\left\{\sqrt{2}\lsp \epsilon, \lVert\Delta\rVert_2^2\right\}. \label{eq:CharPropFE}
	\end{align}	
\end{lemma}

Next, suppose $\sigma$ is a pure state and $\rho$ is tied to $\sigma$ by a Pauli channel. Our goal is to derive a nontrivial upper bound for $\lVert\Xi_{\Delta,\sigma}\rVert_2^2$ in terms of the infidelity $\epsilon$ between $\rho$ and $\sigma$. To this end, we need to introduce some auxiliary notation and results. Given two  Pauli operators $P_i, P_j$ in $\bcaP_n$,  define
\begin{gather}
	\eta_i=1-\left[\tr(\sigma P_i)\right]^2,\quad 	K_{\sigma}(P_i):=\sum_{ k\, |\, \{P_k, P_i\}=0}\left[\tr(\sigma P_k)\right]^4,\quad  K_{\sigma}(P_i,P_j):=\sum_{ k\, |\, \{P_k, P_i\}=\{P_k, P_j\}=0}\left[\tr(\sigma P_k)\right]^4.\label{eq:DefKfunction}
\end{gather}
By definition $K_{\sigma}(P_j,P_i)=K_{\sigma}(P_i,P_j)$ and $K_{\sigma}(P_i,P_i)=K_{\sigma}(P_i)$; if $i=0$, then $K_{\sigma}(P_i,P_j)=K_{\sigma}(P_i)=\eta_i=0$.

\begin{lemma}\label{lem:KetaIneq}
	Suppose $\sigma$ is a pure state on $\caH$ and  $P_i, P_j\in\bcaP_n$ with $i,j\geq 1$. Then 
	\begin{align}
		K_{\sigma}(P_i)\le\frac{d\eta_i^2}{2} ,\quad K_{\sigma}(P_i,P_j)\le\frac{d \eta_i \eta_j}{2},\label{eq:KetaIneq}
	\end{align}
	where the first inequality is saturated when $\sigma$ is a stabilizer state, and the second inequality is saturated when  $\sigma$ is a stabilizer state and $P_i=P_j$.
\end{lemma}

\begin{lemma}\label{lem:InfidCrossCharUB}
	Suppose $\sigma$ is a pure state on $\caH$, $\scrP$ is the Pauli channel characterized by the vector $\bfp=(p_1,p_2, \ldots, p_{d^2-1})$, $\rho=\scrP(\sigma)$,  $\Delta=\rho-\sigma$, and $\epsilon$ is the infidelity between $\rho$ and $\sigma$. Then 
	\begin{align}
		\epsilon &=\sum_{i=1}^{d^2-1}p_i\eta_i, \quad 
		\lVert\Xi_{\Delta,\sigma}\lVert_2^2=4\sum_{i,j=1}^{d^2-1}p_ip_j K_{\sigma}(P_i,P_j)\leq 2d\epsilon^2,  \label{eq:InfidCrossCharUB}	
	\end{align}
	where the inequality is saturated when $\sigma$ is a stabilizer state and $\scrP$ is a single-error Pauli channel.
\end{lemma}

\begin{figure}
	\centering
	\includegraphics[width=0.9\columnwidth]{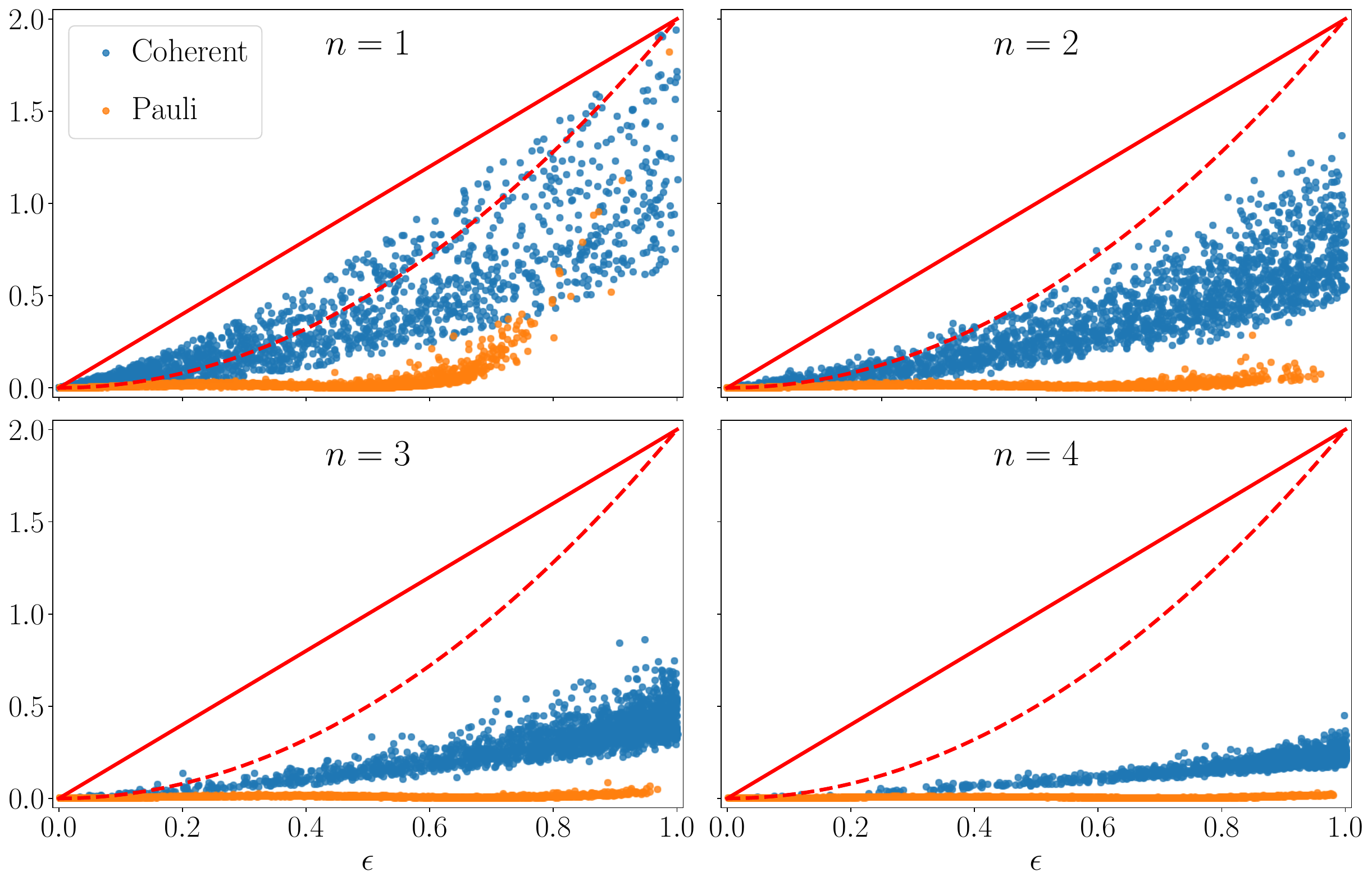}
	\caption{Scatter plots on the relation  between $\lVert\Xi_{\Delta,O}\rVert_2^2/d$ and the infidelity $\epsilon$, where $\Delta=\rho-\sigma$ and $O=\sigma-\bbone/d$. Here each $\sigma$ is an $n$-qubit Haar-random pure state with $n=1,2,3, 4$, $\rho=\scrN(\sigma)$, and $\scrN$ is either a random coherent channel  or a random Pauli channel as in \fref{fig:compare_three_cases}. For each plot, 2000 Haar-random pure states together with 2000 random coherent channels and 2000 random Pauli channels are generated. The red solid line denotes the upper bound $2\epsilon$ for $\lVert\Xi_{\Delta,O}\rVert_2^2$ as shown in \eref{eq:CharPropFE0}, while the red dashed line denotes the upper bound $2\epsilon^2$ for $\lVert\Xi_{\Delta,O}\rVert_2^2$ when $\scrN$ is a Pauli channel as shown in \lref{lem:InfidCrossCharUB}.
	}\label{fig:compare_three_cases_char}
\end{figure}

Next, we determine the expectation value of $\lVert\Xi_{\Delta,\sigma}\rVert_2^2$ when $\sigma = |\phi\>\<\phi|$ is a Haar-random pure state on $\caH$. 

\begin{proposition}\label{pro:pchannel_haar_char}
	Suppose $\sigma = |\phi\>\<\phi|$ is a Haar-random pure state on $\caH$, $\scrP$ is the Pauli channel characterized by the vector $\bfp=(p_1,p_2, \ldots, p_{d^2-1})$, $\rho=\scrP(\sigma)$,  $\Delta=\rho-\sigma$, and $\epsilon$ is the infidelity between $\rho$ and $\sigma$. Then  
	\begin{gather}
		\bbE_{|\phi\>\sim \haar}\lVert\Xi_{\Delta,\sigma}\rVert_2^2=\frac{3d^2}{(d+1)(d+3)}\left(\lVert\bfp\rVert_1^2+\lVert\bfp\rVert_2^2\right)=\frac{3d}{d+3}\bbE_{|\phi\>\sim \haar}\lVert\Delta\rVert_2^2,\label{eq:MeanChar2NormPauli}\\
		\frac{3d\bepsilon^2}{d+2}\le\bbE_{|\phi\>\sim \haar}\lVert\Xi_{\Delta,\sigma}\rVert_2^2\le\frac{6(d+1)}{d+3}\bepsilon^2.\label{eq:MeanChar2NormPauli2}
	\end{gather}
\end{proposition}

\Fref{fig:compare_three_cases_char} illustrates the relation between $\lVert \Xi_{\Delta,O}\rVert_2^2/d$ and  the infidelity $\epsilon$,  where $\Delta=\rho-\sigma$, $O=\sigma-\bbone/d$, $\sigma$ is an $n$-qubit Haar-random pure state, $\rho=\scrN(\sigma)$, and $\scrN$ is either a coherent channel or a Pauli channel. Note that $\lVert \Xi_{\Delta,O}\rVert_2^2/d\leq 2\epsilon$ according to \eref{eq:CharPropFE0} in \lref{lem:CharProp}. When $\scrN$ is a Pauli channel, we have a tighter upper bound $\lVert \Xi_{\Delta,O}\rVert_2^2/d\leq 2\epsilon^2$  thanks to \lref{lem:InfidCrossCharUB} (see also \pref{pro:pchannel_haar_char}). When  $\scrN$ is  a coherent channel, the upper bound $2\epsilon^2$ is also applicable with a high probability, except when $n$ is very small.

\subsection{Technical proofs}

\begin{proof}[Proof of \lref{lem:CharProp}]
	The first inequality and equality in \eref{eq:CharProp1} follow from \rcite{chen2024nonstab}. The second inequality in \eref{eq:CharProp1} can be proved as follows:
	\begin{align}
		\lVert\Xi_{\Delta,Q}\rVert_2^2 &=\sum_{P\in \bcaP_n}\left[\tr(\Delta P)\right]^2 \left[\tr(Q P)\right]^2\le \lVert\Delta\rVert_1^2 \sum_{P\in \bcaP_n}\left[\tr(Q P)\right]^2=d\lVert\Delta\rVert_1^2\lVert Q\rVert_2^2,  \label{eq:CharPropProof1}\\
		\lVert \Xi_{\Delta,Q}\rVert_2^2&\le \sqrt{\sum_{P\in \bcaP_n}[\tr(\Delta P)]^4}\sqrt{\sum_{P\in \bcaP_n}[\tr(Q P)]^4}\le
		\sqrt{d}\lsp \lVert\Delta\rVert_1\lVert\Delta\rVert_2\lVert\Xi_Q\rVert_4^2.  \label{eq:CharPropProof2}
	\end{align}
	Here the first inequality in \eref{eq:CharPropProof1} and the second inequality in \eref{eq:CharPropProof2}  hold because $|\tr(\Delta P)|\le \lVert\Delta\rVert_1$ for all $P\in \bcaP_n$. The first inequality in \eref{eq:CharPropProof2} follows from the Cauchy-Schwarz inequality.  The last inequality in \eref{eq:CharProp1} follows from the fact that $\lVert\Delta\rVert_1^2\leq 4\epsilon$ by \pref{pro:Delta12NormIFgen}.
	
	If $\sigma$ is a pure state, then \eqsref{eq:CharPropPure1}{eq:CharPropPure2} follow from \eref{eq:CharProp1} and \pref{pro:Delta12NormIF}.
	
	Finally, suppose $\sigma$ is a pure state and $Q=\sigma-\bbone/d$. Then 
\begin{align}
		|\tr(QP)|\le |\tr(\sigma P)|\le 1,\quad |\tr(\Delta P)\tr(QP)|=|\tr(\Delta P)\tr(\sigma P)|\le 1  \quad \forall\,  P\in \bcaP_n. 
	\end{align}	
	 Therefore,
	\begin{align}
		\lVert\Xi_{\Delta,Q}\rVert_2^2 &= \lVert\Xi_{\Delta,\sigma}\rVert_2^2=\sum_{P\in \bcaP_n}\left[\tr(\Delta P)\right]^2 \left[\tr(\sigma P)\right]^2		
		\le  \sum_{P\in \bcaP_n}[\tr(\Delta P)]^2 = d\lVert\Delta\rVert_2^2\leq 2d\epsilon, 
	\end{align}
	which confirms \eref{eq:CharPropFE0} and the equality in \eref{eq:CharPropFE}. Here the last inequality follows from \pref{pro:Delta12NormIF}. 	
	The two inequalities in \eref{eq:CharPropFE} follow from \eref{eq:CharPropPure2} and the fact that $\lVert\Xi_\sigma\rVert_4^4=2^{-M_2(\sigma)}d$ according to  the definition of 2-SRE. This observation completes the proof of \lref{lem:CharProp}. 
\end{proof}

\begin{proof}[Proof of  \lref{lem:KetaIneq}]
	By assumption $\sigma$ is the projector onto a pure state $|\phi\>$ in $\caH$ that can be expressed as follows:
	\begin{align}
		|\phi\> = \alpha_i |\psi_i^+\> + \sqrt{1-\alpha_i^2} \lsp|\psi_i^-\>,
	\end{align}
	where $0\leq \alpha_i\leq 1$ and $|\psi_i^+\>$
	($|\psi_i^-\>$) is an eigenstate of $P_i$ with eigenvalue $+1$ ($-1$), so
	\begin{equation}
		\eta_i =1-[\tr(\sigma P_i)]^2= 4\alpha_i^2\left(1-\alpha_i^2\right).
	\end{equation}
	If $P_k\in \bcaP_n$ anticommutes with $P_i$, that is, $\{P_k,P_i\}=0$, then 
	\begin{align}
		\tr(\sigma P_k)=2\alpha_i\sqrt{1-\alpha_i^2}\lsp \tr(M_i P_k)=\sqrt{\eta_i}\lsp  \tr(M_i P_k),
	\end{align}
	where  
	\begin{gather}
		M_i =\frac{|\psi_i^+\>\<\psi_i^-|
			+ |\psi_i^-\>\<\psi_i^+|}{2}.
	\end{gather}
	Now, the first inequality in \eref{eq:KetaIneq} can be derived as follows:
	\begin{align}
		K_{\sigma}(P_i)&=\sum_{ k\, |\, \{P_k, P_i\}=0}\left[\tr(\sigma P_k)\right]^4=
		\eta_i^2\sum_{ k\, |\, \{P_k, P_i\}=0}	 \left[\tr(M_i P_k)\right]^4\leq
		\eta_i^2\sum_{ k=0}^{d^2-1}	
		\left[\tr(M_i P_k)\right]^4\nonumber\\
		&\leq \eta_i^2\sum_{ k=0}^{d^2-1}	
		\left[\tr(M_i P_k)\right]^2
		=d\eta_i^2\tr\left(M_i^2\right)=\frac{d\eta_i^2}{2},
	\end{align}	
	where the second inequality follows from the fact that $|\tr(M_i P_k)|\leq \lVert M_i\rVert _1 \lVert P_k\rVert =1$. As a corollary, we have
	\begin{align}
		K_\sigma(P_i,P_j)
		&\le \min\left[K_\sigma(P_i),K_\sigma(P_j)\right]
		\le \sqrt{K_\sigma(P_i) K_\sigma(P_j)}
		\le \frac{d \eta_i \eta_j}{2}, \label{eq:exactK2function}
	\end{align}	
	which confirms the second inequality in \eref{eq:KetaIneq}.

	Next, suppose $\sigma$ is a stabilizer state. Then $\tr(\sigma P)$ for $P\in \bcaP_n$ can take on three possible values, namely, 0, 1, and $-1$. Let 
	\begin{align}
		\caG_\sigma=\left\{P\in \bcaP_n \,|\, [\tr(\sigma P)]^2=1\right\};
	\end{align}
	then $\caG$ contains $d$ Pauli operators. 
	If in addition $P_i$ or $-P_i$ is a stabilizer of $\sigma$, then $[\tr(\sigma P_i)]^2=1$ and $\eta_i=0$. Meanwhile, all Pauli operators in $\caG_\sigma$ commute with $P_i$. In other words, any Pauli operator $P_k$ that anticommutes with $P_i$ cannot be a stabilizer of $\sigma$ and satisfies $\tr(P_k\sigma)=0$. So $	K_{\sigma}(P_i)=0=d\eta_i^2/2$. 
	If instead neither $P_i$ nor $-P_i$ is a stabilizer of $\sigma$, then $[\tr(\sigma P_i)]^2=0$ and $\eta_i=1$. Meanwhile, exactly one half of the Pauli operators in $\caG_\sigma$ anticommutes with $P_i$. So  $K_{\sigma}(P_i)=d/2=d\eta_i^2/2$. In both cases, the first inequality in \eref{eq:KetaIneq} is saturated. If in addition $i=j$, then $\eta_i=\eta_j$ and $K_{\sigma}(P_i,P_j)=K_{\sigma}(P_i)$, so the second inequality in \eref{eq:KetaIneq} is also saturated. This observation completes the proof of \lref{lem:KetaIneq}.
\end{proof}

\begin{proof}[Proof of \lref{lem:InfidCrossCharUB}]
	By assumption $\sigma$ is a pure state and $\rho =\scrP(\sigma)=\left(1-\lVert\bfp\rVert_1\right)\sigma+\sum_{i=1}^{d^2-1}p_i P_i \sigma P_i$, so we have
	\begin{align}
		\epsilon &= 1-\tr(\rho\sigma) =\sum_{i=1}^{d^2-1} p_i \left[1-\tr(P_i\sigma P_i \sigma)\right]=\sum_{i=1}^{d^2-1} p_i \left\{1-[\tr(\sigma P_i)]^2\right\}=\sum_{i=1}^{d^2-1} p_i \eta_i,
	\end{align}
	which confirms the first equality in  \eref{eq:InfidCrossCharUB}. In addition,
	\begin{gather}
		\Delta = \sum_{i=1}^{d^2-1}p_i\left(P_i\sigma P_i-\sigma\right), \;\; \tr(\Delta P_k)=\sum_{i=1}^{d^2-1}p_i \tr(P_i\sigma P_i P_k-\sigma P_k)=-2\tr(\sigma P_k)\sum_{i\geq 1\,|\,\{P_i,P_k\}=0}p_i\quad \forall\, k\geq 1.
	\end{gather}	
	Therefore,
	\begin{align}
		&\lVert\Xi_{\Delta,\sigma}\lVert_2^2 =\sum_{k=1} ^{d^2-1}[\tr(\Delta P_k)]^2[\tr(\sigma P_k)]^2=4\sum_{k=1}^{d^2-1}[\tr(\sigma P_k)]^4\sum_{i,j\geq 1\, |\, \{P_i, P_k\}=\{P_j,P_k\}=0}p_ip_j\notag\\
		&=4\sum_{i,j= 1}^{d^2-1} p_ip_j \sum_{k\,|\,\{P_k,P_i\}=\{P_k,P_j\}=0,}\left[\tr(\sigma P_k)\right]^4=4\sum_{i,j=1}^{d^2-1}p_ip_j K_{\sigma}(P_i,P_j)\leq 2d \left(\sum_{i= 1}^{d^2-1}p_i\eta_i\right)^2=2d\epsilon^2,
	\end{align}
	which confirms the second equality and the inequality in \eref{eq:InfidCrossCharUB}.  Here the inequality follows from \lref{lem:KetaIneq}.

	Next, suppose $\sigma$ is a stabilizer state and $\scrP$ is a single-error Pauli channel. Then 
	\begin{align}
		\epsilon=p_i\eta_i, \quad \lVert\Xi_{\Delta,\sigma}\lVert_2^2=4p_i^2 K_{\sigma}(P_i)=2dp_i^2\eta_i^2,
	\end{align}
	for some $i\in \{1,2,\ldots,d^2-1\}$,
	where the last equality follows from \lref{lem:KetaIneq}. So the inequality in  \eref{eq:InfidCrossCharUB}	is saturated, 	which completes the proof of \lref{lem:InfidCrossCharUB}.
\end{proof}

\begin{proof}[Proof of \pref{pro:pchannel_haar_char}]
	According to Schur–Weyl duality, we have 
	\begin{equation}
		\bbE_{|\phi\>\sim\haar}\left(\Tensor{\sigma}{4}\right)=\frac{\sum_{\pi\in S_4}W_{\pi}}{d(d+1)(d+2)(d+3)},
	\end{equation}
	where $W_\pi$ is the permutation operator on $\Tensor{\caH}{4}$ associated with the permutation $\pi\in S_4$.  Based on this observation, for any Pauli operator $P\in \bcaP_n$ that is not proportional to the identity operator, we can deduce that
	\begin{align}
		\bbE_{|\phi\>\sim\haar}\left[\tr(\sigma P)\right]^4&=\tr\left[\Tensor{P}{4}\bbE_{|\phi\>\sim\haar}\left(\Tensor{\sigma}{4}\right)\right]=\frac{\sum_{\pi\in S_4}\tr\left(W_{\pi}\Tensor{P}{4}\right)}{d(d+1)(d+2)(d+3)}
		=\frac{3\left[\tr\left(P^2\right)\right]^2+6\tr\left(P^4\right)}{d(d+1)(d+2)(d+3)} \nonumber\\
		&	
		=\frac{3d^2+6d}{d(d+1)(d+2)(d+3)}	
		=\frac{3}{(d+1)(d+3)}.
	\end{align}
	Furthermore,
	\begin{align}
		&\bbE_{|\phi\>\sim\haar}\lVert\Xi_{\Delta,\sigma}\lVert_2^2=	\bbE_{|\phi\>\sim\haar}\sum_{k=1} ^{d^2-1}[\tr(\Delta P_k)]^2[\tr(\sigma P_k)]^2\nonumber\\
		&=4	\bbE_{|\phi\>\sim\haar}\sum_{k=1}^{d^2-1}[\tr(\sigma P_k)]^4\sum_{i,j\geq 1\, |\, \{P_i, P_k\}=\{P_j,P_k\}=0}p_ip_j\notag\\
		&=4\sum_{i,j= 1}^{d^2-1} p_ip_j \sum_{k\,|\,\{P_k,P_i\}=\{P_k,P_j\}=0,}	\bbE_{|\phi\>\sim\haar}\left[\tr(\sigma P_k)\right]^4=\frac{12}{(d+1)(d+3)}\sum_{i,j=1}^{d^2-1}p_ip_jC_{i,j},
	\end{align}	
	where $C_{i,j}$ denotes the number of Pauli operators in $\bcaP_n$ that anti-commute with both  $P_i$ and $P_j$. It is straightforward to verify that 
	\begin{equation}
		C_{i,j}=\frac{d^2(1+\delta_{ij})}{4}\quad \forall\, i,j\geq 1. 
	\end{equation}
	Therefore,
	\begin{align}
		\bbE_{|\phi\>\sim\haar}\lVert\Xi_{\Delta,\sigma}\rVert_2^2&=\frac{3d^2}{(d+1)(d+3)}\sum_{i,j=1}^{d^2-1}p_ip_j(1+\delta_{ij})
		=\frac{3d^2\left(\lVert\bfp\rVert_1^2+\lVert\bfp\rVert_2^2\right)}{(d+1)(d+3)}=\frac{3d}{d+3}\bbE_{|\phi\>\sim \haar}\lVert\Delta\rVert_2^2,
	\end{align}	
	which confirms \eref{eq:MeanChar2NormPauli}. Here the last equality follows from \eref{eq:MeanDelta2NormPauli} in \pref{pro:pchannel_haar}. 
	
	Finally, \eref{eq:MeanChar2NormPauli2} follows from \eqsref{eq:MeanChar2NormPauli}{eq:MeanDelta2NormIFpauli}, which completes the proof of \pref{pro:pchannel_haar_char}. 
\end{proof}

\section{\label{SM:CliffordVar}Variances in CRM shadow estimation based on Clifford measurements}
In this section, we first derive an alternative formula for the variance $\bbV_*(O,\Delta)$ in CRM shadow estimation based on the Clifford group (see \lref{lem:V*CliffordAlt} below)
and then prove \thsref{thm:CliffordVstar}{thm:CliffordHPFEpn} in the main text. In the course of study, we also clarify  the variances $\bbV(O,\rho)$, $\bbV_*(O,\rho)$, and $\bbV_*(O,\Delta)$ in CRM shadow estimation in the presence of depolarizing noise.

\subsection{Alternative formula for the variance $\bbV_*(O,\Delta)$}

\begin{lemma}\label{lem:V*CliffordAlt}
	Suppose $O$ is a traceless observable on $\caH$. Then the variance $\bbV_*(O,\Delta)$ in $\CRM(\Cl_n,\sigma,R)$ reads
	\begin{align}
		V_*(O,\Delta)&=\frac{2(d+1)}{d^2(d+2)}\sum_{\substack{i,j\geq 1\,|\,i \ne j, \\ [P_i, P_j]=0}} \tr(OP_i) \tr(OP_j) \tr(\Delta P_i) \tr(\Delta P_j)  \notag \\
		&\equad + \frac{(d+1)}{d^2} \sum_{i=1}^{d^2-1} \left[\tr\left(OP_i\right)\right]^2 \left[\tr\left(\Delta P_i\right)\right]^2- \left[\tr(\Delta O)\right]^2.  \label{eq:V*CliffordAlt}
	\end{align}
\end{lemma}

Before proving \lref{lem:V*CliffordAlt}, we first introduce an auxiliary result. Given two Pauli operators $P_i, P_j$ in $\bcaP_n$,  define $\scrD(P_i)$ as the set of Clifford unitaries in $\Cl_n$ that can diagonalize $P_i$ with respect to the computational basis; define $\scrD(P_i,P_j)$  as the set of Clifford unitaries  that can diagonalize $P_i$ and $P_j$ simultaneously. Note that $\scrD(P_i,P_j)$ is empty whenever $P_i$ and $P_j$ do dot commute. 

\begin{lemma}\label{lem:mixing}
  Suppose $P_i$ and $P_j$ are two distinct Pauli operators (with $i\neq j$) in $\bcaP_n$ that commute with each other. Then 
    \begin{align}
       \bbE_{U\sim\Cl_n}\left[\delta_{U \in \scrD(P)}\right] & = \frac{1}{d+1}, \quad 
        \bbE_{U\sim\Cl_n}\left[\delta_{U \in \scrD(P_1, P_2)}\right] = \frac{2}{(d+1)(d+2)}.
    \end{align}  
\end{lemma}
\Lref{lem:mixing} follows from the fact that the Clifford group acts transitively on nontrivial Pauli operators  and also on  commuting pairs of distinct nontrivial Pauli operators \cite{Zhu20173Design}; in other words, 
the Clifford group is both Pauli-mixing and Pauli 2-mixing \cite{webb2016clifford}.

\begin{proof}[Proof of \lref{lem:V*CliffordAlt}]
Since the Clifford group $\Cl_n$ is a unitary 3-design  and $O$ is traceless, the reconstruction map $\caM^{-1}(\cdot)$  has a simple form:  $\caM^{-1}(O)=(d+1)O$ \cite{huang2020pred}. In conjunction with \eref{eq:V*} we can deduce that
\begin{align}\label{eq:V*CliffordProof1}
	V_*(O,\Delta)&=(d+1)^2\tr\left[\Omega(\Cl_n)\Tensor{\left(O\otimes\Delta\right)}{2}\right]- \left[\tr(\Delta O)\right]^2,
\end{align}
where $\Omega(\Cl_n)$ is defined according to  \eref{eq:CrossMoment}:
\begin{equation}
	\Omega(\Cl_n)=\sum_{\bfs,\bft}\bbE_{U\sim\Cl_n}\dagtensor{U}{4}\left[\Tensor{\left(|\bfs\>\<\bfs|\right)}{2}\otimes\Tensor{\left(|\bft\>\<\bft|\right)}{2}\right]\Tensor{U}{4}.
\end{equation}

Next, decompose $O$ in terms of Pauli operators: $O=\sum_{i=0}^{d^2-1}\tr(OP_i)P_i/d$. Then   $\tr\left[\Omega(\Cl_n)\Tensor{\left(O\otimes\Delta\right)}{2}\right]$ featured in \eref{eq:V*CliffordProof1} can be expressed as follows:
	\begin{align}
		&\tr\left[\Omega(\Cl_n)\Tensor{\left(O\otimes\Delta\right)}{2}\right]=\sum_{\bfs,\bft}\bbE_{U\sim\Cl_n}\left(\<\bfs|UOU^{\dag}|\bfs\>\<\bfs|U\Delta U^{\dag}|\bfs\>\<\bft|UOU^{\dag}|\bft\>\<\bft|U\Delta U^{\dag}|\bft\>\right)\notag \\
		&=\frac{1}{d^2}\sum_{\bfs,\bft}\sum_{i,j=1}^{d^2-1}\tr(OP_i)\tr(OP_j)\bbE_{U\sim\Cl_n}\left(\<\bfs|UP_iU^{\dag}|\bfs\>\<\bfs|U\Delta U^{\dag}|\bfs\>\<\bft|UP_jU^{\dag}|\bft\>\<\bft|U\Delta U^{\dag}|\bft\>\right).  \label{eq:V*CliffordProof2}
	\end{align}
If 	$U\notin\scrD(P_i)$, then $\<\bfs|UP_iU^{\dag}|\bfs\>=0$ for all $\bfs$. Otherwise, $UP_iU^{\dag}$ is diagonal in the computational basis, and 
\begin{align}
\sum_{\bfs} 	\<\bfs|UP_iU^{\dag}|\bfs\>\<\bfs|U\Delta U^{\dag}|\bfs\>=\tr\left(UP_iU^{\dag}U\Delta U^{\dag}\right)=\tr(\Delta P_i). 
\end{align}
Based on this observation, \eref{eq:V*CliffordProof2} can be simplified as follows:
\begin{align}
	\tr\left[\Omega(\caU)\Tensor{\left(O\otimes\Delta\right)}{2}\right]&=\frac{1}{d^2}\sum_{i,j=1}^{d^2-1}\tr(OP_i)\tr(OP_j)\bbE_{U\sim\Cl_n}\left[\delta_{U\in \scrD(P_i,P_j)}\right]\tr(\Delta P_i)\tr(\Delta P_j)\nonumber\\
&=\frac{1}{d^2}\sum_{\substack{i,j\geq 1\,|\,i \ne j, \\ [P_i, P_j]=0}} \tr(OP_i)\tr(O P_j)\tr(\Delta P_i)\tr(\Delta P_j)\bbE_{U\sim\Cl_n}\left[\delta_{U\in \scrD(P_i,P_j)}\right] \notag\\
&\equad+\frac{1}{d^2}\sum_{i=1}^{d^2-1}\left[\tr(OP_i)\right]^2\left[\tr(\Delta P_i)\right]^2\bbE_{U\sim\Cl_n}\left[\delta_{U\in \scrD(P_i)}\right]\notag\\
		&= \frac{2}{d^2(d+1)(d+2)}\sum_{\substack{i,j\geq 1\,|\,i \ne j, \\ [P_i, P_j]=0}} \tr(OP_i) \tr(OP_j) \tr(\Delta P_i) \tr(\Delta P_j)  \notag \\
		&\equad + \frac{1}{d^2(d+1)} \sum_{i=1}^{d^2-1} \left[\tr\left(OP_i\right)\right]^2 \left[\tr\left(\Delta P_i\right)\right]^2, \label{eq:omega_2terms}
	\end{align}
where the last equality follows from \lref{lem:mixing}. 	
Together with \eref{eq:V*CliffordProof1}, this equation implies \eref{eq:V*CliffordAlt} and completes the proof of	\lref{lem:V*CliffordAlt}. 
\end{proof}

\subsection{Proof of \thref{thm:CliffordVstar}}

\begin{proof}[Proof of \thref{thm:CliffordVstar}] First, we will show that the formula for $\bbV_*(O,\Delta)$ in \eref{eq:Clifford_variance} in \thref{thm:CliffordVstar}
is equivalent to the formula in \eref{eq:V*CliffordAlt} in \lref{lem:mixing}. To this end, we need to expand
$\lVert\Xi_{\Delta,O}\rVert_2^2$ and $\tXi_{\Delta,O}\cdot \Xi_{\Delta,O}$. According to the definitions in \eref{eq:CrossChar}, for $P\in \bcaP_n$ we have 
\begin{equation}
\begin{gathered}
	\Xi_{\Delta,O}(P)=\tr(O P)	\tr(\Delta P),\\
	\tXi_{\Delta,O}(P) = \tr(\Delta P O P)=\frac{1}{d}\sum_{j=1}^{d^2-1}\tr(OP_j) \tr(\Delta P P_j P)=\frac{1}{d}\sum_{j=1}^{d^2-1}f(P,P_j)\tr(OP_j) \tr(\Delta P_j ),
\end{gathered}
\end{equation}
where 
\begin{equation}
	f(P,P_j)=
	\begin{cases}
		1& \text{if }[P,P_j]=0, \\
		\\
		-1& \text{if } \{P,P_j\}=0.
	\end{cases}
\end{equation}
Note that $f(P,P_j)=f(P_j,P)$ by definition.
Therefore,
\begin{equation}
\begin{aligned}
	\lVert\Xi_{\Delta,O}\rVert_2^2&=\sum_{i=1}^{d^2-1}[\tr(OP_i)]^2\left[\tr\left(\Delta P_i\right)\right]^2,\\
	\tXi_{\Delta,O}\cdot\Xi_{\Delta,O}&=\frac{1}{d}\sum_{i,j=1}^{d^2-1}f(P_i,P_j)\tr(O P_i)\tr(O P_j)\tr(\Delta P_i)\tr(\Delta P_j)\\
	&=\frac{1}{d}\sum_i [\tr(OP_i)]^2\left[\tr\left(\Delta P_i\right)\right]^2+\frac{1}{d}\sum_{\substack{i,j\geq 1\,|\,i \ne j, \\ [P_i, P_j]=0}}\tr(OP_i)\tr(OP_j)\tr(\Delta P_i)\tr(\Delta P_j)\\
	&\equad-\frac{1}{d}\sum_{i,j\geq 1\,|\,\{P_i,P_j\}=0}\tr(OP_i)\tr(OP_j)\tr(\Delta P_i)\tr(\Delta P_j).\label{eq:character_prod}
\end{aligned}
\end{equation}
In addition, we have  
\begin{align}
  [\tr(\Delta O)]^2&=\frac{1}{d^2}\sum_{i,j=1}^{d^2-1}\tr(\Delta P_i)\tr(\Delta P_j)\tr(OP_i)\tr(OP_j)\notag\\
    &=\frac{1}{d^2}\sum_i [\tr(OP_i)]^2\left[\tr(\Delta P_i)\right]^2+    
    \frac{1}{d^2}\sum_{\substack{i,j\geq 1\,|\,i \ne j, \\ [P_i, P_j]=0}}\tr(OP_i)\tr(OP_j)\tr(\Delta P_i)\tr(\Delta P_j)\notag\\
    &\equad+\frac{1}{d^2}\sum_{i,j\geq 1\,|\,\{P_i,P_j\}=0}\tr(OP_i)\tr(OP_j)\tr(\Delta P_i)\tr(\Delta P_j).\label{eq:trace_square}
\end{align}
In conjunction with \lsref{lem:V*CliffordAlt} and \ref{lem:mixing} we can deduce that
\begin{align}
    \bbV_*(O,\Delta)&=\frac{2(d+1)}{d^2(d+2)}\sum_{\substack{i,j\geq 1\,|\,i \ne j, \\ [P_i, P_j]=0}} \tr(OP_i) \tr(OP_j) \tr(\Delta P_i) \tr(\Delta P_j)  \nonumber \\
    &\equad + \frac{d+1}{d^2} \sum_{i=1}^{d^2-1} \left[\tr\left(OP_i\right)\right]^2 \left[\tr\left(\Delta P_i\right)\right]^2- \left[\tr(\Delta O)\right]^2\notag\\
    &=\frac{d+1}{d(d+2)}\left(\lVert \Xi_{\Delta,O}\rVert_2^2+\tXi_{\Delta,O}\cdot \Xi_{\Delta,O}\right) -\frac{\left[\tr\left(\Delta O\right)\right]^2}{d+2}\leq   \frac{2}{d}\lVert\Xi_{\Delta,O}\rVert_2^2\le 2\lVert\Delta\rVert_1^2\lVert O\lVert_2^2\le 8\epsilon \lVert O \rVert_2^2, \label{V*Cliffordproof}
        \end{align}
which confirms  \eref{eq:Clifford_variance}.  Here the three inequalities follow from \eref{eq:CharProp1} in  \lref{lem:CharProp}.

If $\sigma$ is a pure state and $O=\sigma-\bbone/d$, then \eref{eq:CliffordVarFE1} follows from
\eref{V*Cliffordproof} above [or \eref{eq:Clifford_variance}] and \eref{eq:CharPropFE} in \lref{lem:CharProp}.  This observation completes the proof of \thref{thm:CliffordVstar}.
\end{proof}

\subsection{\label{SM:VariancesDepol}Variances in HPFE in the presence of depolarizing noise}
In this section we clarify the variances $\bbV(O,\rho)$, $\bbV_*(O,\rho)$, and $\bbV_*(O,\Delta)$ in HPFE in the presence of depolarizing noise and the impacts of the 2-SRE and the number $R$ of circuit reusing on these variances and circuit sample cost.  The discussions in this section also pave the way for proving \thref{thm:CliffordHPFEpn} in the main text.

\begin{lemma}\label{lem:depo_terms}
Suppose $\sigma$ is a pure state on $\caH$, $\rho=(1-p)\sigma+p\bbone/d$, $\Delta=\rho-\sigma$, $\epsilon=1-\tr(\rho\sigma)$, and $O=\sigma-\bbone/d$. Then the variances $\bbV(O,\rho)$, $\bbV_*(O,\rho)$, and $\bbV_*(O,\Delta)$ in $\CRM(\Cl_n,\sigma,R)$ read 
	\begin{align}
		\bbV(O,\rho)&=\left(\frac{d-1}{d}\right)^2\left[-p^2+\frac{4dp}{(d-1)(d+2)}+\frac{d^2-3d-2}{(d-1)(d+2)}\right]+\frac{d^2-1}{d^2} \le \frac{2d(d+1)}{(d+2)^2},\label{eq:hpfe_depo_thr}\\
		\bbV_*(O,\rho)&=\frac{(1-p)^2\left[2^{1-M_2(\sigma)}(d+1)-4\right]}{d+2}, \label{eq:hpfe_depo_thr*}\\
		\bbV_*(O,\Delta)&=\frac{p^2\left[2^{1-M_2(\sigma)}(d+1)-4\right]}{d+2}\le2^{1-M_2(\sigma)}\epsilon^2.\label{eq:hpfe_depo_crm*}
	\end{align}
\end{lemma}

Here \eqsref{eq:hpfe_depo_thr}{eq:hpfe_depo_thr*} are known before \cite{chen2024nonstab}. According to \lref{lem:VRCRM}, the variances in THR and CRM shadow estimation of the  observable $O$ are respectively given by 
\begin{align}
	\bbV_R^{(\rmT)}(\hO)&=\bbV_*(O,\rho)+\frac{\bbV(O,\rho)-\bbV_*(O,\rho)}{R},  \quad
	\bbV_R^{(\rmC)}(\hO)=\bbV_*(O,\Delta)+\frac{\bbV(O,\rho)-\bbV_*(O,\rho)}{R}.  \label{eq:VRCRM}
\end{align}
In HPFE, usually $\rho$ is close to the target state $\sigma$, that is, $p\ll1 $, in which case  \lref{lem:depo_terms} implies that 
	\begin{align}
	\bbV(O,\rho)&\approx \bbV(O,\sigma)\approx 2,\quad 
	\bbV_*(O,\rho)\approx \bbV_*(O,\sigma)\approx 2^{1-M_2(\sigma)},\quad 
	\bbV_*(O,\Delta)\approx 2^{1-M_2(\sigma)}\epsilon^2, \label{eq:VariancesHPapprox}
\end{align}
which in turn implies that
\begin{align}
	\bbV_R^{(\rmT)}(\hO)&\approx 2\left(2^{-M_2(\sigma)}+\frac{1-2^{-M_2(\sigma)}}{R}\right), \quad 
	\bbV_R^{(\rmC)}(\hO)\approx 2\left(2^{-M_2(\sigma)}\epsilon^2+\frac{1-2^{-M_2(\sigma)}}{R}\right).  \label{eq:VRTHRCRM}
\end{align}
In conjunction with \lref{lem:upperbound_samp} we can derive the circuit sample costs required in THR and CRM shadow estimation:
\begin{align}
	N_U^{(\rmT)}&= \caO\left(\frac{2^{-M_2(\sigma)}}{r^2\epsilon^2}+\frac{1-2^{-M_2(\sigma)}}{Rr^2\epsilon^2}\right), \quad 
	N_U^{(\rmC)}= \caO\left(\frac{2^{-M_2(\sigma)}}{r^2}+\frac{1-2^{-M_2(\sigma)}}{Rr^2\epsilon^2}\right).  \label{eq:NUTHRCRM}
\end{align}
In the large-$R$ limit ($R\gg 2^{M_2(\sigma)}/\epsilon^2$), the above formulas simplify to
\begin{align}
	N_U^{(\rmT)}&= \caO\left(\frac{2^{-M_2(\sigma)}}{r^2\epsilon^2}\right), \quad 
	N_U^{(\rmC)}= \caO\left(\frac{2^{-M_2(\sigma)}}{r^2}\right).  \label{eq:NUTHRCRMlim}
\end{align}
Both $N_U^{(\rmT)}$ and $N_U^{(\rmC)}$ decrease monotonically with the 2-SRE $M_2(\sigma)$. On the other hand,
 $N_U^{(\rmT)}$ scales as $1/\epsilon^2$, while $N_U^{(\rmC)}$ is independent of $\epsilon$. In THR, when $R\geq 5\times 2^{M_2(\sigma)}$, further increasing $R$ has negligible influence on $N_U^{(\rmT)}$. In CRM, by contrast, $R$ needs to scale with $1/\epsilon^2$ to guarantee that $N_U^{(\rmC)}$ is independent of $\epsilon$ as shown in \eref{eq:NUTHRCRMlim}. If $R$ is too small, then the second term in the expression of $N_U^{(\rmC)}$ in \eref{eq:NUTHRCRM} will dominate, and $N_U^{(\rmC)}$ will scale with $1/\epsilon^2$, just like $N_U^{(\rmT)}$. 
 Moreover, when $R<1/\epsilon^2$, $N_U^{(\rmC)}$ may increase (rather than decrease) with $M_2(\sigma)$, so the impact of the 2-SRE on the circuit sample cost of CRM shadow estimation is quite subtle.

When depolarizing noise is replaced by general Pauli noise, it is difficult to derive precise and informative formulas for the variances as shown above. Nevertheless, the above conclusions still hold approximately, which is confirmed by extensive numerical calculations presented in \sref{SM:SimulationClifford}.

\begin{proof}[Proof of \lref{lem:depo_terms}]
\Eqsref{eq:hpfe_depo_thr}{eq:hpfe_depo_thr*} follow from Propositions 3 and 4 in \rcite{chen2024nonstab}, so it remains to prove \eqref{eq:hpfe_depo_crm*}.
By assumption  $\Delta=\rho-\sigma = p(\bbone/d-\sigma)$ and $O=\sigma-\bbone/d$, where $\sigma$ is a pure state. Therefore, 
\begin{equation}
	\begin{aligned}
	\tr(\Delta \sigma) &=\tr(\Delta O) = 		-p\tr\left[\left(\sigma-\frac{\bbone}{d}\right)^2\right]
	=-p\left(1-\frac{1}{d}\right),\\
	\epsilon&=1-\tr(\rho \sigma)=-\tr(\Delta \sigma)=p\left(1-\frac{1}{d}\right).
\end{aligned} 
\end{equation} 
Now, from the definitions of cross characteristic functions in \eref{eq:CrossChar}  we can deduce that   
\begin{equation}
\begin{aligned}
	\lVert\Xi_{\Delta,O}\rVert_2^2&=\sum_{i=1}^{d^2-1}[\tr(\Delta P_i)]^2[\tr(OP_i)]^2=p^2\sum_{i=1}^{d^2-1} [\tr(\sigma P_i)]^4=p^2\left[ 2^{-M_2(\sigma)}d-1\right],\\
	\tXi_{\Delta,O}\cdot \Xi_{\Delta,O}&=\sum_{i=1}^{d^2-1} \tr(\Delta P_iOP_i)\tr(O P_i)\tr(\Delta P_i)=p^2\sum_{i=1}^{d^2-1}
	\left[\tr(\sigma P_i \sigma P_i)-\frac{1}{d}\right][\tr(\sigma P_i)]^2 \\
	&=p^2\sum_{i=1}^{d^2-1} \left[\tr(\sigma P_i)\right]^4-\frac{p^2}{d}\sum_{i=1}^{d^2-1} [\tr(\sigma P_i)]^2=p^2\left[2^{-M_2(\sigma)}d-2+\frac{1}{d}\right].
\end{aligned}
\end{equation}   
In conjunction with \eqsref{eq:Vcirc}{eq:Clifford_variance} we can prove \eref{eq:hpfe_depo_crm*} as follows:
	\begin{align}
		\bbV_*(O,\Delta)&=\frac{d+1}{d(d+2)}\left(\lVert\Xi_{\Delta,O}\rVert_2^2+\tXi_{\Delta,O}\cdot \Xi_{\Delta,O}\right)-\frac{[\tr(\Delta O)]^2}{d+2}\notag\\
		&=\frac{d+1}{d(d+2)}p^2\left[2^{1-M_2(\sigma)}d-3+\frac{1}{d}\right]-\frac{p^2(d-1)^2}{d^2(d+2)}=\frac{p^2\left[2^{1-M_2(\sigma)}(d+1)-4\right]}{d+2}
	\nonumber\\
	&\leq\frac{2^{1-M_2(\sigma)}p^2(d-1)}{d+2}
=\frac{2^{1-M_2(\sigma)}d^2}{(d-1)(d+2)}\epsilon^2	
\le 2^{1-M_2(\sigma)}\epsilon^2,
	\end{align}
where the first inequality holds because $M_2(\sigma)\geq 0$, and the second inequality holds because $d\geq 2$, which means $d^2\leq (d-1)(d+2)$. This observation completes the proof of \lref{lem:depo_terms}. 
\end{proof}

\subsection{\label{SM:PauliGeneral}Proof of \thref{thm:CliffordHPFEpn}}

\begin{proof}[Proof of \thref{thm:CliffordHPFEpn}]
	By assumption $\sigma$ is a pure state and  $\rho=\scrP(\sigma)$, where $\scrP$ is a Pauli channel. In conjunction with \thref{thm:CliffordVstar} and \lref{lem:InfidCrossCharUB} we can deduce that 
	\begin{align}
		\bbV_*(O,\Delta)& \le\frac{2}{d}\lVert\Xi_{\Delta,O}\rVert_2^2\le 4\epsilon^2.
	\end{align}
	So \eref{eq:hpfe_pauligen} in \thref{thm:CliffordHPFEpn} is a simple corollary of \lref{lem:upperbound_samp2}. 
	
	Next, suppose $\scrP$ is a depolarizing channel.  Then $\bbV_*(O,\Delta)\leq 2^{1-M_2(\sigma)}\epsilon^2$ according to \lref{lem:depo_terms}. So \eref{eq:hpfe_depo}  is also a simple corollary of \lref{lem:upperbound_samp2}, which completes the proof of \thref{thm:CliffordHPFEpn}.
\end{proof}

\section{Variances in CRM shadow estimation based on Pauli measurements}

To complement the results on 4-design and Clifford measurements, here we derive general analytical formulas for the variances $\bbV(O,\rho)$, $\bbV_*(O,\rho)$, and $\bbV_*(O,\Delta)$ in CRM shadow estimation based on  Pauli measurements. Although this topic has been studied before \cite{vermersch2024enhanced}, such general analytical formulas were not known before. Now, the underlying unitary ensemble $\caU$ corresponds to the uniform distribution on the $n$-qubit local Clifford group $\Tensor{\Cl_1}{n}$.

Let  $P_i$ and $P_j$ be two Pauli operators in $\bcaP_n$. The two Pauli operators  are locally commutative, denoted by 
$P_i\diamond P_j$, if their Pauli factors on each qubit are commutative. The weight (also known as the locality) of $P_i$, denoted by $w(P_i)$, is  the number of nontrivial (single-qubit) Pauli factors in $P_i$, that is, the number of qubits on which $P_i$ acts nontrivially. As generalization, we denote by $w(P_i\cap P_j)$  the number of common nontrivial  Pauli factors in $P_i$ and $P_j$  and denote by $w(P_i \cup P_j)$ the number of qubits on which $P_i$ or $P_j$ acts nontrivially. By definition it is easy to verify that  $w(P_i \cup P_j)=w(P_i)+w(P_j)-w(P_i\cap P_j)$. In addition, define $\scrD_\lc(P_i)$ as the set of unitaries in $\Tensor{\Cl_1}{n}$ that can diagonalize $P_i$ with respect to the computational basis; define $\scrD_\lc(P_i,P_j)$  as the set of unitaries in $\Tensor{\Cl_1}{n}$ that can diagonalize $P_i$ and $P_j$ simultaneously.  Note that $\scrD_\lc(P_i,P_j)$ is not empty only when $P_i \diamond P_j$.

\begin{theorem}\label{thm:PauliCRMVar}
    Suppose $\caU=\Tensor{\Cl_1}{n}$, $O$ is a traceless observable on $\caH$, and $\rho\in \caD(\caH)$ is the system state. Then the variances $\bbV(O,\rho)$, $\bbV_*(O,\rho)$, and $\bbV_*(O,\Delta)$ in $\CRM(\caU,\sigma,R)$ read
\begin{align}
    \bbV(O,\rho)&=\sum_{i,j>1\,|\, P_i\diamond P_j}^{d^2-1}\frac{3^{w(P_i \cap P_j )}}{d^2}\tr(OP_i)\tr(OP_j)\tr(\rho P_i P_j)-\left[\tr(O\rho)\right]^2,  \label{eq:PauliV} \\	
    \bbV_*(O,\rho)&=\sum_{i,j>1\,|\, P_i\diamond P_j}^{d^2-1}\frac{3^{w(P_i \cap P_j )}}{d^2}\tr(OP_i)\tr(OP_j)\tr(\rho P_i)\tr(\rho P_j)-\left[\tr(O\rho)\right]^2, \label{eq:PauliV*Rho}\\
     \bbV_*(O,\Delta)&=\sum_{i,j>1\,|\, P_i\diamond P_j}^{d^2-1}\frac{3^{w(P_i \cap P_j )}}{d^2}\tr(OP_i)\tr(OP_j)\tr(\Delta P_i)\tr(\Delta P_j)-\left[\tr(O\Delta)\right]^2. \label{eq:PauliV*Delta}  
\end{align}
\end{theorem}

Before proving \thref{thm:PauliCRMVar}, we first introduce an auxiliary lemma, which is straightforward to verify (cf. \lref{lem:mixing}). 

\begin{lemma}\label{lem:PauliMappingPauli}
Suppose $P_i$ and $P_j$ are two Pauli operators in $\bcaP_n$ and $P_i \diamond P_j$. Then
        \begin{align}
       \bbE_{U\sim\Tensor{\Cl_1}{n}}\left[\delta_{U \in \scrD_\lc(P_i)}\right] & = \frac{1}{3^{w(P_i)}}, \quad 
        \bbE_{U\sim\Tensor{\Cl_1}{n}}\left[\delta_{U \in \scrD_\lc(P_i, P_j)}\right] = \frac{1}{3^{w(P_i\cup P_j)}}.
    \end{align} 
\end{lemma}

\begin{proof}[Proof of \thref{thm:PauliCRMVar}]
Let $\caM^{-1}(\cdot)$ be the reconstruction map associated with Pauli measurements. By assumption $O$ is traceless,  so $\caM^{-1}(O)$ can be expressed as follows:
    \begin{equation}
        \caM^{-1}(O)=\sum_{i=1}^{d^2-1}\frac{3^{w(P_i)}}{d}\tr(O P_i)P_i. \label{eq:ReconPauli}
    \end{equation}
In conjunction with the definition of $\bbV(O,\rho)$ \cite{huang2020pred} we can derive that
\begin{align}
	\bbV(O,\rho)&=\bbE_{U\sim \Tensor{\Cl_1}{n}}\left(\sum_{\bfs}P_{\rho}(\bfs\mid U)\left\{\left[\caM^{-1}(O)\right](U,\bfs)\right\}^2\right)-\left[\tr(O\rho)\right]^2\notag\\
	&=\sum_{i,j=1}^{d^2-1}\frac{3^{w(P_i)+w(P_j)}}{d^2}\tr(OP_i)\tr(OP_j)\sum_{\bfs}\bbE_{U\sim \Tensor{\Cl_1}{n}}\left(\<\bfs|U \rho U^\dag|\bfs\>\<\bfs|U P_i U^\dag|\bfs\>\<\bfs|U P_j U^\dag|\bfs\>\right)-\left[\tr(O\rho)\right]^2.
\end{align}
If 	$U\notin\scrD_\lc(P_i,P_j)$, then $\<\bfs|UP_iU^{\dag}|\bfs\>\<\bfs|UP_jU^{\dag}|\bfs\>=0$ for all $\bfs$; otherwise, $UP_iU^{\dag}$ and $UP_j U^{\dag}$ are diagonal in the computational basis, and we have 
\begin{equation}
	\sum_{\bfs} 	\<\bfs|U \rho U^\dag|\bfs\>\<\bfs|U P_i U^\dag|\bfs\>\<\bfs|U P_j U^\dag|\bfs\>=\tr\left(U\rho U^{\dag}UP_iU^{\dag}UP_jU^{\dag}\right)=\tr(\rho P_i P_j). 
\end{equation}
Based on this observation, $\bbV(O,\rho)$ can be expressed as follows:
\begin{align}
	\bbV(O,\rho)&=\sum_{i,j>1\,|\, P_i\diamond P_j}^{d^2-1}\frac{3^{w(P_i)+w(P_j)}}{d^2}\tr(OP_i)\tr(OP_j)\tr(\rho P_i P_j)\bbE_{U\sim \Tensor{\Cl_1}{n}}\left[\delta_{U\in\scrD_\lc(P_i,P_j)}\right]-\left[\tr(O\rho)\right]^2\notag\\
	&=\sum_{i,j>1\,|\, P_i\diamond P_j}^{d^2-1}\frac{3^{w(P_i \cap P_j )}}{d^2}\tr(OP_i)\tr(OP_j)\tr(\rho P_i P_j)-\left[\tr(O\rho)\right]^2,
\end{align}
which confirms \eref{eq:PauliV}. 
Here the last equality follows from \lref{lem:PauliMappingPauli}.

Next, let $\tau$ be an arbitrary linear operator on $\caH$. By virtue of \eqsref{eq:V*}{eq:ReconPauli} we can derive that 
\begin{align}
    &\bbV_*(O,\tau)=\sum_{\bfs,\bft}\bbE_{U\sim\Tensor{\Cl_1}{n}}\left[\<\bfs|U \caM^{-1}(O) U^\dag|\bfs\>\<\bfs|U \tau U^\dag|\bfs\>\<\bft|U \caM^{-1}(O)  U^\dag|\bft\>\<\bft|U \tau U^\dag|\bft\>\right]-[\tr(O\tau)]^2\notag\\
    \quad&=\sum_{i,j=1}^{d^2-1}\frac{3^{w(P_i)+w(P_j)}}{d^2}\tr(OP_i)\tr(OP_j)\sum_{\bfs,\bft}\bbE_{U\sim\Tensor{\Cl_1}{n}}\left[\<\bfs|U P_i U^\dag|\bfs\>\<\bfs|U \tau U^\dag|\bfs\>\<\bft|U P_j U^\dag|\bft\>\<\bft|U \tau U^\dag|\bft\>\right]\nonumber\\
    &\equad-[\tr(O\tau)]^2.
\end{align}
If 	$U\notin\scrD_\lc(P_i)$, then $\<\bfs|UP_iU^{\dag}|\bfs\>=0$ for all $\bfs$; otherwise, $UP_iU^{\dag}$ is diagonal in the computational basis, and 
\begin{align}
\sum_{\bfs} 	\<\bfs|UP_iU^{\dag}|\bfs\>\<\bfs|U\tau U^{\dag}|\bfs\>=\tr\left(UP_iU^{\dag}U\tau U^{\dag}\right)=\tr(\tau P_i). 
\end{align}
Based on this observation, $\bbV_*(O,\tau)$ can be expressed as follows:
\begin{align}
    \bbV_*(O,\tau)&=\sum_{i,j=1 \,|\, P_i\diamond P_j}^{d^2-1}\frac{3^{w(P_i)+w(P_j)}}{d^2}\tr(OP_i)\tr(OP_j)\tr(\tau P_i)\tr(\tau P_j)\bbE_{U\sim\Tensor{\Cl_1}{n}}\left[\delta_{U\in\scrD_\lc(P_i, P_j)}\right]-[\tr(O\tau)]^2 \notag\\
    &=\sum_{i,j>1\,|\, P_i\diamond P_j}^{d^2-1}\frac{3^{w(P_i \cap P_j )}}{d^2}\tr(OP_i)\tr(OP_j)\tr(\tau P_i)\tr(\tau P_j)-[\tr(O\tau)]^2,
\end{align}
where the last equality follows from \lref{lem:PauliMappingPauli}. This equation implies  \eqsref{eq:PauliV*Rho}{eq:PauliV*Delta} and completes the proof of  \thref{thm:PauliCRMVar}. 
\end{proof}

\section{\label{SM:NumericsAdd}Additional numerical results}

In this section, we first provide additional numerical results on the performance of THR and CRM shadow estimation based on Clifford measurements and discuss the impacts nonstabilizerness and the number $R$ of circuit reusing on the circuit sample cost. Then we compare the performance of Clifford measurements with  4-design and Pauli measurements.

Thanks to \lsref{lem:upperbound_samp} and \ref{lem:VRCRM}, the circuit sample cost $N_U$ is determined by the variance $\bbV_R(\hO)$, which in turn is  determined by $\bbV(O,\rho)$, $\bbV_*(O,\rho)$, and $\bbV_*(O,\Delta)$. When the measurement ensemble forms a {3-design} (including Clifford and 4-design measurements), $\bbV(O,\rho)$ is determined by \eref{eq:3designVar}. For Clifford  measurements, $\bbV_*(O,\Delta)$ is determined by \eref{eq:Clifford_variance}, and $\bbV_*(O,\rho)$ is determined by \eref{eq:Clifford_variance} with $\Delta$ replaced by $\rho$; in the presence of depolarizing noise, we can also use more explicit formulas in \lref{lem:depo_terms}. 
For 4-design measurements, $\bbV_*(O,\rho)$ and $\bbV_*(O,\Delta)$ are determined by
\lref{lem:V*4design}. For Pauli measurements, $\bbV(O,\rho)$, $\bbV_*(O,\rho)$, and $\bbV_*(O,\Delta)$ are determined by \thref{thm:PauliCRMVar} (see also \eref{eq:PauliCRMPauliOp} below).

\subsection{\label{SM:SimulationClifford}CRM shadow estimation based on  Clifford measurements}

\begin{figure}[b]
	\centering
	\includegraphics[width=0.6\columnwidth]{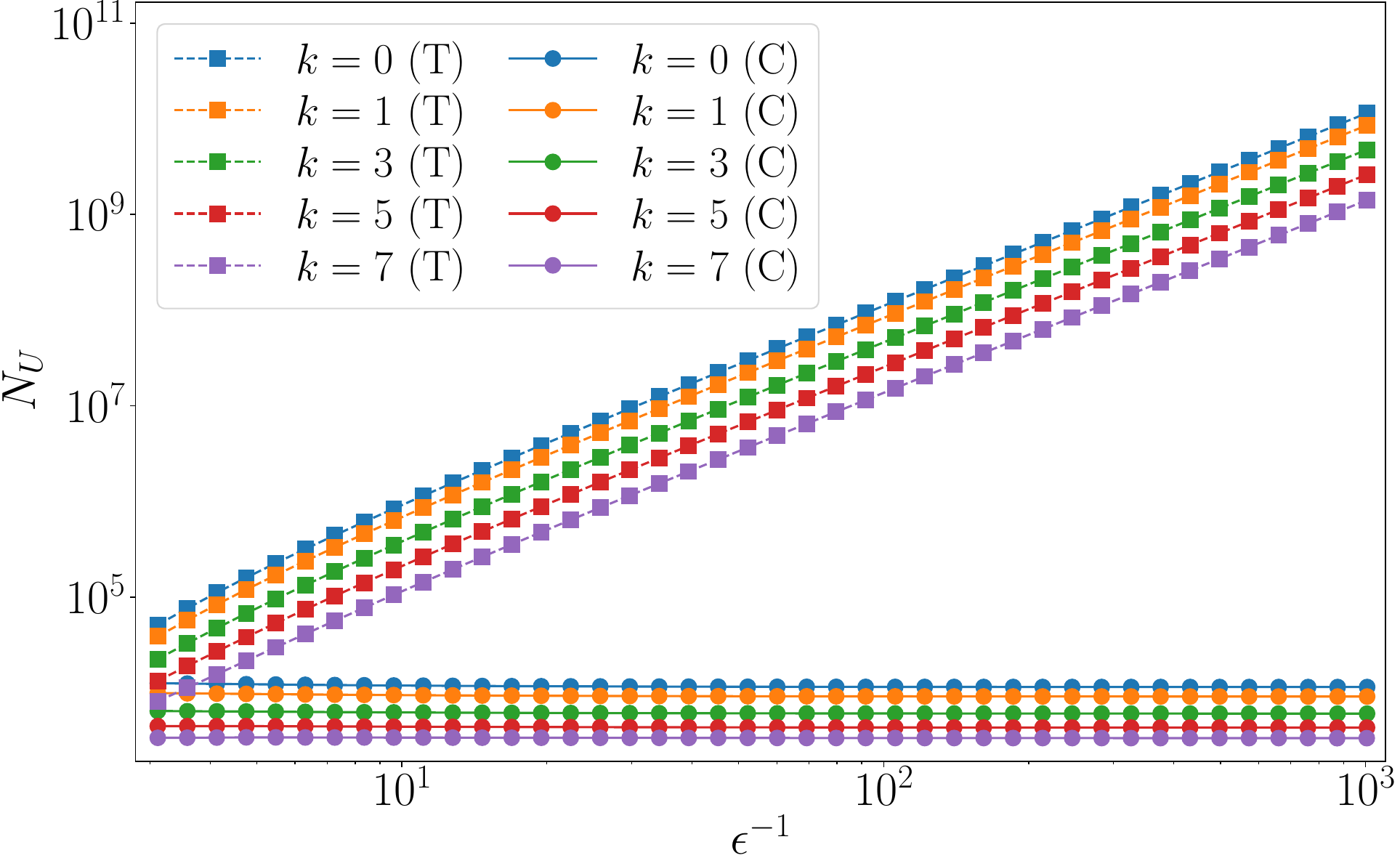}
	\caption{The  circuit sample costs  $N_U$ required for HPFE in THR (T) and  CRM (C) shadow estimation based on  Clifford measurements. Here, $ r = 0.25 $, $ \delta = 0.01$, and $R = \lceil 10/\epsilon^2 \rceil$. The target and prior state has the form $\sigma = |S_{7,k}\>\<S_{7,k}|$ with $k=0,1,3,5,7$, and different system states $\rho$ are generated by applying depolarizing noise to the target state and have the form $\rho= (1-p)\sigma+p\bbone/d$ with $p=\epsilon/(1-d^{-1})$. }\label{fig:depo_bench}
\end{figure}

\subsubsection{HPFE in the presence of depolarizing noise}

Here we consider HPFE in the presence of depolarizing noise. This case is particularly instructive because exact and informative formulas for the variances $\bbV(O,\rho)$, $\bbV_*(O,\rho)$, and $\bbV_*(O,\Delta)$ can be derived as shown in  \lref{lem:depo_terms}. Accordingly, we can derive a simple upper bound for the circuit sample cost as shown in \eref{eq:hpfe_depo} in \thref{thm:CliffordHPFEpn} [see also \eref{eq:NUTHRCRM}]. In numerical simulation, we choose $|\phi\>=|S_{7,k}\>$ as the target and prior state, where $k=0,1,3,5,7$.  \Fref{fig:depo_bench} illustrates the circuit sample costs associated with THR and  CRM shadow estimation as functions of the target infidelity $\epsilon$, where each circuit is reused $R = \lceil  10 /\epsilon^2 \rceil$ times. Notably, CRM only requires a constant circuit sample cost, while 
THR requires $\caO(1/\epsilon^2)$ circuits. This result highlights the significant advantages of CRM over THR. Meanwhile, in both CRM and THR, the circuit sample costs decrease exponentially with the 2-SRE $M_2(\phi)$ of the target state $|\phi\>$. Similar results also hold when depolarizing noise is replaced by general Pauli noise, as illustrated in the right plot in \fref{fig:hpfe_together} (see also \thref{thm:CliffordHPFEpn}), given that random Pauli noise has a similar effect to  the depolarizing noise.

\subsubsection{HPFE in the presence of a special type of coherent noise}

\begin{figure}[b]
	\centering
	\includegraphics[width=0.6\columnwidth]{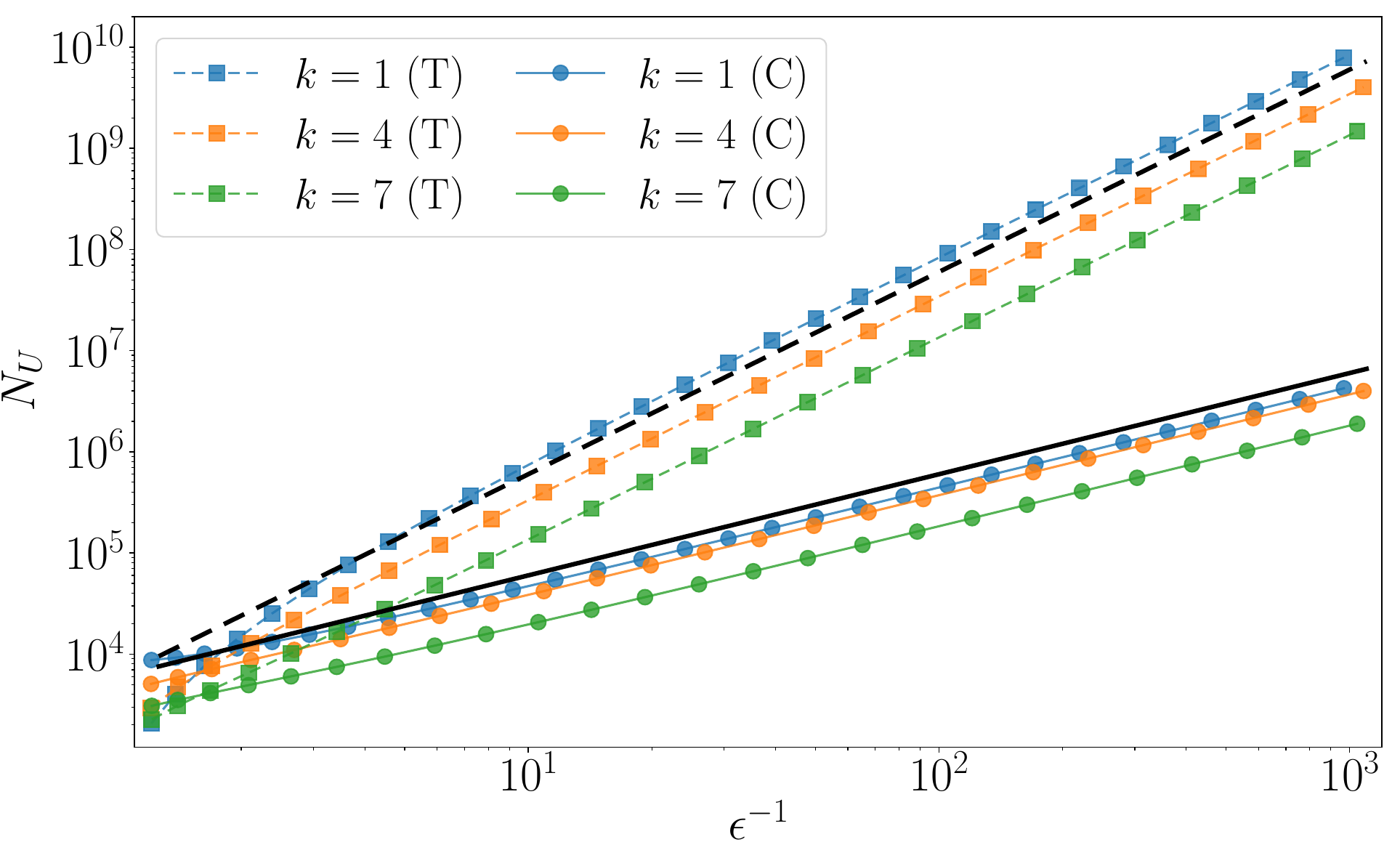}
	\caption{The  circuit sample costs  $N_U$ required for HPFE in THR (T) and CRM (C) shadow estimation based on  Clifford measurements. Here, $ r = 0.25 $, $ \delta = 0.01$, and $R = \lceil 10/\epsilon^2 \rceil$. The target and prior state has the form $\sigma=|S_{7,k}\> \<S_{7,k}|$ with $k=1,4,7$. Different system  states $\rho$ are generated by applying  random  local rotations described in \eref{eq:CollectiveRotations} to the target state $\sigma$. The black solid and dashed lines represent $N_U =6000/\epsilon $ and $N_U =6000 /\epsilon^2$, respectively. }\label{fig:hpfe_u_p3}
\end{figure}

To better understand the scaling behavior of the circuit sample cost in CRM shadow estimation based on Clifford measurements, here we consider HPFE in the presence of a special type of coherent noise, which corresponds to a collective rotation described as follows:
\begin{equation}
    U=\Tensor{\left[U(\theta)\right]}{n} , \quad U(\theta)=   \begin{pmatrix}
    \rme^{-\rmi\theta}\cos(\theta/2) & -\sin(\theta/2) \\
        \sin(\theta/2) & \rme^{\rmi\theta}\cos(\theta/2)
    \end{pmatrix}.\label{eq:CollectiveRotations}
\end{equation}
In the numerical simulation, we choose $|\phi\>=|S_{7,k}\>$ with $k=1,4,7$ as  the target and prior states. 
\Fref{fig:hpfe_u_p3} illustrates the circuit sample costs associated with THR and CRM shadow estimation based on Clifford measurements as functions of the infidelity $\epsilon$, 
where each circuit is reused $R = \lceil 10 /\epsilon^2 \rceil$ times. 
As expected, the circuit sample cost of THR shadow estimation scales as $1/\epsilon^2$. 
By contrast, the  cost of CRM shadow estimation scales as $1/\epsilon$,
which is compatible with the prediction of \thref{thm:CliffordHPFE} given that $\lVert\Delta\rVert_2^2 = 2\epsilon$ for coherent noise by \pref{pro:pchannel_haarSim}. 
Note that CRM shadow estimation is substantially better than THR shadow estimation, but fails to achieve the constant scaling behavior observed in \fref{fig:hpfe_together}. 
This result demonstrates that the $1/\epsilon$ scaling behavior is inevitable for certain  coherent noise although a constant scaling behavior can usually be achieved in the presence of random coherent noise.

\subsubsection{Impact of nonstabilizerness}\label{SM:linear_nonstab}

\begin{figure*}[t]
	\centering\includegraphics[width= 0.9\linewidth]{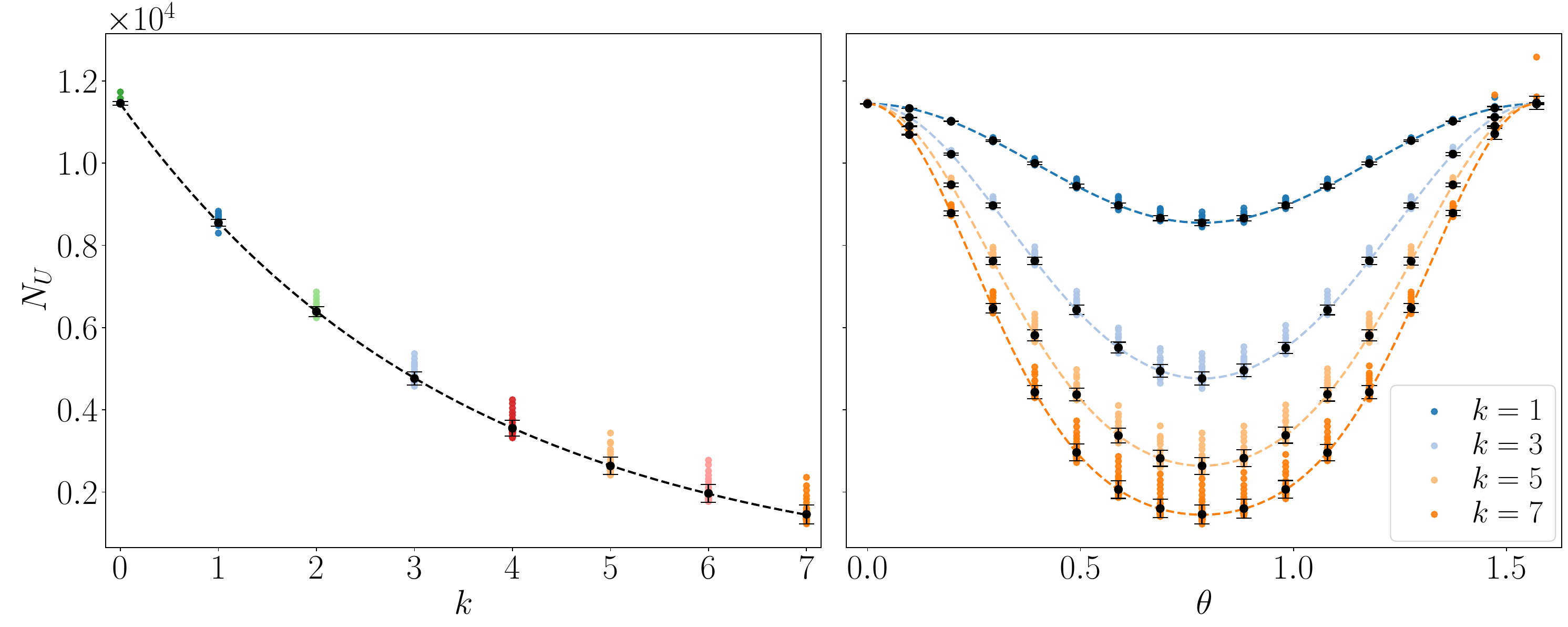}
	\caption{The  circuit sample cost  $N_U$ required for HPFE in CRM  shadow estimation based on  Clifford measurements. Here, $ r = 0.25 $, $ \delta = 0.01$, and $R = \lceil 10/\epsilon^2 \rceil$. The target and prior state has the form $\sigma = |S_{7,k}(\theta)\>\<S_{7,k}(\theta)|$, and different system  states $\rho$ are generated by applying random Pauli noise  (described in \sref{sec:DescriptRandomPauli}), sampled 50 times for given values of $k$ and $\theta$, to the target state $\sigma$. The fitting curves follow the formula in \eref{eq:NUfitting}.
In the left plot, we set $\theta = \pi/4$, and the fitting parameters read $a = 11538.21$ and $b = 94.87$; in the right plot, the fitting parameters depend on $k$ and satisfy $a\in[11537,11571]$ and $b\in[93,128]$. 
}\label{fig:SRE}
\end{figure*}

According to \Thsref{thm:CliffordHPFE}{thm:CliffordHPFEpn},
nonstabilizerness can significantly reduce the circuit sample cost required for HPFE in CRM shadow estimation based on the Clifford group, as illustrated in \fsref{fig:hpfe_together}{fig:depo_bench}.  To further clarify the impact of nonstabilizerness, here we consider HPFE in which the target and prior state has the form
\begin{equation}
    |S_{n,k}(\theta)\>:=\Tensor{|0\>}{(n-k)}\otimes\Tensor{\left[\frac{1}{\sqrt{2}}\left(|0\>+\rme^{\rmi\theta}|1\>\right)\right]}{k},
\end{equation}
whose 2-SRE reads
   \begin{equation}
       M_2(S_{n,k}(\theta))=-k\log_2\left[\frac{7+\cos(4\theta)}{8}\right].\label{eq:nonstab_theta}
   \end{equation}
In addition, the system state is affected by general Pauli noise. 
\Fref{fig:SRE} illustrates the dependence of $N_U$ on the parameters $k$ and $\theta$. When $r=0.25$, $\delta=0.01$, and $R = \lceil 10/\epsilon^2 \rceil$, a pretty good  fitting formula for $N_U$ can be expressed as follows: 
\begin{align}\label{eq:NUfitting}
	N_U = a \times 2^{-M_2(S_{7,k}(\theta))} - b,\quad a\approx 11500, \quad b\approx 100\ll a,
\end{align}
which is similar to \eref{eq:hpfe_depo} in \thref{thm:CliffordHPFEpn}. Notably, $N_U$ decreases exponentially with $M_2(S_{7,k}(\theta))$, and the effect of Pauli noise resembles the counterpart of depolarizing noise. This result also echoes the fact that nonstabilizerness can enhance the performance of THR shadow estimation \cite{chen2024nonstab}.

\subsubsection{Dependence of $N_U$ on the number $R$ of circuit reusing}\label{SM:scaling_with_R}

\begin{figure*}[t]
	\centering\includegraphics[width= 0.9\linewidth]{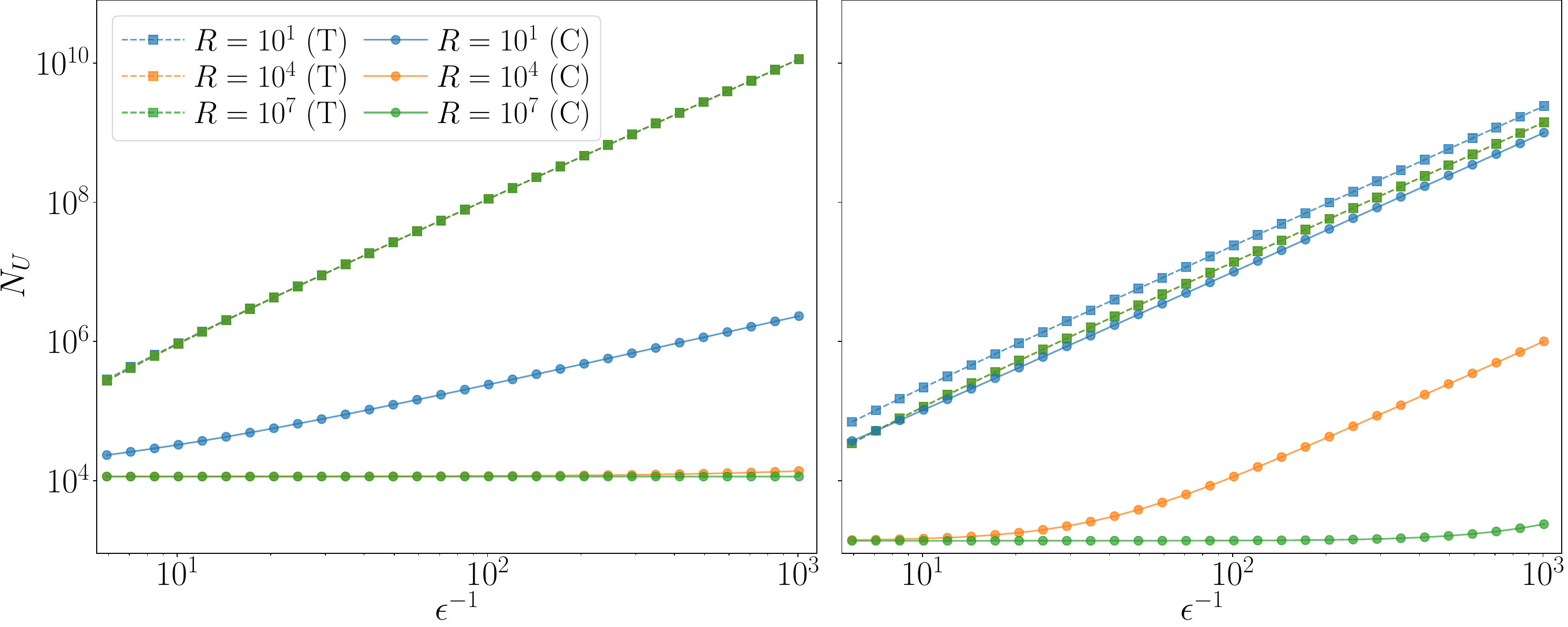}
	\caption{The  circuit sample costs  $N_U$ required for HPFE in THR (T) and CRM (C) shadow estimation based on  Clifford measurements. Here, $ r = 0.25 $ and $ \delta = 0.01 $. The target and prior state has the form $\sigma=|S_{7,k}\> \<S_{7,k}|$ with $k=0$ (left plot) and $k=7$ (right plot), and different system  states $\rho$ are generated by applying  random Pauli channels (described in \sref{sec:DescriptRandomPauli}) to the target state $\sigma$.   For THR, the dependence of  $N_U$ on $R$ is invisible for all three choices of $R$ in the left plot and for two choices of $R$ ($R=10^4$ and $R=10^7$) in the right plot. 
	}\label{fig:N_M7}
\end{figure*}

\begin{figure}
	\centering
	\includegraphics[width=0.6\columnwidth]{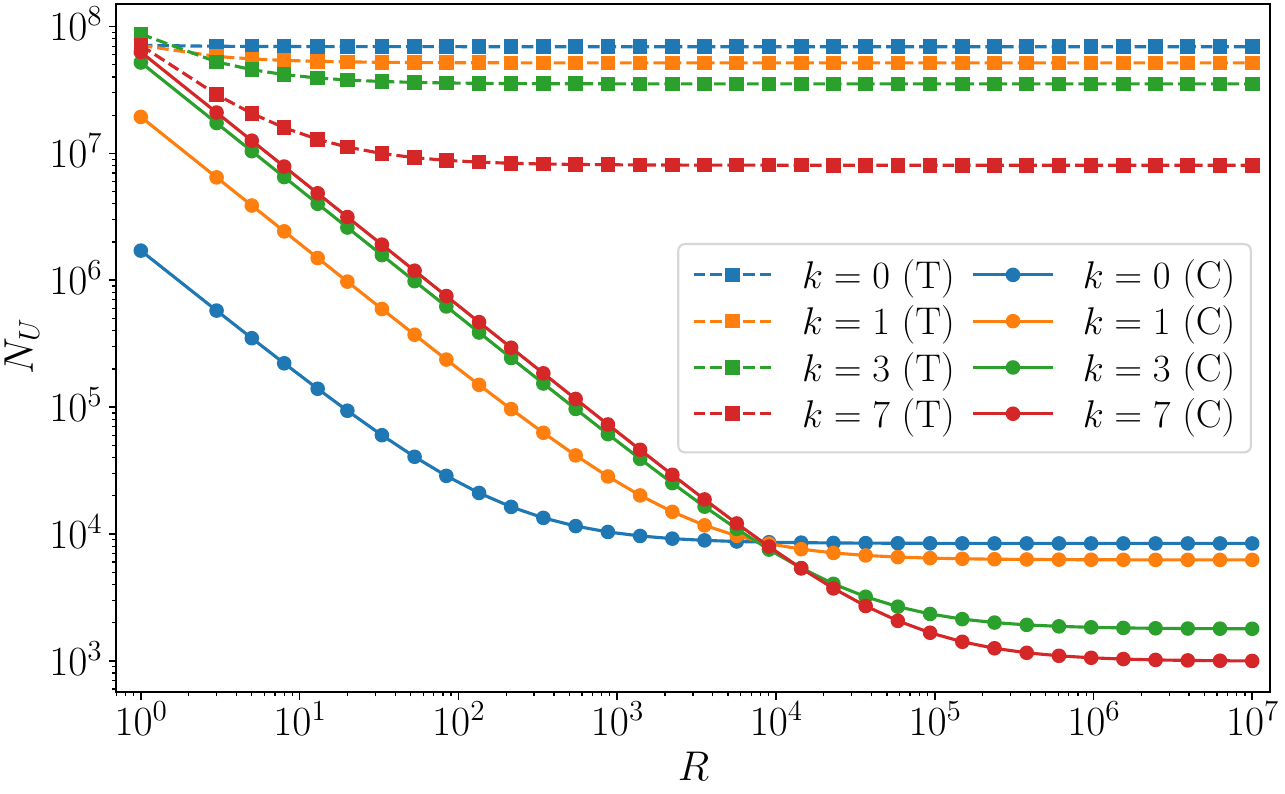}
	\caption{The  circuit sample costs  $N_U$ required for HPFE in THR (T) and CRM (C) shadow estimation based on Clifford measurements. Here, $ r = 0.25 $, $ \delta = 0.01 $, and $\epsilon=0.01$. The target and prior state is $\sigma=|S_{7,k}\> \<S_{7,k}|$ with  $k=0,1,3,7$. For each value of $k$, the system state $\rho$ is generated by applying a random Pauli channel (as described in \sref{sec:DescriptRandomPauli}) to the target state $\sigma$.}\label{fig:varing_NM}
\end{figure}

In this section, we turn to  the dependence of the circuit sample costs on the number $R$ of circuit reusing in THR and CRM shadow estimation based on the Clifford group. In the numerical simulation,  the target and prior state has the form $|\phi\>=|S_{7,k}\>$ with $k=0,1,3,7$, and the system state $\rho$ is generated by applying a random Pauli channel (as described in \sref{sec:DescriptRandomPauli}) to the target state. \Fref{fig:N_M7} illustrates the scaling behaviors of $N_U$ with respect to the infidelity $\epsilon$ for three different choices of $R$, and \fref{fig:varing_NM} illustrates the variations  of $N_U$ with $R$ directly for different choices of $k$. Not surprisingly, increasing $R$ tends to decrease $N_U$ in both THR and CRM shadow estimation. In THR, when $R\geq 5\times 2^{M_2(\phi)}$, however, further increasing $R$ has negligible influence on $N_U^{(\rmT)}$. In CRM, by contrast,  $N_U^{(\rmC)}$ decreases with $R$ until  $R\gtrsim 2^{M_2(\phi)}/\epsilon^2$, and $R$ needs to scale with $1/\epsilon^2$ to guarantee that $N_U^{(\rmC)}$ is almost independent of $\epsilon$. 
When $R$ is sufficiently large, we have
 $N_U^{(\rmC)}\approx N_U^{(\rmT)}\epsilon^2$,  so CRM can significantly reduce the circuit sample cost. In addition, both $N_U^{(\rmC)}$ and $ N_U^{(\rmT)}$ decrease exponentially with ${M_2(\phi)}$. When $R$ is small ($R\lesssim 1/\epsilon^2$), however,  $N_U^{(\rmC)}$ may increase (rather than decrease) with $M_2(\phi)$. These results resemble the counterparts on shadow estimation in the presence of depolarizing noise as discussed in  \sref{SM:VariancesDepol}.

\subsection{Comparison of different measurement ensembles}

In this section, we compare  the efficiencies of three  CRM shadow estimation protocols based on three measurement ensembles, namely, the Clifford measurements, Pauli measurements, and 4-design measurements.

\begin{figure*}[b]
	\centering
	\begin{minipage}[!t]{0.45\linewidth}
		\includegraphics[width=1\columnwidth]{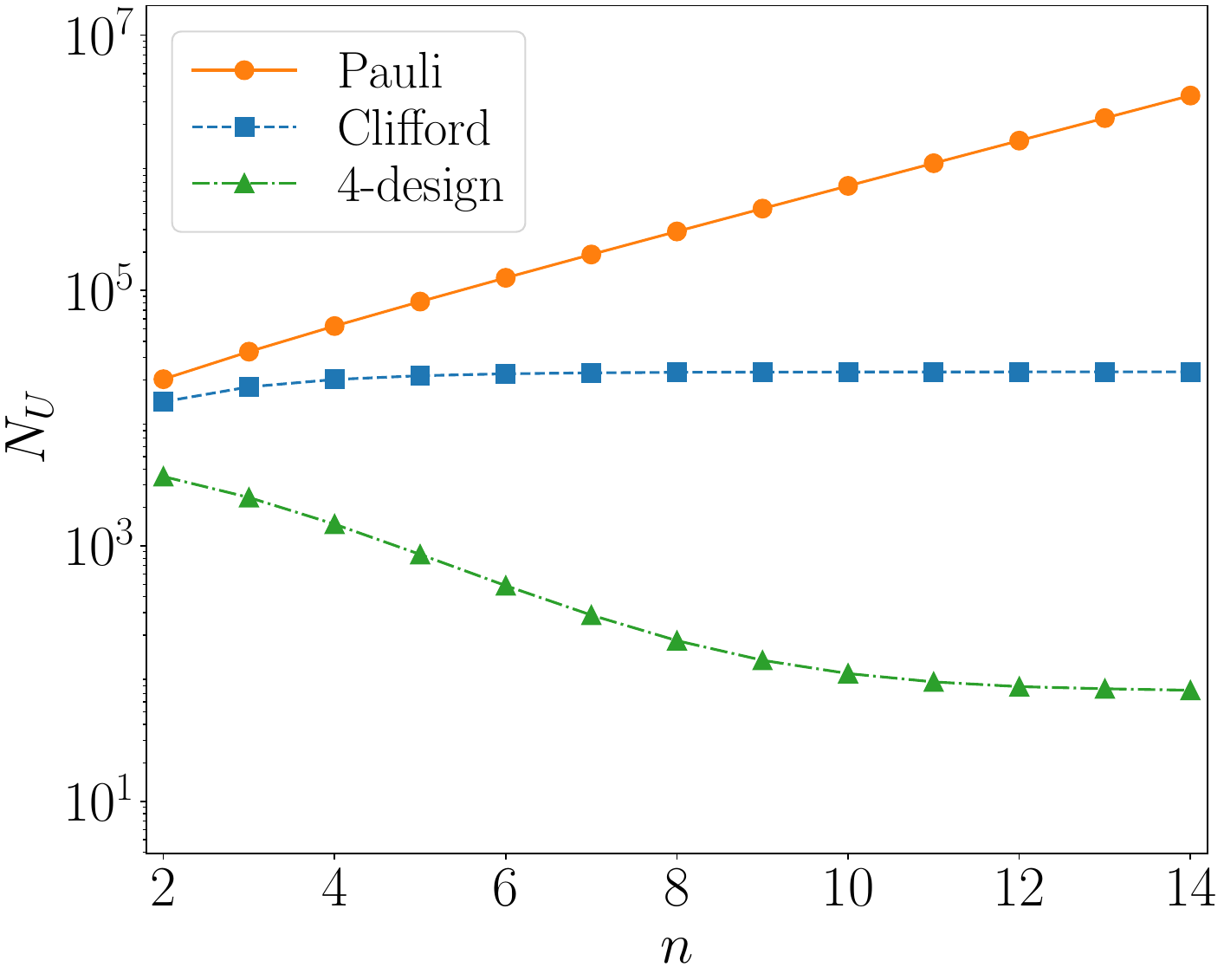}
		\caption{The  circuit sample costs  $N_U$ required for HPFE in CRM shadow estimation based on three measurement ensembles. Here, $r=0.25$, $\delta=0.01$, $\epsilon=0.01$, and $R=20000$. The target and prior state  is the $n$-qubit GHZ state  $\sigma=|\GHZ_n\>\<\GHZ_n|$ with $n=2,3,\ldots,14$, and the system state $\rho$ has the form $\rho = 0.99\sigma + 0.01Z_1 \sigma Z_1$, where $Z_1$ denotes the Pauli $Z$ operator acting on the first qubit.  } \label{fig:three_ensemble_fid}
	\end{minipage}
	\hspace{0.06\linewidth} 
	\begin{minipage}[!t]{0.45\linewidth}\includegraphics[width=1\columnwidth]{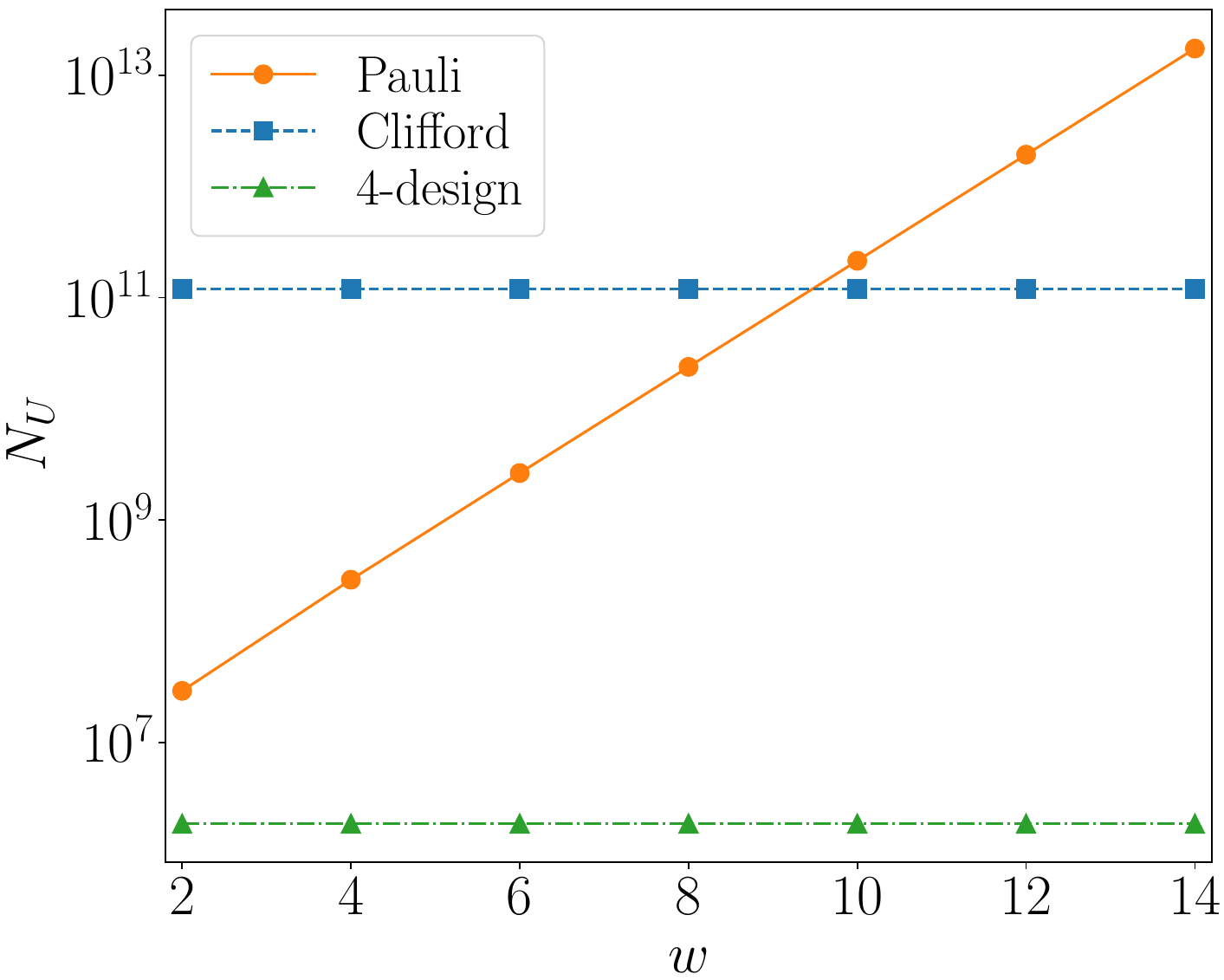}\caption{The circuit sample costs required to estimate the Pauli operator $O_i$ in \eref{eq:PauliEstimated}
			using CRM shadow estimation based on three  measurement ensembles, plotted as functions of the weight $w$. Here, $\varepsilon=0.01$, $\delta=0.01$,  and $R=100$. The prior state has the form $\sigma=|\GHZ_{14}\>\<\GHZ_{14}|$, and the system state reads $\rho =  \left(\sigma + X_1 \sigma X_1\right)/2$, where $X_1$ denotes the Pauli $X$ operator acting on the first qubit. 
		}\label{fig:three_ensemble_pauli}\end{minipage}
\end{figure*}

\subsubsection{HPFE of the GHZ state and estimation of Pauli operators}

First, we consider HPFE using  CRM shadow estimation based on three measurement ensembles. Here the target state is the $n$-qubit GHZ state, and the system state is affected by phase-flip noise on the first qubit. \Fref{fig:three_ensemble_fid} illustrates the circuit sample costs associated with three measurement ensembles as functions of~$n$. 
The sample cost $N_U$ increases exponentially with $n$ for Pauli measurements, is  almost independent of $n$ for Clifford measurements, and decreases exponentially with $n$ for 4-design measurements. These results are consistent with the behaviors of $V_*(O,\Delta)$ as clarified in \thssref{thm:4designVstar}{thm:CliffordVstar}{thm:PauliCRMVar}.

Next, we consider the task of estimating Pauli observables. According to \Lref{lem:VRCRM} and \thref{thm:PauliCRMVar}, for a Pauli operator $O$ with weight (locality) $w$, we have \cite{vermersch2024enhanced}
\begin{equation}\label{eq:PauliCRMPauliOp}
	\bbV_R(\hO)=\left(3^{w}-1\right)[\tr(\Delta O)]^2+\frac{3^w\left\{1-[\tr(\rho O)]^2\right\}}{R}\le 3^w\left\{[\tr(\Delta O)]^2+\frac{1}{R}\right\},
\end{equation}
which indicates that the variance tends to scale with $3^{w}$. In CRM shadow estimation based on Clifford measurements, by contrast, the variance in  estimating the Pauli operator $O$ scales
only with the system size:
\begin{equation}
	\bbV_R(\hO)=d[\tr(\Delta O)]^2+\frac{(d+1)\left\{1-\left[\tr(\rho O)\right]^2\right\}}{R},\label{eq:CliffordCRMPauliOp}
\end{equation}
which follows from \eqssref{eq:generalvar0}{eq:Vcirc}{eq:Clifford_variance}. Therefore, the relative efficiency of Pauli measurements versus Clifford measurements is determined by the competition between the operator weight and the system size. For CRM shadow estimation based on  4-design measurements, the variance scales with neither the weight nor the system size in the large-$R$ limit according to \lref{lem:upperbound_samp2} and \thref{thm:4designVstar}.

In the example we consider, the prior state $\sigma$ is the 14-qubit GHZ state, and the system state has the form  $\rho =  \left(\sigma + X_1 \sigma X_1\right)/2$, where $X_1$ denotes the Pauli $X$ operator acting on the first qubit. In addition, 
the  Pauli operators to be estimated have the form
\begin{equation}
    O_i=\bigotimes_{j=1}^{2i} Z_j,\quad i=1,2,\ldots,7,  \label{eq:PauliEstimated}
\end{equation}    
which have weights $w=2,4,\ldots,14$, respectively. \Fref{fig:three_ensemble_pauli} shows the circuit sample costs required to estimate the expectation values of these Pauli operators in CRM shadow estimation based on three measurement ensembles. Not surprisingly, 4-design measurements require the smallest circuit sample costs and are thus the most efficient; Clifford measurements are better (worse) than Pauli measurements when $w\geq 9$ ($w<9$).

\begin{figure}[t]
	\centering
	\includegraphics[width=0.6\columnwidth]{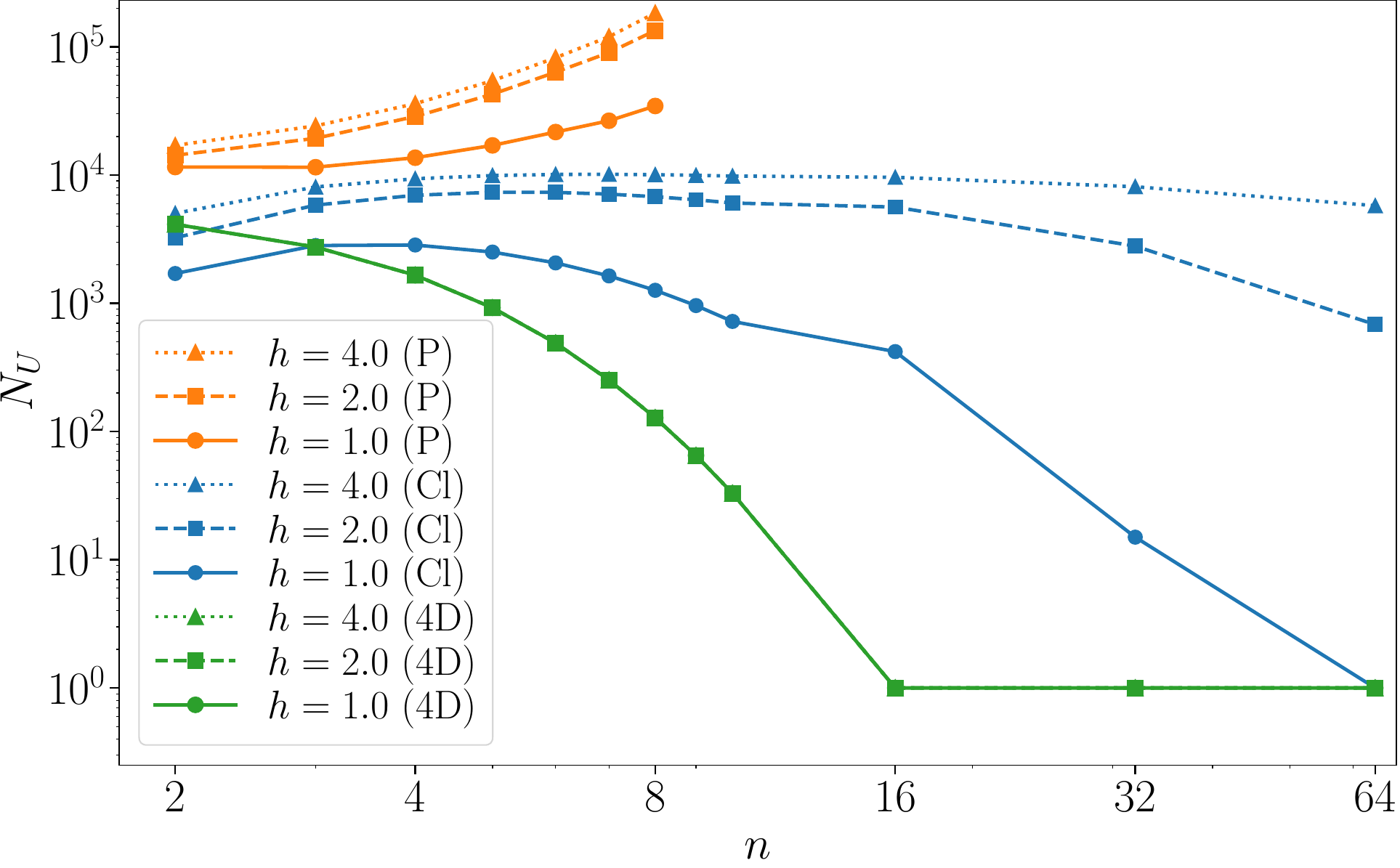}
	\caption{The  circuit sample costs  $N_U$ required for HPFE in CRM shadow estimation based on Pauli (P), Clifford (Cl), and 4-design (4D) measurements. Here $r = 0.25 $, $ \delta = 0.01 $, $\epsilon=0.01$, and $R = 100000$. For given values of $n$ and $h$, the target and prior state  is the ground state of the TFIM Hamiltonian  in \eref{eq:def_TFIM_H} with $J=1$, and the system state $\rho$ has the form $\rho = (1-p)\sigma+p\bbone/d$ with $p=\epsilon/(1-d^{-1})$. 
		The results on 4-design measurements are almost independent of $h$.}\label{fig:tfim_naxis}
\end{figure}

\begin{figure}
	\centering
	\includegraphics[width=0.6\columnwidth]{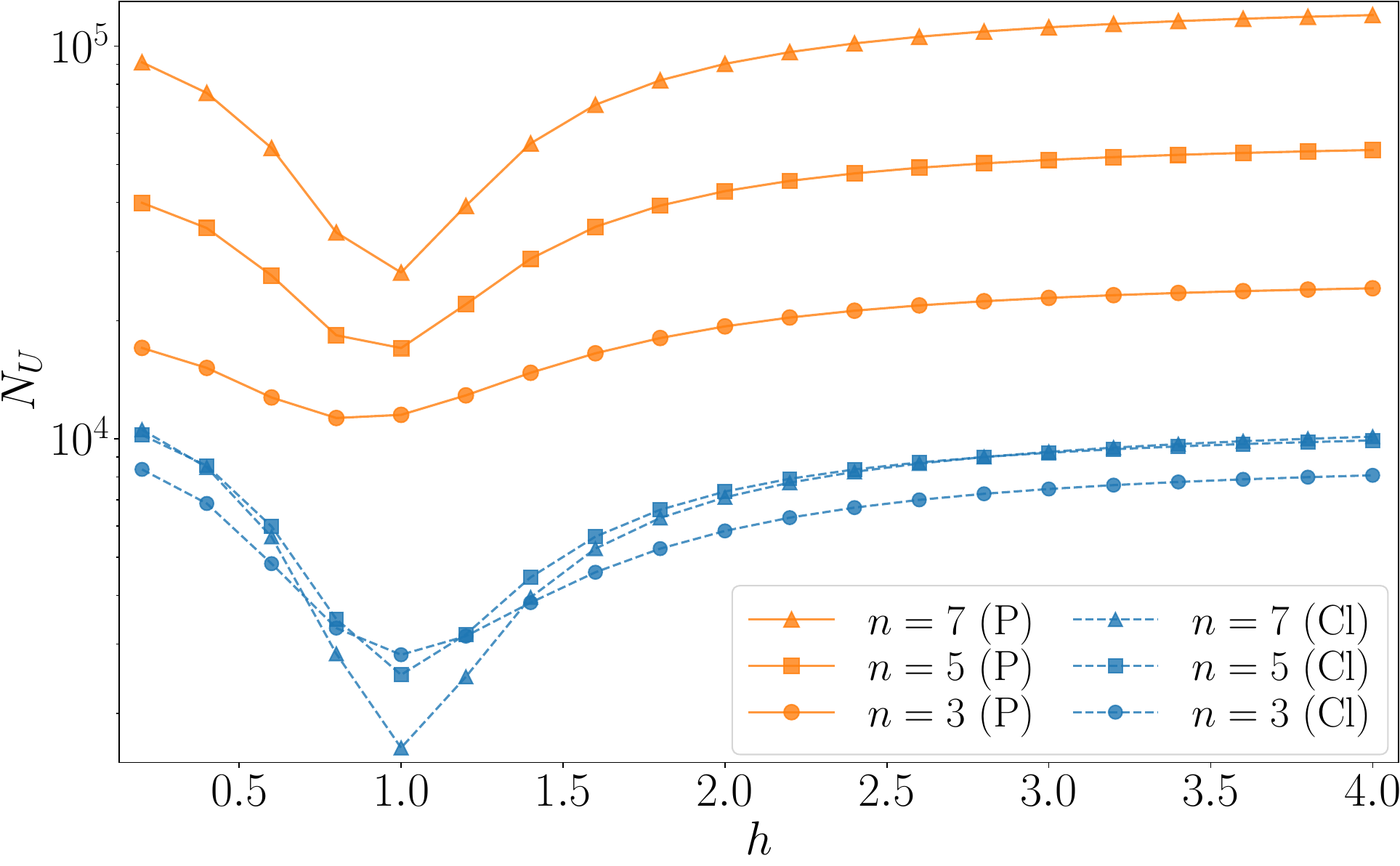}
	\caption{The  circuit sample costs  $N_U$ required for HPFE in CRM shadow estimation based on Pauli (P) and Clifford (Cl) measurements. Here $r = 0.25 $, $ \delta = 0.01 $, $\epsilon=0.01$, and $R = 100000$.  For given values of $n$ and $h$, the target and prior state  is the ground state of the TFIM Hamiltonian in \eref{eq:def_TFIM_H} with $J=1$, and the  system state $\rho$ has the form $\rho = (1-p)\sigma+p\bbone/d$ with $p=\epsilon/(1-d^{-1})$. }\label{fig:tfim_transition}
\end{figure}

\subsubsection{HPFE of the ground state of the transverse-field Ising model}

Here we consider HPFE of the ground state of the transverse-field Ising model (TFIM), which is a typical model widely used in condensed matter physics, quantum information, and quantum computation \cite{Pfeuty1970,sachdev2007quantum,Oliviero2022Magic,Tarabunga2023Many,Tarabunga2024Non,Tarabunga2024Non,Ding2025Eval}. For a one-dimensional system with a periodic boundary condition, the Hamiltonian of the model reads
\begin{equation}\label{eq:def_TFIM_H}
    H=-J\sum_{\<i,j\>}Z_iZ_j-h\sum_{i}X_i.
\end{equation}
Suppose $J=1$; as $h$ increases from $0$, the system undergoes a phase transition at $h=1$, changing from an ordered phase to a disordered phase and exhibits rich features at the critical point $h=1$. In the numerical calculation, the target state is the ground state of the above Hamiltonian, which can be obtained by exact diagonalization, and the system state is affected by  depolarizing noise.

\Fsref{fig:tfim_naxis} and \ref{fig:tfim_transition} illustrate the circuit sample costs associated with  three measurement ensembles as functions of the system size $n$ and the effective magnetic field $h$. 
The circuit sample cost required for HPFE increases exponentially with $n$ for CRM shadow estimation based on Pauli measurements, 
while it decreases exponentially with $n$ for 4-design measurements, 
resembling the behaviors observed in HPFE of the GHZ state as shown in \fref{fig:three_ensemble_fid}. For Clifford measurements, the circuit sample cost  also decreases exponentially with $n$, especially when $h$ is close to the critical point $h=1$. This behavior, in contrast with the constant scaling observed in \fref{fig:three_ensemble_fid}, is tied to the increase in the 2-SRE of the ground state \cite{Oliviero2022Magic,Tarabunga2023Many,Tarabunga2024Non,Ding2025Eval} (an open dataset on the 2-SRE is available at \rcite{Ding2025Eval}), which is expected according to  \eref{eq:hpfe_depo} (see also \lref{lem:depo_terms}). Interestingly, for Pauli measurements, the circuit sample cost also exhibits a similar dependence on $h$.

\appendix

\end{document}